\documentclass{article}

\usepackage{amsthm}
\usepackage{amssymb}
\usepackage{amsmath}
\usepackage{graphicx}
\usepackage{url}
\usepackage{paralist}
\usepackage{color}

\newcommand{\Rec}{\mathit{Rec}}
\newcommand{\cof}{\mathit{Cof}}
\newcommand{\card}{\mathit{Card}}
\newcommand{\res}{\upharpoonright}

\newcommand{\Stab}{\mathit{Stab}}

\newcommand{\tree}{\mathit{tree}}
\newcommand{\ipath}{\mathit{path}}
\newcommand{\notpath}{\mathit{notpath}}
\newcommand{\control}{\mathit{control}}
\newcommand{\samelength}{\mathit{samelength}}
\newcommand{\diff}{\mathit{diff}}
\newcommand{\shorter}{\mathit{shorter}}
\newcommand{\length}{\mathit{length}}
\newcommand{\notincluded}{\mathit{notincluded}}
\newcommand{\noteq}{\mathit{noteq}}

\newcommand{\term}{\mathit{term}}

\newcommand{\notfirstO}{\mathit{notfirst0}}
\newcommand{\firstO}{\mathit{first0}}
\newcommand{\notthirdO}{\mathit{notthird0}}
\newcommand{\thirdO}{\mathit{third0}}
\newcommand{\num}{\mathit{num}}
\newcommand{\lengthonetwo}{\mathit{length12}}
\newcommand{\Agreeonetwo}{\mathit{agree12}}
\newcommand{\lengththree}{\mathit{length3}}
\newcommand{\pz}{$\Pi^0_1$\ }
\newcommand{\lar}{\leftarrow}

\newcommand{\la}{\langle}
\newcommand{\ra}{\rangle}
\newcommand{\sg}{\sigma}
\newcommand{\afs}{a.a. FS}
\newcommand{\FS}{\mathit{FS}}
\newcommand{\rfs}{\mathit{rec. FS}}
\newcommand{\arfs}{\mathit{a.a. rec. FS}}

\newtheorem{proposition}{Proposition}[section]
\newtheorem{theorem}{Theorem}[section]
\newtheorem{lemma}{Lemma}[section]
\newtheorem{corollary}{Corollary}[section]
\newtheorem{definition}{Definition}[section]
\newtheorem{example}{Example}[section]

\newcommand{\PS}{\mathbb{S}}

\begin{document}

\author{D. Cenzer\\
Department of Mathematics\\ University of Florida\\
Gainesville, FL 32611\thanks{Corresponding author. Email:cenzer@math.ufl.edu}\\
V.W. Marek \\
Department of Computer Science\\ University of Kentucky,\\
Lexington, KY 40506 \thanks{Email: marek@cs.uky.edu}\\
J.B. Remmel\\
Department of Mathematics\\ University of California at San Diego, \\
La Jolla, CA 92903 \thanks{Email: jremmel@ucsd.edu}
}

\title{Index sets for Finite Normal Predicate Logic Programs}

\maketitle

\begin{abstract}
Let $\mathcal{L}$ be a computable first order predicate language with
infinitely many constant symbols  and infinitely many $n$-ary predicate symbols
and $n$-ary functions symbols for all $n \geq 1$ and let 
$Q_0,Q_1,\ldots $  be an effective list
all the finite normal predicate logic programs  over $\cal
L$.   Given some property ${\cal P}$ of finite normal predicate logic programs
over $\cal L$, we define the index set $I_{\cal P}$ to be the set of indices
$e$ such that $Q_e$ has property $\cal P$.  Let $T_0,T_1, \ldots $ be an
effective  list of all primitive recursive trees contained in $\omega^{<
\omega}$.  Then $[T_0],[T_1], \ldots $ is an effective list of all 
$\Pi^0_1$ classes where for any tree $T \subseteq \omega^{< \omega}$, $[T]$
denotes the set of infinite paths through $T$.  We modify constructions of
Marek, Nerode,
and Remmel \cite{MNR} to construct recursive functions $f$ and $g$ such that
for all $e$, (i) there is a one-to-one degree preserving correspondence between
the set of stable models of $Q_e$ and the set of infinite paths through
$T_{f(e)}$ and (ii) there is a one-to-one degree preserving correspondence
between the set of infinite paths through $T_e$ and the set of stable models
$Q_{g(e)}$.  We shall use these two recursive functions  to reduce the problem
of finding the complexity of the index set $I_{\mathcal{P}}$ for 
various properties $\cal P$ of normal finite predicate logic programs 
to the problem of computing index sets for  primitive recursive trees  for 
which there is a large variety of results
\cite{JS71,JS72a,JS72b,jlr91,CSW86,clo85}. 

For example, we use our correspondences to determine the complexity of the
index sets relative to all finite predicate logic programs and relative to
certain special classes of finite predicate logic programs of properties such
as (i) having no stable models, (ii) having at least one stable model, (iii)
having exactly $c$ stable models for any given positive integer $c$, (iv)
having only finitely many stable models, or (vi) having infinitely many stable
models.
\end{abstract}

\section{Introduction} \label{intro}

Past research demonstrated that logic programming with 
the stable model semantics and, more generally, with the answer-set
semantics, is an expressive
knowledge representation formalism. The availability of the
non-classical negation operator $\neg$ allows the user to model
incomplete information, frame axioms, and default assumptions. Modeling
these concepts in classical propositional logic is less direct and
requires much larger representations.
In this paper, we investigate the complexity of index 
sets of various properties of
finite normal  predicate logic programs  associated with the stable model 
semantics as defined by Gelfond and Lifschitz \cite{gl88}.  There 
are several other semantics of logic programs that have been 
studied in the literature such as the 
well-founded semantics \cite{VG89} and other 3-valued semantics 
\cite{P91}.  An algebraic analysis of well-founded semantics in terms of
four-valued logic and the four-valued van Emden-Kowalski operator has been done
in \cite{dmt03}, see also \cite{dmt10}.

It is generally accepted that the stable models semantics is the 
correct semantics for logic programs. In particular a number of 
implementations 
of the stable semantics of logic programs (usually known as {\em Answer Set
Programming}) are now available \cite{ns97,dlvtocl06,clasp07}. 
These implementations
are, basically, limited to finite propositional programs or finite predicate
programs not admitting function symbols.   In addition, the
well-founded semantics of fragments of first-order logic extended by inductive
definitions has been implemented as well \cite{dt06,et02}.

The main goal of this paper is to study the 
the complexity of various properties  
finite predicate logic programs with respect 
to the stable model semantics.  
To be able to precisely state our results, we must briefly review 
the basic concepts of normal logic programs. 
We shall fix a recursive language $\mathcal{L}$ 
which  has infinitely many constant symbols $c_0,c_1, \ldots $, 
infinitely many variables $x_0,x_1, \ldots $, 
infinitely many propositional letters $A_0, A_1, \ldots $, and for 
each $n \geq 1$, infinitely many $n$-ary relation symbols 
$R^n_0, R^n_1, \ldots $ 
and $n$-ary function symbols $f^n_0, f^n_1, \ldots$.   We note here that we
shall generally use the terminology \emph{recursive} rather than the equivalent
term \emph{computable} and likewise use \emph{recursively enumerable} rather
than \emph{computably enumerable}.  
These terms have the same meaning, but the former are standard 
in the logic programming community which is an important audience for our paper. 

A literal is an atomic formula or its negation. A ground literal is a literal
which has no free variables.  The Herbrand base of $\mathcal{L}$ is the set
$H_{\cal L}$ of all ground atoms (atomic statements) of the language.  

A (normal) logic programming clause $C$ is of the form 
\begin{equation}\label{clause1}
c \leftarrow a_1, \ldots, a_n, \neg b_1, \ldots, \neg b_m
\end{equation} 
where $c,a_1, \ldots, a_n, b_1, \ldots, b_m$ are atoms of $\mathcal{L}$. 
Here we
allow either $n$ or $m$ to be zero.  In such a situation, 
we call $c$ the {\em conclusion} of $C$, $a_1, \ldots, a_n$ the 
{\em premises} of 
$C$, $b_1, \ldots, b_n$ the {\em constraints} of $C$ and 
$a_1, \ldots, a_n, \neg b_1, \ldots, \neg b_m$ the {\em body} of 
$C$ and write  $concl(C) =c$, $prem(C) = \{a_1, \ldots, a_n\}$,  
$constr(C) = \{b_1, \ldots, b_m\}$. A ground clause is 
a clause with no free variables. $C$ is called a Horn clause  
if $constr(C) = \emptyset$, i.e., if $C$ has no negated 
atoms in its body.

A finite normal predicate logic program $P$ is a finite set 
of clauses of the form (\ref{clause1}).  $P$ is said to 
be a Horn program if all its clauses are Horn clauses. 
A ground instance of a clause $C$ is a clause obtained by
substituting ground terms (terms without free variables) for all
the free variables in $C$.  The set of all ground instances of the
program $P$ is called {\it ground}$(P)$.  
The Herbrand base of $P$, $H(P)$, is the set of 
all ground atoms that are instances of atoms that appear 
in $P$.  For any set $S$, we let $2^S$ denote the set of all subsets of 
$S$. 

Given a Horn program $P$, we let $T_P:2^{H(P)} \rightarrow 
2^{H(P)}$ denote the usual one-step provability operator \cite{ll89} 
associated with $ground(P)$. That is, for $S \subseteq H(P)$, 
\[
T_P(S) = \{c: \exists_{C \in \mathit{ground}(P)}( (C = c \leftarrow a_1, \ldots,
a_n ) \land ( a_1, \ldots, a_n \in S))\}.
\]
Then $P$ has a least model Herbrand  $M = T_P\uparrow_\omega(\emptyset) = 
\bigcup_{n\geq 0} T_P^n(\emptyset)$
where for any $S \subseteq H(P)$, $T^0_P(S) =S$ and 
$T^{n+1}_P(S) = T_P(T^n_P(S))$.  We denote the least model of a Horn program
$P$ by $lm(P)$.

Given a normal  predicate logic program $P$  and $M\subseteq H(P)$, we define
the {\em Gelfond-Lifschitz reduct} of $P$, $P_M$, via the following two step
process. In Step 1, we eliminate all clauses  $C = p\leftarrow q_1,\ldots, q_n,
\neg r_1, \ldots, \neg r_m$ of $ground(P)$ such that there exists an atom $r_i
\in M$. In Step 2, for each remaining clause $C = p\leftarrow q_1,\ldots, q_n,
\neg r_1, \ldots, \neg r_m$ of $ground(P)$, we replace $C$ by the Horn clause
$C = p\leftarrow q_1,\ldots, q_n$.  The resulting program $P_{M}$ is a Horn
propositional program and, hence, has a least model. If that least model of
$P_M$ coincides with $M$, then $M$ is called a {\em stable model} for $P$. 

Next we define the notion of $P$-proof scheme of a normal {\em propositional}
logic program $P$.  Given a normal propositional  logic program $P$, a
$P$-proof scheme is defined by induction on its length $n$. Specifically, the
set of $P$-proof schemes is defined inductively by declaring that 
\begin{compactdesc}
\item[(I)] $\la \la C_1, p_1 \ra ,U\ra$ is a $P$-proof scheme of length 1 
if $C_1 \in P$,  $p_1 = concl(C_1)$, $prem(C_1) = 
\emptyset$, and $U = constr(C_1)$ and  
\item[(II)]  for $n > 1$, 
$\la \la C_1,p_1\ra,\ldots, \la C_n , p_n \ra, U \ra$ is a $P$-proof
scheme of length $n$ if $\la \la C_1,p_1\ra,\ldots, \la C_{n-1} , p_{n-1} \ra,
\bar{U} \ra$ is a $P$-proof scheme of length $n-1$ and 
$C_n$ is a clause in $P$ such that $concl(C_n) = p_n$, $prem(C_n) \subseteq
\{p_1,\ldots,p_{n-1}\}$ and $U = \bar{U} \cup constr(C_n)$
\end{compactdesc}
If  $\PS = \la \la C_1,p_1\ra,\ldots, \la C_n , p_n \ra, U \ra$ is a $P$-proof
scheme of length $n$, then we let $supp(\PS) =U$ and $concl(\PS) = p_n$.

\begin{example}\label{ex1}
{\rm 
Let $P$ be the normal propositional  logic  program consisting of the 
following four clauses:\\
$ C_1 = p \leftarrow$,
$ C_2 = q \leftarrow p,\neg r$,
$ C_3 = r \leftarrow \neg q$, and
$ C_4 = s \leftarrow \neg t$. \\
Then we have the following useful examples of $P$-proof schemes:
\begin{compactenum}
\item[(a)] $\la\la C_1,p\ra,\emptyset\ra$ is a $P$-proof scheme of length
$1$ with conclusion $p$ and empty support.
\item[(b)] $\la\la C_1,p\ra,\la C_2, q\ra ,\{ r\}\ra$ is a $P$-proof scheme of 
length $2$ with conclusion $q$ and support $\{r\}$.
\item[(c)] $\la\la C_1,p\ra,\la C_3, r\ra ,\{ q\}\ra$ is a $P$-proof scheme of 
length $2$ with conclusion $r$  and support $\{q\}$.
\item[(d)] $\la\la C_1,p\ra, \la C_2, q \ra , \la C_3, r\ra ,\{ q,r\}\ra$ 
is a $P$-proof scheme of 
length $3$ with conclusion $r$ and support $\{q,r\}$.
\end{compactenum}
In this example we see that the proof scheme in (c) had an unnecessary
item, the first term, while in (d) the proof scheme was supported by
a set containing $q$, one of atoms that were proved on the way to
$r$. $\hfill\Box$}
\end{example}

A $P$-proof scheme differs from the usual Hilbert-style proofs 
in that it carries within itself its own applicability condition. In
effect, a $P$-proof scheme is a {\em conditional} proof of its
conclusion. It becomes applicable when all the constraints collected in
the support are satisfied. Formally, for a set $M$ of atoms, we say 
that a $P$-proof scheme $\PS$ is $M$-{\em applicable} or 
that $M$ {\em admits} $\PS$ if $M \cap \mathit{supp}(\PS) = \emptyset$. The
fundamental connection between proof schemes and stable models is
given by the following proposition.

\begin{proposition}\label{keystable}
For every normal propositional  logic program $P$ and every set $M$ of
atoms, $M$ is a stable model of $P$ if and only if  
\begin{compactenum}
\item[(i)] for every $p \in M$,  there is a $P$-proof scheme $\PS$ with
conclusion $p$ such that $M$ admits $\PS$ and
\item[(ii)] for every $p \notin M$,  there is no $P$-proof scheme $\PS$ 
with conclusion $p$ such that $M$ admits $\PS$.
\end{compactenum}
\end{proposition}

A $P$-proof scheme may not need all its clauses to prove its
conclusion.  It may be possible to omit some clauses and still
have a proof scheme with the same conclusion.  Thus we 
define a pre-order  on $P$-proof schemes $\PS$, $\mathbb{T}$ 
by declaring that  
$\PS \prec \mathbb{T}$ if

\begin{compactenum}
\item $\PS ,\mathbb{T}$ have the same conclusion,
\item Every clause in $\PS$ is also a clause of $\mathbb{T}$.
\end{compactenum}
The relation $\prec$ is reflexive, transitive, and well-founded. 
Minimal elements of $\prec$ are minimal proof schemes.
A given atom may be the conclusion of no, one, finitely many, or
infinitely many different minimal $P$-proof schemes.  These
differences are clearly computationally significant if one is
searching for a justification of a conclusion.  

If $P$ is a finite normal 
predicate logic program, then we define a $P$-proof scheme 
to be a $ground(P)$-proof scheme. Since we are considering 
finite normal programs over our fixed recursive language $\mathcal{L}$,
we can use standard G\"odel numbering techniques to assign 
code numbers to atomic formulas, clauses, and proof schemes. That is, 
we can effectively assign a natural number to each symbol in 
$\mathcal{L}$.  Strings may be coded by natural numbers 
in the usual fashion. Let $\omega =\{0,1,2,\ldots \}$ denote the set of 
natural numbers and let 
$[x,y]$ denote the standard pairing function $\frac12(x^2+2xy+y^2+3x+y)$ 
and, for $n \geq 2$, we let 
 $[x_0,\ldots,x_n] = [[x_0,\ldots,x_{n-1}],x_n]$. Then 
a string $\sigma =(\sg(0), \ldots, \sg(n-1))$ of length $n$ may be coded by $
c(\sigma) = [n,[\sigma(0),\sigma(1),\ldots,\sigma(n-1)]]$ and also
$c(\emptyset ) = 0$. We define the canonical index 
of any finite set $X= \{x_1 < \cdots < x_n\} \subseteq \omega$ by 
$can(X) = 2^{x_1} + 2^{x_2} + \cdots + 2^{x_n}$. We define 
$can(\emptyset) =0$. 
Then we can think of formulas of $\mathcal{L}$ 
as sequences of natural numbers so that the code of a formula 
is just the code of the sequence of numbers associated with 
the symbols in the formula. 
Then a clause $C$ as in (\ref{clause1}) 
can be assigned the code of the triple $(x,y,z)$ where 
$x$ is the code of the conclusion of $C$, $y$ is the canonical index 
of the set of codes of $prem(C)$, and $z$ is the canonical index 
of the set of codes of $constr(C)$. 
Finally the code of a proof scheme $\PS = \la \la C_1,p_1\ra,\ldots, \la C_n , p_n \ra, U \ra$ consists of the code of a pair $(s,t)$ where 
$s$ is the code of the sequence $(a_1, \ldots, a_n)$ where 
$a_i$ is the code of the pair of codes for $C_i$ and $p_i$ and 
$t$ is the canonical index of the set of codes for elements of $U$. 
It is then not difficult to verify that 
for any given finite normal predicate logic program $P$, the questions 
of whether a given $n$ is the code of a ground atom, a ground 
instance of a clause in $P$, or a $P$-proof 
are  primitive recursive predicates.  The key observation to 
make is that since $P$ is finite and the usual unification 
algorithm is effective, we can explicitly test whether 
a given number $m$ is the code of a ground atom or a ground instance of 
a clause in $P$ without doing any unbounded searches. It is 
then easy to see that, once we can determine if a number 
$m$ is a code of ground instance of a clause of $P$ in a primitive 
recursive fashion, then there is a primitive recursive algorithm 
which determines whether a given number $n$ is the 
code of a minimal $P$-proof scheme.  

If $P$ is a finite normal predicate logic program over $\cal L$, 
we let $N_k(P)$ be the set of all codes of minimal $P$-proof schemes 
$\PS$ such that all the atoms appearing in all the rules used 
in $\PS$ are smaller than $k$. Obviously $N_k(P)$ is finite. 
Since the predicate ``minimal $P$-proof scheme'', which holds only on codes 
of minimal $P$-proof schemes, is a primitive recursive predicate, it 
easily follows that we can uniformly construct 
a primitive recursive function $h_P$ such 
that $h_P(k)$ equals the canonical index for $N_k(P)$.

A finite normal predicate logic program $Q$ over $\cal L$ may be written out as a finite
string over a finite alphabet and thus may be assigned a G\"{o}del number
$e(Q)$ in the usual fashion. The set of G\"{o}del numbers of well-formed
programs is well-known to be primitive recursive (see Lloyd
\cite{ll89}). Thus we may let $Q_e$ be the program with G\"{o}del number
$e$ when this exists and let $Q_e$ be the empty program otherwise. 
For any property $\mathcal{P}$ of finite normal predicate logic programs, let 
$I(\mathcal{P})$ be the set of indices $e$ such that $Q_e$ has property
$\mathcal{P}$. 

Next we define the notions of decidable normal logic programs   and of 
normal logic programs which have the finite support property.  Proposition
\ref{keystable} says that the presence and absence of the atom $p$ in a stable
model of a finite normal predicate logic program $P$ depends {\em only} on the
supports of its $ground(P)$-proof schemes.  This fact naturally leads to a
characterization of stable models in terms of propositional satisfiability.
Given $p \in H(P)$, the {\em defining equation} for $p$ with respect to  $P$ is
the following propositional formula:
\begin{equation}\label{defeq}
p \Leftrightarrow (\neg U_1 \lor \neg U_2 \lor \ldots )
\end{equation}
where $\la U_1,U_2,\ldots \ra$ is the list of all supports of minimal
$ground(P)$-proof schemes. Here for any finite set $S = \{s_1, \ldots, s_n\}$
of atoms, $\neg S = \neg s_1 \wedge \cdots \wedge \neg s_n$. If $U =
\emptyset$, then $\neg U = \top$.  Up to a total ordering of the finite sets of
atoms such a formula is unique. For example, suppose we fix a total order on
$H(P)$, $p_1 < p_2 < \ldots $. Then given two sets of atoms, $U = \{u_1 <
\ldots < u_m\}$ and $V = \{v_1 < \ldots < v_n\}$, we say that $U \prec V$, if
either (i) $u_m < v_n$, (ii) $u_m = v_n$ and $m < n$, or (iii) $u_m = v_n$, $n=
m$, and $(u_1, \ldots, u_n)$ is lexicographically less than $(v_1, \ldots,
v_n)$.  We also define $\emptyset \prec U$ for any finite nonempty set $U$.  We
say that (\ref{defeq}) is the {\it defining equation} for $p$ relative to $P$
if $U_1 \prec U_2 \prec \ldots $. We will denote  the defining equation for $p$
with respect to $P$ by $\mathit{Eq}_p^P$.  When $P$ is a Horn program, an atom
$p$ may have an empty support or no support at all. The first of these
alternatives occurs when $p$ belongs to the least model of $P$,
$\mathit{lm}(P)$. The second alternative occurs when $p \notin \mathit{lm}(P)$.
The defining equations are $p \Leftrightarrow \top$ when $p \in \mathit{lm}(P)$
and $p \Leftrightarrow \bot$  when $p \notin \mathit{lm}(P)$. 

Let $\Phi_P$ be the set $\{\mathit{Eq}_p^P : p \in H(P)\}$. We then have the
following consequence of Proposition \ref{keystable}.
\begin{proposition}\label{p.eq}
Let $P$ be a normal propositional logic program.  Then the stable models of $P$ are
precisely the propositional models of the theory $\Phi_P$.
\end{proposition}
When $P$ is {\em purely negative}, i.e. all clauses $C$ of $P$ have $prem(C) =
\emptyset$, the stable and supported models of $P$ coincide  \cite{dk89} and
the defining equations reduce to Clark's completion \cite{cl78} of $P$.

Let us observe that, in general, the propositional formulas on the
right-hand-side of the defining equations may be infinitary.

\begin{example}\label{e.2}
{\rm Let $P$ be an infinite normal propositional logic program consisting of
clauses $p \lar \neg p_i$, for all $i \in n$. Then the defining equation for
$p$ in $P$ is the infinitary propositional formula 
\[
p \Leftrightarrow (\neg p_1 \lor \neg p_2 \lor \neg p_3 \ldots ).
\]
}
\mbox{ }\hfill $\Box$
\end{example}
The following observation is quite useful. If $U_1, U_2$ are two finite
sets of propositional atoms, then
\[
U_1 \subseteq U_2 \ \mbox{if and only if\ }
\ \neg U_2 \models \neg U_1
\]
Here $\models$ is the propositional consequence relation. The effect
of this observation is that only the inclusion-minimal supports are important.

\begin{example}\label{ex.3}
{\rm Let $P$ be an infinite normal propositional logic program consisting of
clauses $p \lar \neg p_1, \ldots ,\neg p_i$, for all $i \in N$. The defining
equation for $p$ in $P$ is 
\[
p \Leftrightarrow [\neg p_1 \lor (\neg p_1 \land \neg p_2)
 \lor (\neg p_1 \land \neg p_2 \land \neg p_3) \ldots\ ]
\]
which is infinitary.  But our observation above implies that this formula is
{\em equivalent} to the formula
\[
p \Leftrightarrow \neg p_1. 
\]
}
\mbox{ }\hfill $\Box$
\end{example}

Motivated by the Example \ref{ex.3}, we define the {\em reduced defining
equation} for $p$ relative to $P$ to be the formula
\begin{equation}\label{reddefeq}
p \Leftrightarrow (\neg U_1 \lor \neg U_2 \lor \ldots )
\end{equation}
where $U_i$ range over {\em inclusion-minimal} supports of minimal $P$-proof
schemes for the atom $p$ and $U_1 \prec U_2 \prec \cdots$.  We denote this
formula as $\mathit{rEq}_p^P$, and define $r\Phi_P$ to be the theory consisting
of $\mathit{rEq}_p^P$ for all $p \in H(P)$. We then have the following
strengthening of Proposition \ref{p.eq}.

\begin{proposition}\label{p.req}
Let $P$ be a normal propositional program.
Then stable models of $P$ are precisely the propositional models of the
theory $r\Phi_P$.
\end{proposition}
In our example \ref{ex.3}, the theory $\Phi_P$ was infinitary, but the
theory $r\Phi_P$ was finitary.

Suppose that $P$ is a  normal propositional  logic program $P$ which consists 
of ground clauses from $\mathcal{L}$ and $a$ is an atom in $H(P)$.  Then we say
that $a$ has the {\em finite support property relative of $P$} if the reduced
defining equation for $a$ is finite.  We say that $P$ has the   {\it finite
support ($\FS$) property}   if for all $a \in H(P)$, the reduced defining
equation for $a$ is a finite propositional formula. Equivalently, a program $P$
has the finite support property  if for every atom $a \in H(P)$, there are only
finitely many inclusion-minimal supports of minimal $P$-proof schemes for $a$.
We say that $P$ has the {\it almost always finite support  ($\afs$) property}
if for all but finitely many atoms $a \in H(P)$, there are only finitely many
inclusion-minimal supports of minimal $P$-proof schemes for $a$.  We say that
$P$ is {\em recursive} if the set of codes of clauses of $P$ is recursive and
the set of codes of atoms in $H(P)$ is recursive.  Note that for any finite
normal predicate logic program $Q$, $ground(Q)$ will automatically be a
recursive normal propositional  logic program.  We say that $P$ has the {\it
recursive finite support ($\rfs$) property} if $P$ is recursive, has the finite
support property, and there is a uniform effective procedure which given any
atom $a \in H(P)$ produces the code of the set of the inclusion-minimal
supports of $P$-proof schemes for $a$.  We say that $P$ has the {\it almost
always recursive finite support ($\afs$) property} if $P$ is recursive, has the
$\afs$ property, and there is a uniform effective procedure which for all but a
finite set of atoms $a \in H(P)$ produces the code of the set of the
inclusion-minimal supports of $P$-proof schemes for $a$.  We say that a finite
normal predicate logic program has the $\FS$ property ($\rfs$ property, $\afs$
property, $\arfs$ property) if $ground(P)$ has the $\FS$ property ($\rfs$
property, $\afs$ property, $\arfs$ property).

Next we define two additional properties of recursive normal propositional
logic programs that have not been previously defined in the literature.
Suppose that $P$ is a recursive normal propositional  logic program consisting
of ground clauses in $\mathcal{L}$ and $M$ is a stable model of $P$.  Then for
any atom $p \in M$, we say that a minimal $P$-proof scheme $\PS$ is the {\em
smallest minimal $P$-proof for $p$ relative to $M$} if $concl(\PS)=p$ and
$supp(\PS) \cap M = \emptyset$ and there is no minimal $P$-proof scheme $\PS'$
such that $concl(\PS')=p$ and $supp(\PS') \cap M = \emptyset$ and the G\"odel
number  of $\PS'$ is less than the G\"odel number of $\PS$.  We say that $P$ is
{\em decidable} if for any finite set of ground atoms $\{a_1, \ldots, a_n\}
\subseteq H(P)$ and any finite set of minimal $P$-proof schemes $\{\PS_1,
\ldots, \PS_n\}$ such that $concl(\PS_i) =a_i$, we can effectively decide
whether there is a stable model of $M$ of $P$ such that \\
(a) $a_i \in M$ and $\PS_i$ is the smallest minimal $P$-proof scheme for $a_i$
such that $supp(\PS_i) \cap M = \emptyset$ and \\
(b) for any ground atom $b \notin  \{a_1, \ldots, a_n\}$ such 
that the code of $b$ is strictly less than the maximum of the 
codes of $a_1, \ldots, a_n$, $b \notin M$. 
 
We now introduce and illustrate a technical concept that will be useful for our
later considerations.   At first glance, there are some obvious differences
between stable models of normal propositional  logic programs and models of
sets of sentences in a propositional logic. For example, if $T$ is a set of
sentences in a propositional logic and $S \subseteq T$, then it is certainly
the case that every model of $T$ is a model of $\PS$. Thus a set of
propositional sentences $T$ has the property that if $T$ has a model, then
every subset of $T$ has a model.  This is certainly not true for normal
propositional  logic programs. That is, consider the following example.

\begin{example}
{\rm Let $P$ consists of the following two clauses:
\begin{quote}
$C_1 = a \leftarrow \neg a, \neg b$ and \\
$C_2 = b \leftarrow$
\end{quote}

Then it is easy to see that $\{b\}$ is a stable model of $P$. However the
subprogram $Q$ consisting of just clause $C_1$ does not have a stable model.
That is, $b$ can not be in any stable model of $Q$ since there is no clause in
$Q$ whose conclusion is $b$.  Thus the only possible stable models of $Q$ are
$M_1 = \emptyset$ and $M_2 = \{a\}$. But it is easy to see that both $M_1$ and
$M_2$ are not stable models of $Q$. That is, the Gelfond-Lifschitz reduct
$Q_{\emptyset} = a \leftarrow$ whose least model is $\{a\}$ and the
Gelfond-Lifschitz reduct $Q_{\{a\}} = \emptyset $ whose least model is
$\emptyset$.}
\end{example}

Next we note that there is no analogue of the Compactness Theorem for stable
models.  That is, the Compactness Theorem for propositional logic says that if
$\Theta$ is a collection of sentences and every finite subset of $\Theta$ has a
model, then $\Theta$ has a model. Marek and Remmel \cite{MRComp} proposed the
following analogue of the Compactness Theorem for normal propositional  logic
programs. \\
\ \\
({\em Comp}) {\em If for any finite normal propositional logic program
$P^\prime \subseteq P$, there exist a finite program $P^{\prime \prime}$ such
that $P^\prime \subseteq P^{\prime \prime} \subseteq P$ such that $P^{\prime
\prime}$ has a stable model, then $P$ has a stable model.}\\
\ \\
However, Marek and Remmel \cite{MRComp} showed that {\em Comp} fails for normal
propositional  logic programs. 

Finally, we observe that a normal propositional  logic program $P$ can fail to
have a stable model for some trivial reasons. That is, suppose that $P_0$ is a
normal propositional  logic program which has a stable model and $a$ is atom
which is not in the Herbrand base of $P_0$, $H(P_0)$. Then if $P$ is the normal
propositional  logic program consisting of $P_0$ plus the clause $C= a
\leftarrow \neg a$, then $P$ automatically does not have a stable model. That
is, consider a potential stable model $M$ of $P$. If $a \in M$, then $C$ does
not contribute to $P_M$ so that there will be no clause of $P_M$ with $a$ in 
the head. Hence, $a$ is not in the least model of $P_M$ so that $M$ is not a
stable model of $P$.  On the other hand, if $a \not \in M$, then $C$ will
contribute the clause $a \leftarrow$ to $P_M$ so that $a$ must be in the least
model of $P_M$ and, again, $M$ is not equal to the least model of $P_M$. For
this reason, we say that a finite normal predicate logic program $Q_e$ over
$\cal L$ has an {\em explicit initial blocking set} if there is an $m$ such
that 
\begin{compactenum}
\item for every $i \leq m$, either $i$ is not the code of an atom of
$ground(P)$ or the atom $a$ coded by $i$ has the finite 
support property relative to $P$ and 
\item for all $S \subseteq \{0,\ldots, m\}$, either 
\begin{compactdesc}
\item[(a)] there exists an $i \in S$ such that $i$ is not 
the code of an atom in $H(P)$,
\item[(b)] there is an $i \not \in S$ such that 
there exists a minimal $P$-proof scheme $p$ such that $concl(p) =a$ where 
$a$ is the atom of $H(P)$ with code $i$ and $supp(p) \subseteq 
\{0,\ldots, m\} -S$, or 
\item[(c)] there is an $i \in S$ such that every minimal $P$-proof scheme 
$\PS$ of the atom $a$ of $H(P)$ with code $i$ has 
$supp(\PS) \cap S \neq \emptyset$.  
\end{compactdesc}
\end{compactenum}
The definition of a finite normal predicate logic program $Q_e$ over $\cal L$
having an {\it initial blocking set} is the same as the definition of $Q_e$
having an explicit initial blocking set except that we drop the condition that 
for every $i \leq m$ which is the code of an atom $a \in H(P)$, $a$ must have
the finite support property relative to $P$.

If $\Sigma \subseteq \omega$, then $\Sigma^{< \omega}$ denotes the set of
finite strings of letters from $\Sigma$ and $\Sigma^\omega$ denotes the set of
infinite sequences of letters from $\Sigma$. For a string $\sigma =
(\sigma(0),\sigma(1),\ldots,\sigma(n-1))$, we let $|\sigma|$ denote the length
$n$ of $\sigma$.  The empty string has length 0 and will be denoted by
$\emptyset$.  A constant string $\sigma$ of length $n$ consisting entirely of
$k$'s  will be denoted by $k^n$.  For $m < |\sigma|$, $\sigma \res m$ is the
string $(\sigma(0),\ldots,\sigma(m-1))$. We say $\sigma$ is an {\em initial
segment} of $\tau$ (written $\sigma \prec \tau$) if $\sigma = \tau\res m$ for
some $m < |\sg|$.  The concatenation  $\sigma^\smallfrown \tau$ (or sometimes
just $\sigma\tau$) is defined by
\[
\sigma^\smallfrown \tau = (\sigma(0),\sigma(1),\ldots,\sigma(m-
1),\tau(0),\tau(1),\ldots,\tau(n-1)) 
\]
where $|\sigma| = m$ and $|\tau| = n$. We write $\sigma^\smallfrown a$ for
$\sigma^\smallfrown (a)$ and $a^\smallfrown \sigma$ for $(a)^\smallfrown
\sigma$.  For any $x \in \Sigma^{\omega}$ and any finite $n$, the
{\it initial segment} $x\res n$ of $x$ is $(x(0),\ldots , x(n-1))$. We write
$\sigma \prec x$ if $\sigma = x\res n$ for some $n$. For any $\sigma \in
\Sigma^n$ and any $x \in \Sigma^{\omega}$, we let  $\sigma^\smallfrown x =
(\sigma(0),\ldots,\sigma(n-1),x(0),x(1),\ldots)$. 

If $\Sigma \subseteq \omega$, a {\it tree T} over $\Sigma^\ast$ is a set of
finite strings from $\Sigma^{< \omega}$ which contains the empty string
$\emptyset$ and which is closed under initial segments. We say that $\tau \in
T$ is an {\it immediate successor} of a string $\sigma \in T$ if $\tau =
\sigma^\smallfrown a$ for some $a \in \Sigma$. We will identify $T$ with the
set of codes $c(\sigma)$ for $\sigma \in T$. Thus we say that $T$ is recursive,
r.e., etc. if $\{c(\sg): \sg \in T\}$ is recursive, r.e., etc.  If each node of
$T$ has finitely many immediate successors, then $T$ is said to be {\it
finitely branching}.  

\begin{definition}
{\rm Suppose that $g:\omega^{< \omega} \rightarrow \omega$.  Then we say that 
\begin{compactenum}
\item $T$ is {\em $g$-bounded} if for all $\sigma$ and all integers $i$,
$\sigma^\smallfrown i \in T$ implies $i \leq g(\sigma)$,  
\item $T$ is {\em almost always  $g$-bounded} if there is a finite set $F
\subseteq T$ of strings such that for all strings $\sigma \in T \setminus F$
and all integers $i$, $\sigma^\smallfrown i \in T$ implies $i < g(\sigma)$,  
\item $T$ is {\em nearly  $g$-bounded} if there is an $n \geq 0$ such that 
for all strings $\sigma \in T$ with $|\sg| \geq n$ and all integers $i$,
$\sigma^\smallfrown i \in T$ implies $i < g(\sigma)$,  
\item $T$ is {\em bounded}  if it is $g$-bounded for some $g:\omega^{< \omega}
\rightarrow \omega$,
 \item $T$ is {\em almost always bounded}  ($a.a.b.$) if it is  almost 
always $g$-bounded for some $g:\omega^{< \omega} \rightarrow \omega$,
\item $T$ is {\em nearly bounded}  if it is nearly $g$-bounded for some
$g:\omega^{< \omega} \rightarrow \omega$, 
\item $T$ is {\em recursively  bounded} ($r.b.$)  if $T$ is $g$-bounded 
for some recursive  $g:\omega^{< \omega} \rightarrow \omega$, 
\item $T$ {\em almost recursively bounded} ($a.a.r.b$.)  
if it is almost always $g$-bounded for some recursive  $g:\omega^{< \omega}
\rightarrow \omega$, and \item $T$ {\em nearly recursively bounded} ( nearly
$r.b$.)  if it is nearly $g$-bounded for some recursive  $g:\omega^{< \omega}
\rightarrow \omega$.
\end{compactenum}
}
\end{definition}

For any tree $T$, an {\it infinite path} through $T$ is a sequence
$(x(0),x(1),\ldots)$ such that $x \res n \in T$ for all $n$. Let $[T]$ be the
set of infinite paths through $T$. We let $Ext(T)$ denote the set of all
$\sigma \in T$ such that $\sigma \prec x$ for some $x \in [T]$. Thus $Ext(T)$
is the set of all $\sigma$ in $T$ that lie on some infinite path through 
$T$.  We say that $T$ is {\em decidable} if $T$ is recursive and $Ext(T)$ is
recursive. 

The two main results of this paper are the following theorems. 

\begin{theorem}\label{tree2prog} There is a uniform effective procedure  which
given any recursive tree $T \subseteq \omega^{< \omega}$ produces a finite
normal predicate logic program $P_T$ such that the following hold. 
\begin{compactenum} 
\item There is an effective one-to-one degree preserving correspondence 
between the set of stable models of $P_T$ and the set of infinite paths through
$T$.  
\item $T$ is bounded if and only if $P_T$ has the $\FS$ property. 
\item $T$ is recursively bounded  if and only if 
$P_T$ has the $\rfs$ property. 
\item $T$ is decidable and recursively bounded if and only if 
$P_T$ is decidable and has the $\rfs$ property. 
\end{compactenum} 
\end{theorem}

\begin{theorem} \label{prog2trees}
There is a uniform recursive procedure which given any finite normal predicate
logic program $P$ produces a primitive recursive tree $T_P$ such that the
following hold.
\begin{compactenum}
\item There is an effective one-to-one degree-preserving correspondence 
between the set of stable models of $P$ and the set of infinite 
paths through $T_P$.
\item $P$ has the $\FS$ property or $P$ has an explicit initial blocking set if
and only if $T_P$ is bounded.
\item If $P$ has a stable model, then $P$ has the $\FS$ property if and only if
$T_P$ is bounded.
\item $P$ has the $\rfs$  property or an explicit initial blocking set if and
only if $T_P$ is recursively bounded.
\item If $P$ has a stable model, then $P$ has the $\rfs$ property if and only
if $T_P$ is recursively bounded.
\item $P$ has the $\afs$ property or $P$ has an explicit initial blocking set
if and only if $T_P$ is nearly bounded.
\item If $P$ has a stable model, then $P$ has the $\afs$ property if and only
if $T_P$ is nearly bounded.
\item $P$ has the $\arfs$  property or an explicit initial blocking set if and
only if $T_P$ is nearly recursively bounded.
\item If $P$ has a stable model, then $P$ has the $\arfs$ property if and only
if $T_P$ is nearly recursively bounded.
\item If $P$ has a stable model, then $P$ is decidable if and only if $T_P$ is
decidable. 
\end{compactenum} 
 \end{theorem}

The idea of Theorems \ref{tree2prog} and \ref{prog2trees} is to show that index
sets for certain properties of trees have the same complexity as corresponding
index sets for various properties of finite normal predicate logic programs.
For example, suppose that we want to find the complexity of 
\[
A =\{e:Q_e \ \mbox{has the $\FS$ property and has exactly 2 stable models}\}.
\]
Let $\displaystyle B =\{e:T_e \ \mbox{is $r.b.$ and $\card([T_e]) =2$}\}$.
Then Theorem \ref{tree2prog} allows us to prove that $B$ is one-to-one
reducible to $A$ and Theorem \ref{prog2trees} allows us to prove that $A$ is
one-to-one reducible to $B$.  Now Cenzer and Remmel \cite{CR1,CR2} have proved
a large number of results about the index sets for primitive recursive trees. 
In particular, they have shown that $B$ is $\Sigma^0_3$-complete. Thus $A$ is
also $\Sigma^0_3$-complete. \\
\ \\
The outline of this paper is as follows. In Section \ref{classes}, we shall
provide the basic background on $\Pi^0_1$ classes and 
recursive trees that we shall need. In Section \ref{proofs}, we shall 
give the proofs of Theorems \ref{tree2prog} and \ref{prog2trees}. In Section 
\ref{compl}, we shall use Theorems \ref{tree2prog} and \ref{prog2trees} to prove
a variety of index set results relative to all finite normal predicate logic
programs, to all finite normal predicate logic programs which have the $\FS$
property, and to all finite normal predicate logic programs which have the
$\rfs$ property. In Section \ref{aa}, we shall prove a variety of index set
results relative to all finite normal predicate logic programs  which have the 
 $a.a.$$\FS$ property and to all finite normal predicate logic programs which 
have the $\arfs$ property. Section \ref{concl} contains conclusions and
suggestions of further work.

A preliminary extended abstract of this paper \cite{CMR} appeared in the
proceedings of a workshop at the Federated Logic Conference FLOC'99 which were 
distributed at the conference.

\section{$\Pi^0_1$ classes and trees}\label{classes}

In this section, we shall review the basic background facts on 
the complexity of various properties of $\Pi^0_1$ classes and 
primitive recursive trees that are relevant to classifying 
the index sets of the 
properties of finite normal predicate logic 
programs that will be of interest to us.

Let $\phi_e$ denote the partial recursive 
function which is computed by the $e$-th Turing machine. Thus 
$\phi_0, \phi_1, \ldots $ is a list of all partial recursive functions.
We let $W_e$ be the set of all $x \in \omega$ such $\phi_e(x)$ converges. 
Thus $W_0, W_1, \ldots $ is a list of all recursively enumerable (r.e.) 
sets. More generally, a  
recursive functional $\phi$ takes as inputs both numbers $a \in
\omega$ and functions $x: \omega \to \omega$. The function inputs are
treated as ``oracles'' to be called on when needed. Thus a particular
computation $\phi(a_1,\ldots,a_n;x_1,\ldots,x_m)$ only uses a finite
amount of information $x_i \res c$ about each function $x_i$.
Thus we shall write $\phi_e(a_1,\ldots,a_n;x_1,\ldots,x_m)$ for 
the recursive functional computed by the $e$-th oracle machine. 
In the special case where $n =m =1$ and $x_1$ is a sequence of 
0s and 1s and $X = \{n:x_1(n) = 1\}$, then we shall 
write $\phi_e^X(a_1)$ or $\{e\}^X(a_1)$ instead of $\phi_e(a_1;x_1)$. 
The jump of a set $A \subseteq \omega$, denoted $A^\prime$, 
is the set of all $e$ such that $\phi_e^A(e)$ converges. We 
let $0^\prime$ denote the jump of the empty set. 
For $A,B \subseteq \omega$, we write $A \leq_T B$ if $A$ is 
Turing reducible to $B$ and $A \equiv_T B$ if $A \leq_T B$ and 
$B \leq_T A$.

We shall assume the reader is familiar with the usual 
arithmetic hierarchy of $\Sigma^0_n$ and $\Pi^0_n$ subsets 
of $\omega$ as well as $\Sigma^1_1$ and $\Pi^1_1$ sets, 
see Soare's book \cite{Soare} for any unexplained notation. 
A subset $A$ of $\omega$ is said to be 
$D^m_n$ if it is the difference of two $\Sigma^m_n$ sets.
A set $A \subseteq \omega$ is said to be an \emph{index set} if for any
$a,b$, $a \in A$ and $\phi_a = \phi_b$ imply that $b \in A$. For
example, $\mathit{Fin} = \{a: W_a \mbox{ is finite}\}$ is an index set. 
We are particularly interested in the complexity of such index sets.
Recall that a subset $A$ of $\omega$ is said to be
$\Sigma^m_n$-complete (respectively, $\Pi^m_n$-complete, $D^m_n$-complete ) if
$A$ is $\Sigma^m_n$ (respectively, $\Pi^m_n$, $D^m_n$) and any $\Sigma^m_n$
(respectively, $\Pi^m_n$, $D^m_n$) set $B$ is many-one reducible to $A$. For
example, the set $Fin = \{e: W_e {\rm\ is\ finite}\}$ is $\Sigma^0_2$-complete.

A recursive tree $T$ is said to be
 {\it highly recursive} if $T$ is finitely branching and 
there is a partial recursive function $f$
such that, for any $\sigma \in T$, $f(\sigma)$ is the canonical 
index of the set of codes of all 
immediate successors in $T$. It is easy to show 
that $T$ is highly recursive if and only if $T$ is recursive and 
recursively bounded. 

A set $\cal C$ of functions $f:N \rightarrow N$ is a $\Pi^0_1$-{\em class}
if and only if
 \[
f \in \mathcal{C} \Leftrightarrow \forall n ([ f(0),\ldots,f(n)] \in R)
\]
where $R$ is some 
recursive predicate.
It is well known that 
$\mathcal{C}$ is a $\Pi^0_1$-{\it class} if and only if $X = [T]$ for some 
recursive tree $T$. In fact, the following lemma is true. 

\begin{lemma} \label{lem:eqv} For any class ${\cal C} \subseteq \omega^\omega$,
 the following are equivalent.
\begin{compactenum}
\item ${\cal C} = [T]$ for some recursive tree $T \subseteq \omega^{<\omega}$.
\item ${\cal C} = [T]$ for some primitive recursive tree $T$.
\item ${\cal C} = \{x:\omega \rightarrow \omega: (\forall n) R(n,[x\res
n])\}$, for some recursive relation $R$.  
\item ${\cal C} = [T]$ for some tree $T \subseteq \omega^{<\omega}$ which is
$\Pi^0_1$. 
\end{compactenum} 
\end{lemma}

We say that a $\Pi^0_1$ class $\mathcal{C}$ is 
\begin{compactenum}
\item {\em bounded} if $\mathcal{C} = [T]$ for some recursive 
tree $T$ which is bounded, 
\item {\em almost always bounded} ($a.a.b.$) 
 if $\mathcal{C} = [T]$ for some recursive 
tree $T$ which is almost always bounded, 
\item {\em nearly bounded} ($n.b.$)
 if $\mathcal{C} = [T]$ for some recursive tree $T$ which is nearly bounded, 
\item {\em recursively bounded} ($r.b.$) if $\mathcal{C} = [T]$ for some 
highly recursive tree $T$,  
\item {\em almost always recursively bounded} ($a.a.r.b.$) if $\mathcal{C} =
[T]$ for some recursive tree $T$ which is almost always recursively bounded, 
\item {\em nearly recursively bounded} ($n.r.b.$) if 
$\mathcal{C} = [T]$ for some 
recursive tree $T$ which is nearly recursively bounded, and 
\item {\em decidable}  if $\mathcal{C} = [T]$ for some decidable tree $T$.
\end{compactenum}

We now spell out the indexing for \pz classes and primitive recursive trees
that we will use in this paper.  Let $\pi_0,\pi_1,\ldots$ be an effective
enumeration of the primitive recursive functions from $\omega$ to $\{0,1\}$ and
let
\[
T_e = \{\emptyset\} \cup \{\sigma: (\forall \tau \preceq \sigma) 
\pi_e(c(\tau)) = 1\}
\]
where $c(\tau)$ is the code of $\tau$.  It is clear that each $T_e$ is a
primitive recursive tree. Observe also that if $\{\sigma: \pi_e(c(\sigma)) =
1\}$ is a primitive recursive tree, then $T_e$ will be that tree.  Thus every
primitive recursive tree occurs in our enumeration $T_0, T_1, \ldots$.  (Note
that, henceforth, we will generally identify a finite sequence $\tau \in
\omega^{<\omega}$ with its code.) Then we let ${\cal C}_e = [T_e]$ be the
$e$-th $\Pi^0_1$ class. It follows from Lemma \ref{lem:eqv} that every
$\Pi^0_1$ class occurs in the enumeration ${\cal C}_e$.
 
There is a large literature on the complexity of elements in $\Pi^0_1$ classes
and index sets for primitive recursive trees. In the remainder of this section,
we shall list the key results which will be needed for our applications to
index sets associated with finite normal predicate logic programs.  

\begin{theorem}\label{thm:Ext} For any recursive tree $T\subseteq
\omega^{< \omega}$, the following hold.
\begin{compactdesc}
\item[(a)] $Ext(T)$ is a $\Sigma^1_1$ set.
\item[(b)] If $T$ is finitely branching, then $Ext(T)$ is a $\Pi^0_2$ set.
\item[(c)] If $T$ is highly recursive, then $Ext(T)$ is a $\Pi^0_1$ set.
\end{compactdesc}
\end{theorem}

For any nonempty \pz class ${\cal C} = [T]$, one can compute a member of ${\cal
C}$ from the tree $Ext(T)$  by always taking the leftmost branch in $Ext(T)$.

The following theorem immediately follows from Theorem \ref{thm:Ext}.

\begin{theorem}\label{thm:basis1} For any nonempty $\Pi^0_1$ class 
${\cal C} \subseteq \omega^{< \omega}$, 
\begin{compactdesc}
\item[(a)] ${\cal C}$ has a member which is recursive in some $\Sigma^1_1$ set. 
\item[(b)] If ${\cal C}$ is bounded, nearly bounded,  or 
almost always bounded, then ${\cal C}$ has a member which is 
recursive in ${\bf 0}''$, 
\item[(c)] If ${\cal C}$ is recursively bounded, nearly recursively bounded, 
 or almost always 
recursively bounded, then ${\cal C}$ has a member which is recursive  in 
${\bf 0}'$, and 
\item[(d)] If ${\cal C} = [T]$, where $T$ is decidable, then ${\cal C}$ has a
recursive  member.
\end{compactdesc} 
\end{theorem}

If $T \subseteq \omega^{< \omega}$ is tree and $f \in [T]$, then we say that
$f$ is isolated, if there is $k > 0$ such that $f$ is the only element of $[T]$
which extends  $(f(0), \ldots, f(k))$.  The complexity of isolated paths in
recursive trees was determined by Kreisel. 

\begin{theorem}{\em [Kreisel 59]}  Let ${\cal C}$ be a $\Pi^0_1$ class.
\begin{compactdesc}
\item[(a)] Any isolated member of ${\cal C}$ is hyperarithmetic. 
\item[(b)] Suppose that ${\cal C}$ is bounded, nearly bounded, or almost always
bounded. Then any isolated member of ${\cal C}$ is recursive  in ${\bf 0}'$. 
\item[(c)] Suppose ${\cal C}$ is recursively  bounded, nearly recursively
bounded,  or almost always recursively bounded. Then any isolated member of
${\cal C}$ is recursive.
\end{compactdesc}
\end{theorem}

A set $A \subseteq \omega$ is {\em low} if $A^\prime = 0^\prime$.
Jockusch and Soare \cite{JS71,JS72a,JS72b} proved the following important 
results about recursively bounded $\Pi^0_1$ classes. 

\begin{theorem}   
\begin{compactenum}
\item[(a)]{\em (Low Basis Theorem)} 
Every nonempty $r.b.$ $\Pi^0_1$ class ${\cal C}$ contains a member
of low degree. 

\item [(b)] There is a low degree {\bf a} such that every nonempty
r.b.  $\Pi^0_1$ class contains a member of degree $\leq {\bf a}$.

\item[(c)] If ${\cal C}$ is  $r.b.$, then $P$ contains a member of r.e.  degree.

\item[(d)] Every $r.b.$ \pz class ${\cal C}$ contains members $a$ and $b$ such
that any function recursive in both $a$ and $b$ is recursive. 

\item[(e)] If ${\cal C}$ is s bounded \pz class, then ${\cal C}$ contains a
member of $\Sigma^0_2$ degree.

\item[(f)] Every bounded \pz class contains a member $a$ such that 
${a}' \leq_T {\bf 0}''$. 

\item[(g)] Every bounded \pz class ${\cal C}$ contains members $a$ and $b$ such
that any function recursive in both $a$ and $b$ is recursive in
$\emptyset^\prime$.
\end{compactenum}
\end{theorem}

Cenzer and Remmel \cite{CR1,CR2} proved a large number 
of results about index sets for $\Pi^0_1$ classes and primitive recursive 
trees. Below 
we list a sample of such results which will be important 
for us to establish corresponding results for index 
sets of finite normal predicate logic programs. 

Our first results establish the complexity of determining 
whether a primitive recursive tree is recursively bounded, 
almost always recursively bounded, nearly recursively bounded, 
bounded, almost always bounded, nearly bounded, or decidable. 

\begin{theorem} \label{thm:pirb}
\begin{compactenum}
\item[(a)] $\{e: T_e\ \text{is $r.b.$}\}$ is $\Sigma^0_3$-complete.
\item[(b)] $\{e: T_e\ \text{is $a.a.r.b$.}\}$ is $\Sigma^0_3$-complete.
\item[(c)] $\{e: T_e\ \text{is $n.r.b$.}\}$ is $\Sigma^0_3$-complete.
\item[(d)] $\{e: T_e\ \text{is bounded}\}$ is $\Pi^0_3$-complete.
\item[(e)] $\{e: T_e\ \text{is $a.a.b.$}\}$ is $\Sigma^0_4$-complete.
\item[(f)] $\{e: T_e\ \text{is $n.b.$}\}$ is $\Sigma^0_4$-complete.
\item [(g)]
$\{e: T_e\ \text{is $r.b.$ and decidable}\}$ 
is $\Sigma^0_3$-complete.
\end{compactenum}
\end{theorem}
\begin{proof}
The only parts which are not proved by Cenzer and Remmel  
in \cite{CR1} are parts (b) and  (e). (In \cite{CR1}, Cenzer 
and Remmel used the term almost bounded for what we call  
nearly bounded.)

We shall show how to modify the proofs of (c) and (f) in \cite{CR1} to prove
(b) and (e), respectively.  Similar modifications of the proofs in \cite{CR1}
for index sets relative to nearly bounded and nearly recursively bounded trees
can be used to establish  the remaining index set results which we list in this
section. 

The facts that $\{e: T_e\ \text{is $a.a.r.b$.}\}$ is $\Sigma^0_3$  and $\{e:
T_e\ \text{is $a.a.b$}\}$ is $\Sigma^0_4$ are easily established by simply
writing out the definitions. 

To prove the $\Sigma^0_3$-completeness of $\{e: T_e\ \text{is $a.a.r.b$.}\}$,
we can use the same proof that was used by Cenzer and Remmel \cite{CR1} to
establish that $\{e: T_e\ \text{is $r.b$.}\}$ is $\Sigma^0_3$-complete.  It is
easy to see that a tree $T$ is $r.b.$ if and only if there is a recursive
function $g:\omega \rightarrow \omega$ such that if $(a_0, \ldots, a_n) \in T$,
then $a_i < g(i)$ for all $i \in T$. Similarly, a tree $T$ is $a.a.r.b.$ if and
only if there is a recursive function $g:\omega \rightarrow \omega$ such 
that for all but finitely many $(a_0, \ldots, a_n) \in T$, $a_i < g(i)$ for all
$i \in T$. In each case, we shall call such a function $g$ a {\em bounding
function}. 

Now, $\Rec= \{e: W_e \ \text{is recursive}\}$ is $\Sigma^0_3$-complete, 
see Soare's book \cite{Soare}. We define a reduction $f$ of $\Rec$ to  $\{e:
T_e\ \text{is $r.b$.}\}$.  This will be done so that $[T_{f(e)}]$ is empty if
$W_e$ is finite and $[T_{f(e)}]$ has a single element if $W_e$ is infinite.
The primitive recursive tree $T_{f(e)}$ is defined so that we put $\sigma =
(s_0,s_1,\ldots,s_{k-1}) \in T_{f(e)}$ if and only if $s_0 < s_1 < \dots <
s_{k-1}$ and there exists a sequence $m_0 < m_1 < \dots < m_{k-1}$ such that,
for each $i<k$, $m_i \in W_{e,s_i} \setminus W_{e,s_i-1}$ and $m_i$ is the
least element of $W_{e,s_{k-1}} \setminus \{m_0,\ldots,m_{i-1}\}$.  We observe
that if $W_e$ is finite, then $T_{f(e)}$ is also finite and therefore
recursively bounded.  Now fix $e$ and suppose that $W_e$ is infinite. Then we
define a canonical sequence $n_0 < n_1 < \dots$ of elements of $W_e$ and
corresponding sequence of stages $t_0 < t_1 < \dots$ such that, for each $i$,
$n_i \in W_{e,t_i} \setminus W_{e,t_i-1}$ and $(t_0,t_1,\ldots,t_i) \in
T_{f(e)}$ as follows. Let $n_0$ be the least element of $W_e$ and $t_0$ is the
least stage $t$ such that $n_0 \in W_{e,t}$. Then for each $k$, let $n_{k+1}$
be the least element of $W_e \setminus W_{e,t_k}$ and $t_{k+1}$ be the least
stage $t$ such that $n_{k+1} \in W_{e,t}$.  Then for each $k$,
$(t_0,\ldots,t_k) \in T_{f(e)}$ and $n_k \in W_{e,t_k}$. Furthermore, we can
prove by induction on $k$ that 
\[
k \in W_e \to k \in W_{e,t_k}.
\]

For $k=0$, this is because $n_0 = 0$ if $0 \in W_e$.  Assuming the statement to
be true for all $i<k$, we see that if $k \in W_e$, then either $k \in
W_{e,t_{k-1}}$, or else $n_k = k$. In either case, we have $k \in W_{e,t_k}$.

The key fact to observe is that for any $(s_0,\ldots,s_k) \in T_{f(e)}$, $s_k
\leq t_k$. To see this, let $(s_0,\ldots,s_k) \in T_{f(e)}$, let
$(m_0,\ldots,m_k)$ be the associated sequence of elements of $W_e$.  Suppose by
way of contradiction that $s_k > t_k$. It follows from the definitions of
$T_{f(e)}$ and of $t_0,\ldots,t_k$ that in fact $s_i = t_i$ and $m_i = n_i$ for
all $i \leq k$. Thus if we let $g(n) = t_{n} +1$, then $g$ will be a bounding
function for $T_{f(e)}$.  Now, if $W_e$ is recursive. then the sequence $t_0 <
t_1 < \dots$ is also recursive and thus $T_{f(e)}$ is recursively bounded.

Now suppose that $T_{f(e)}$ has a recursive bounding function $h$.  Then we
must have $t_k <  h(k)$ for each $\sigma$ of length $k$. It then follows from
the equation above that $k \in W_e \iff k \in W_{e,h(k)}$, so that $W_e$ is
recursive. Thus $T_{f(e)}$ is $r.b$ if and only if $W_e$ recursive and, hence,
$\{e:T_e \ \text{is $r.b.$} \}$ is $\Sigma^0_3$-complete.  However, note that
if $h:\omega^{< \omega} \rightarrow \omega$ is a function that witnesses that
$T_{f(e)}$ is almost always recursively bounded, then there will be a $n$ such
that $t_k < h(k)$ for all $k \geq n$.  In that case, for all $k \geq n$,  $k
\in W_e \iff k \in W_{e,h(k)}$ which still implies that $W_e$ is recursive.
Thus $T_{f(e)}$ is $a.a.r.b.$ if and only if $W_e$ is recursive so 
that $\{e: T_e\ \text{is $a.a.r.b$.}\}$ is also $\Sigma^0_3$-complete. 

This argument is typical of the completeness arguments for the properties about
cardinalities of $[T]$ or the number of recursive elements of $[T]$ that appear
in the rest of the theorems in this section. That is, the completeness 
argument for $r.b.$ trees also works for $a.a.r.b.$ trees. 

For the completeness argument for (d), we shall use the fact that $\cof =
\{e:\omega\setminus W_e  \ \text{is finite}\}$ is $\Sigma^0_3$-complete set,
see \cite{Soare}. We let $W_{e,s}$ denote the set of elements that are
enumerated into $W_e$  in $s$ or fewer steps as in \cite{Soare}. By definition, 
all $x \in W_{e,s}$ are less than or equal to $s$ and the question of whether
$x \in W_{e,s}$ is a primitive recursive predicate.  Then we can define a
primitive recursive function $\phi(e,m,s) = (least \ n > m)(n \notin
W_{e,s} \setminus \{0\})$. For any given $e$, let $U_e$ be the tree such that 
$(m) \in U_e$ for all $m \geq 0$ and $(m,s+1) \in U_e$ if and only if $m$ is
the least element such that $\phi(e,m,s+1) > \phi(e,m,s)$.  Note that when $m
\geq s+1$, the least $n$ such that $n > m$ and $n \notin W_{e,s}$ is just $m+1$
since all elements of $W_{e,s+1}$ are less than $s+1$.  Thus the only 
candidates for $(m,s+1)$ to be in $U_e$ are $m \leq s+1$.  Thus the tree $U_e$
will be primitive recursive.  Now if $W_e\setminus \{0\}$ is not cofinite, then
for each $m$, there is a minimal $n>m$ such that $n \notin W_e$. It follows that
$\lim_s \phi(e,m,s) = n$, so that $\phi(e,m,s+1) > \phi(e,m,s)$ for only
finitely many $s$, which will make $U_{e}$ finitely branching.  On the other
hand, if $W_e\setminus \{0\}$ is cofinite and we choose $m$ so that 
$n \in W_e\setminus \{0\}$ for all $n > m$, then it is clear that there will be 
infinitely many $s$ such that $\phi(e,m,s+1)>\phi(e,m,s)$.  It follows that if
$m$ is the largest element not in $W_e\setminus \{0\}$, then for infinitely
many $s$, $(m,s+1)$ will be in $U_e$ and for all $p > m$, there can be only
finitely many $s$ such that $(p,s+1)$ is in $U_e$. Thus if $W_e\setminus \{0\}$
is cofinite, then there will be exactly one node which has infinitely many
successors.  Clearly there is a recursive function $f$ such that $T_{f(e)} =
U_e$. But then 
\[
e \in \omega \setminus \cof \iff T_{f(e)} \ \text{is
bounded}.
\]
Since $\omega \setminus \cof$ is $\Pi^0_3$-complete, it follows that $\{e: T_e
\ \text{is bounded}\}$ is $\Pi^0_3$-complete.

Now, let  $S$ be an arbitrary $\Sigma^0_4$ set and suppose that $a \in S \iff
(\exists k) R(a,k)$ where $R$ is $\Pi^0_3$. By the usual quantifier methods, we
may assume that $R(a,k)$ implies that $R(a,j)$ for all $j>k$.  By our argument
for the $\Pi^0_3$-completeness of $\{e: T_e \ \text{is bounded}\}$, there is a
recursive function $h$ such that $R(a,k)$ holds if and only if $U_{h(a,k)}$ is
bounded and such that $U_{h(a,k)}$ is $a.a.b.$ for every $a$ and $k$.  Now we
can define a recursive function $\phi$ so that 
\[
T_{\phi(a)} = \{(0)\} \cup \{(k+1)^\smallfrown \sigma: \sigma 
\in U_{h(a,k)}\}.
\]
If $a \in S$, then $U_{h(a,k)}$ is bounded for all but finitely many $k$ and is
$a.a.b.$ for the remaining $k$'s. Thus $U_{\phi(a)}$ is $a.a.b.$  If $a \notin
S$, then, for every $k$, $U_{h(a,k)}$ is not bounded, so that $U_{\phi(a)}$ is
not $a.a.b.$ Thus $a \in S$ if and only if $T_{\phi(a)}$ is $a.a.b$ and $\{e:
T_e\ \text{is $a.a.b.$ }\}$ is $\Sigma^0_4$-complete. 
\end{proof}

As it stands, it is clear that there are  no infinite paths through
$T_{\phi(a)}$ since every node $T_{\phi(a)}$ has length at most 3.  The reason
that we constructed the tree $T_{\phi(a)}$ to contain the node $(0)$ is for the
remaining completeness arguments which follow in this section. That is, we are
now free to modify the construction to add a tree above $(0)$ which has a
number of infinite paths. Now, completeness arguments to establish the
complexity for various properties concerning the number of infinite paths or
infinite recursive paths through $r.b.$ trees in \cite{CR1} always produced
bounded trees. Since the complexity results for $r.b.$ trees were bounded by
$\Sigma^0_4$, it follows that we can modify the construction by placing trees
above $(0)$ in  $T_{\phi(a)}$ to show that  complexity for various properties
concerning the number of infinite paths or infinite recursive paths through
$a.a.b.$  trees is $\Sigma^0_4$-complete. Thus we shall not give the details of
such arguments.  $\hfill\Box$

Next, we give several index set results concerning the size of $[T]$ for
primitive recursive trees $T$ which have various properties. These results are 
either proved in \cite{CR1} or follow by modifying the results in \cite{CR1} as
described in Theorem \ref{thm:pirb} to prove results about $a.a.b.$ or
$a.a.r.b.$ trees. In fact, in all the results that follow, the index set
results for properties relative to $a.a.b.$ trees are exactly the same as the
index set results for $n.b.$ trees and the index set results for properties of
$a.a.r.b.$ trees are exactly the same as the index set results for $n.r.b.$
trees.  Thus we shall only state the results for $a.a.$ and $a.a.r.b.$ trees. 

\begin{theorem} \label{thm:pire} 
\begin{compactenum}
\item[(a)] $\{e: T_e \ \text{is r.b. and $[T_e]$ is empty} \}$ is
$\Sigma^0_2$-complete. 
\item[(b)] $\{e: T_e \ \text{is r.b. and $[T_e]$ is nonempty} \}$ is
$\Sigma^0_3$-complete.
\item[(c)] $\{e: T_e \ \text{is bounded and $[T_e]$ is empty} \}$ is
$\Sigma^0_2$-complete. 
\item[(d)] $\{e: T_e \ \text{is bounded and $[T_e]$ is nonempty} \}$ is
$\Pi^0_3$-complete.
\item[(e)] $\{e: T_e \ \text{is a.a.r.b. and $[T_e]$ is nonempty} \}$ 
and \\
$\{e: T_e \ \text{is a.a.r.b. and $[T_e]$ is}$ $\text{empty} \}$ are 
$\Sigma^0_3$-complete.
\item[(f)] $\{e: T_e \ \text{is a.a.b. and $[T_e]$ is nonempty} \}$ 
and \\
$\{e: T_e \ \text{is a.a.b. and $[T_e]$ is}$ $\text{ empty} \}$ are
$\Sigma^0_4$-complete.  
\item[(g)] $\{e: [T_e] \ \text{is nonempty} \}$ is $\Sigma_1^1$-complete and\\
$\{e: [T_e] \ \text{is empty} \}$ is  $\Pi_1^1$-complete.  
\end{compactenum}
\end{theorem}

\begin{theorem} \label{thm:pirc}
For every positive integer $c$,
\begin{compactenum} 
\item[(a)] $\{e: T_e \ \text{is r.b. and}\ \card([T_e]) > c\}$,\\
$\{e: T_e \ \text{is r.b. and}\ \card([T_e]) \leq c\}$,
and \\
$\{e: T_e \ \text{is r.b. and}\ \card([T_e]) = c\}$ are all
$\Sigma^0_3$-complete.
\item[(b)] $\{e: T_e \ \text{is a.a.r.b. and}\ \card([T_e]) > c\}$,\\
$\{e: T_e \ \text{is a.a.r.b. and}\ \card([T_e]) \leq c\}$,
and \\
$\{e: T_e \ \text{is a.a.r.b. and}\ \card([T_e]) = c\}$ are all 
$\Sigma^0_3$-complete.
\item[(c)]
$\{e: T_e \ \text{is bounded and}\ \card([T_e]) \leq c\}$ and \\
$\{e: T_e \ \text{is bounded}$ and $\card([T_e])$ $= 1\}$
are both  $\Pi^0_3$-complete;
\item[(d)]
$\{e: T_e \ \text{is bounded and}\ \card([T_e]) > c\}$
and \\
$\{e: T_e \ \text{is bounded and}\ \card([T_e])$ $ = c+1\}$
are both $D^0_3$-complete.
\item[(e)] $\{e: T_e \ \text{is a.a.b. and}\ \card([T_e])$ $ > c\}$,\\
$\{e: T_e \ \text{is $a.a.$ bounded and}\ \card([T_e])$ $\leq c\}$,
and \\
$\{e: T_e \ \text{is $a.a.$ bounded and}\ \card([T_e]) = c\}$ 
are all $\Sigma^0_4$-complete.
\item[(f)] $\{e: T_e \ \text{is $r.b$, dec. and}\ \card([T_e]) > c\}$,\\
$\{e: T_e \ \text{is $r.b.$, dec. and}\ \card([T_e])$ $\leq c\}$,
and \\
$\{e: T_e \ \text{is $r.b.$, dec. and}\ \card([T_e]) = c\}$ 
are all $\Sigma^0_3$-complete.
\item[(g)]
$(\{e:\card([T_e]) > c\})$ is $\Sigma^1_1$-complete, $\{e:\card([T_e]) \leq c\}$
is $\Pi^1_1$-complete and $\{e:\card([T_e]) = c\}$ is $\Pi^1_1$-complete.
\end{compactenum} 
\end{theorem}

\begin{theorem} \label{thm:pici}
\begin{compactenum}
\item[(a)] $\{e: T_e\ \text{is r.b. and $[T_e]$ is infinite}\}$ is
$D^0_3$-complete and $\{e: T_e\ \text{is r.b. and $[T_e]$ is finite}\}$ is
$\Sigma^0_3$-complete.
\item[(b)] $\{e: T_e\ \text{is a.a.r.b. and $[T_e]$ is infinite}\}$ is
$D^0_3$-complete and \\
$\{e: T_e\ \text{is a.a.r.b. }\ \text{and $[T_e]$ is
finite}\}$ is $\Sigma^0_3$-complete.
\item[(c)] $\{e: T_e \ \text{is bounded and $[T_e]$ is infinite}\}$ is 
$\Pi^0_4$-complete and \\
$\{e: T_e \ \text{is bounded and $[T_e]$ is finite}\}$ is $\Sigma^0_4$-complete.
\item[(d)] $\{e: T_e\ \text{is $a.a.$bounded and $[T_e]$ is
infinite}\}$ is $D^0_4$-complete and \\
$\{e: T_e\  \text{is $a.a.$ bounded and $[T_e]$ is finite}\}$ is $\Sigma^0_4$-complete.
\item[(e)] $\{e: [T_e]\ \text{is infinite}\}$ is 
$\Sigma_1^1$-complete and $\{e: [T_e]\ \text{is
finite}\}$ is $\Pi_1^1$-complete.
\item[(f)]$\{e: T_e\ \text{is r.b.and dec.  and $[T_e]$ is infinite}\}$ is
$D^0_3$-complete and \\
 $\{e: T_e\ \text{is r.b. and dec.  and $[T_e]$ is finite}\}$ is 
$\Sigma^0_3$-complete.
\end{compactenum}
\end{theorem}

\begin{theorem} \label{thm:picu} 
$\{e: [T_e]\ \text{is uncountable}\}$ is $\Sigma^1_1$-complete, 
$\{e: [T_e]\ \text{is countable}\}$
is $\Pi^1_1$-complete, and 
$\{e: [T_e]\ \text{is countably
infinite}\}$ is $\Pi^1_1$-complete. 
The same result holds 
for $r.b.$, $a.a.r.b.$, bounded, $a.a.b.$  
primitive recursive trees. 
\end{theorem}

Next we give some index set results concerning the number of recursive elements
in $[T]$ where $T$ is a primitive recursive tree. Here we say that $[T]$ is
{\em recursively empty} if $[T]$ has no recursive elements and is  {\em
recursively nonempty} if $[T]$ has at least one recursive element. Similarly,
we say that $[T]$ has  {\em recursive cardinality equal to $c$} if $[T]$ has
exactly $c$ recursive members. 

\begin{theorem} \label{thm:prbce}
\begin{compactenum}
\item[(a)] $\{e: T_e\ \text{is r.b. and $[T_e]$ is 
recursively nonempty}\}$ 
is $\Sigma^0_3$-complete, 
$\{e: T_e\ \text{is r.b. and  $[T_e]$ is recursively empty}\}$ 
is $D^0_3$-complete and 
$\{e: T_e\ \text{is r.b. and  $[T_e]$ is  nonempty and recursively empty}\}$ 
is $D^0_3$-complete.
\item[(b)] $\{e: T_e\ \text{is a.a.r.b. and  $[T_e]$ is recursively
nonempty}\}$ is $\Sigma^0_3$-complete, \\
$\{e: T_e\ \text{is}$ $\text{ a.a.r.b. and  $[T_e]$ is recursively empty}\}$ 
is $D^0_3$-complete and \\
$\{e: T_e\ $ $\text{is a.a.r.b. and  $[T_e]$ is nonempty and  recursively
empty}\}$ is \\
$D^0_3$-complete.
\item[(c)] $\{e: T_e\ \text{is bounded and  
$[T_e]$ is recursively nonempty}\}$ 
is $D^0_3$-complete, \\
$\{e: T_e\ \text{is bounded and  $[T_e]$ is recursively empty}\}$ 
is $\Pi^0_3$-complete, and \\
$\{e: T_e \ \text{is bounded and  
$[T_e]$ is nonempty and recursively empty}\}$ 
is \\
$\Pi^0_3$-com\-ple\-te.
\item[(d)] $\{e: T_e\ \text{is $a.a.$bounded and  $[T_e]$ is recursively
nonempty}\}$,  \\
$\{e: T_e\ \text{is $a.a.$bounded}$
$\text{and  $[T_e]$ is recursively empty}\}$, and \\
$\{e: T_e\ \text{is $a.a.$bounded and}$ $\text{  $[T_e]$ is  nonempty and
recursively empty}\}$ are all $\Sigma^0_4$-complete.
\item[(e)] $\{e: [T_e]\ \text{is recursively nonempty}\}$ 
is $\Sigma^0_3$-complete, \\
$\{e: [T_e]\ \text{is recursively empty}\}$ 
is $\Pi^0_3$-complete and \\
$\{e: [T_e]\ \text{is nonempty and recursively empty}\}$ 
is $\Sigma_1^1$-complete.
\end{compactenum}
\end{theorem}

\begin{theorem} \label{thm:pbcc} Let $c$ be a positive integer.
\begin{compactenum}
\item[(a)] $\{e: T_e\ \text{is r.b. and $[T_e]$ has recursive 
  cardinality} > c\}$ is $\Sigma^0_3$-complete, \\
$\{e: T_e\ \text{is r.b. and  $[T_e]$ has recursive cardinality}
  \leq c\}$ is $D^0_3$-complete, and \\
$\{e: T_e\ \text{is r.b. and  $[T_e]$ has recursive cardinality}
  =c \}$ is $D^0_3$-complete. 
\item[(b)] $\{e: T_e\ \text{is a.a.r.b. and  $[T_e]$ has recursive 
  cardinality} > c\}$ is $\Sigma^0_3$-complete, \\
$\{e: T_e\ \text{is a.a.r.b. and  $[T_e]$ has recursive cardinality}
  \leq c\}$ is $D^0_3$-complete, and 
$\{e: T_e\ \text{is a.a.r.b. and  $[T_e]$ has recursive cardinality}
  =c \}$ is \\
  $D^0_3$-complete. 
\item[(c)] $\{e: T_e\ \text{is bounded and  $[T_e]$ has recursive 
  cardinality} > c\}$ is $\Pi^0_3$-complete, \\
$\{e: T_e\ \text{is bounded and  $[T_e]$ has recursive cardinality}
  \leq c\}$ is $D^0_3$-complete, and 
$\{e: T_e\ \text{is bounded and  $[T_e]$ has recursive cardinality}
  =c \}$ is  \\ $D^0_3$-complete. 
\item[(d)] $\{e: T_e\ \text{is $a.a.$bounded and  $[T_e]$ has recursive 
  cardinality} > c\}$, \\
$\{e: T_e\ \text{is $a.a.$bounded  and  $[T_e]$ has recursive cardinality}
  \leq c\}$, and \\
$\{e: T_e\ \text{is $a.a.$bounded and  $[T_e]$ has recursive cardinality}
  =c \}$ are all $\Sigma^0_4$-complete. 
\item[(e)] $\{e: [T_e]\ \text{has recursive cardinality} >c\}$ is 
$\Sigma^0_3$-complete, \\
$\{e:[T_e]\ \text{has recursive cardinality} \leq c\}$ is
$\Pi^0_3$-complete, and \\
 $\{e: [T_e]\ \text{has recursive cardinality}\ = c\}$ 
is $D^0_3$-complete. 
\end{compactenum}
\end{theorem}

\begin{theorem} $\{e: [T_e]\ \text{has finite recursive
    cardinality}\}$ is $\Sigma^0_4$-complete and \\
$\{e: [T_e]\ \text{has infinite recursive cardinality}\}$ is
$\Pi^0_4$-complete.  The same result is true for $r.b.$, $a.a.r.b.$, bounded,
and $a.a.b.$  primitive recursive trees. 
\end{theorem}

Given a primitive recursive tree $[T]$, we say that $[T]$ is {\em perfect} if
it has no isolated elements.  Cenzer and Remmel also proved a number of index
set results for primitive recursive trees $T$ such that $[T]$ is perfect. Here
is one example.

\begin{theorem} \label{thm:pips} \begin{compactenum} 
\item[(a)]
$\{e: T_e\ \text{is r.b. and  $[T_e]$ is perfect}\}$ and \\
$\{e: T_e\ \text{is r.b.and  $[T_e]$ is  nonempty and perfect}\}$
are $D^0_3$-complete.
\item[(b)]
$\{e: T_e\ \text{is a.a.r.b. and  $[T_e]$ is  perfect}\}$ and \\
$\{e: T_e\ \text{is a.a.r.b. and  $[T_e]$ is}$  $\text{nonempty and 
perfect}\}$ are $D^0_3$-complete.
\item[(c)]
$\{e: T_e \ \text{is bounded and  $[T_e]$ is perfect}\}$ and \\
$\{e: T_e\ \text{is bounded and  $[T_e]$ is}$ {\it non\-empty}
$\text{and perfect}\}$
are $\Pi^0_4$-complete.
\item[(d)]
$\{e: T_e \ \text{is $a.a.$bounded and  $[T_e]$ is perfect}\}$ and \\
$\{e: T_e\ \text{is $a.a.$bounded and  $[T_e]$ is nonempty}$ $\text{ and
perfect}\}$ are $D^0_4$-complete.
\item[(e)]
$\{e: [T_e]\ \text{is perfect}\}$ and 
$\{e: [T_e]\ \text{is nonempty and perfect}\}$
are $\Sigma^1_1$-complete.
\end{compactenum} 
\end{theorem}

\section{Proofs of Theorems \ref{tree2prog} and \ref{prog2trees}}
\label{proofs}

The main goal of this section is prove Theorems \ref{tree2prog} and
\ref{prog2trees}. 

Recall that $\{ e\}^B$ denotes the function computed by the $e$-th oracle
machine with oracle $B$. If $A \subseteq \omega$, we write $\{ e\}^B = A$ if
$\{ e\}^B$ is the characteristic function of $A$. If $f$ is a function
$f\colon\omega\rightarrow \omega$, then  we let $gr(f) = \{ \la x,f(x)\ra
\colon x\in \omega\}$.  Given a finite normal predicate logic program $P$ and a
recursive tree $T \subseteq \omega^{< \omega}$, we say that there is an
effective one-to-one degree preserving correspondence between the set of stable
models of $P$ and the set of infinite paths through $T_P$ if there are indices
$e_1$ and $e_2$  of oracle Turing machines such that\\
(i) $( \forall M\in \Stab(P))
(\{ e_1 \}^M = f_M \in [T])$, and\\
(ii) $(\forall f\in [T]) (\{ e_2\}^{gr(f)}
= M_f \in \Stab(P))$, and \\
(iii) $(\forall M \in \Stab(P)) (\forall  f\in [T])
(\{ e_1\}^M = f  \  \Leftrightarrow
\  \{ e_2 \}^{gr(f)} = M )$.\\
Condition (i) says that the stable models of $P$ uniformly produce infinite
paths through the tree $T$ via an algorithm with index $e_1$ and condition (ii)
says that the infinite paths through the tree $T$ uniformly produce stable
models of $P$ via an algorithm with index $e_2$.  Finally,  condition (iii)
asserts that our correspondence is one-to-one and if $\{ e_1 \}^{M} = f$, then
$f$ is Turing equivalent to $M$.  In what follows, we will not explicitly
construct the indices $e_1$ and $e_2$, but our constructions will make it clear
that such indices exist.

\subsection{ The proof of Theorem \ref{tree2prog}.}

Suppose that $T$ is a recursive tree contained in $\omega^{< \omega}$. Note
that by definition, the empty sequence, whose code is 0,  is in $T$. 

A classical result, first explicit in \cite{sm68} and \cite{an78}, but known a
long time earlier in equational form, is that every r.e. relation can be
computed by a suitably chosen predicate over the least model of a finite
predicate logic Horn program.  An elegant method of proof due to Shepherdson
(see \cite{sh91} for references) uses the representation of recursive functions
by means of finite register machines. When such machines are represented by
Horn programs in the natural way, we get programs in which every atom can be
proved in only finitely many ways; see also \cite{ns92}.  Thus we have the
following proposition. 

\begin{proposition}\label{aux}
Let $r(\cdot ,\cdot )$ be a recursive relation. Then there is a finite
predicate logic program $P_r$ computing $r(\cdot, \cdot)$ such that every atom
in the least model $M_r$ of $P_r$ has only finitely many minimal proof schemes
and there is a recursive procedure such that given an atom $a$ in Herbrand base
of $P_r$ produces the code of the set of $P_r$-proof schemes for $a$. Moreover,
the least model of $P_r$ is recursive. $\hfill\Box$
\end{proposition}

It follows that given a recursive tree $T$ there exist the following three
finite normal predicate logic programs such that the ground terms in their
underlying  language are all of the form $0$ or $s^n(0)$ for $n \geq 1$ where
$0$ is a constant symbol and $s$ is a unary function symbol.  We shall use $n$
as an abbreviation for the term $s^n(0)$ for $n \geq 1$. In particular:
\begin{compactenum}
\item[(I)]  There exists a finite predicate logic Horn program $P_{T,0}$
such that for a predicate $\tree(\cdot )$ of the language of $P_{T,0}$, 
the atom $\tree(n)$ belongs to the least Herbrand model of $P_{T,0}$ if
and only if $n$ is a code for a finite sequence $\sigma$ and
$\sigma\in T$. 
\item[(II)]  There is a finite predicate logic Horn program $P_1$ such that for
a predicate $seq(\cdot)$ of the language of $P_1$, the atom $seq(n)$ belongs to
the least Herbrand model of $P_1$ if and  only if $n$ is the code of a finite
sequence $\alpha \in \omega^{< \omega}$.
\item[(III)]  There is a finite predicate logic Horn program $P_2$ which
correctly computes the following recursive predicates on codes of sequences. 
\begin{compactdesc}
\item[(a)] $\samelength (\cdot ,\cdot )$. This succeeds if and only
if both arguments are the codes of sequences of  the same length.
\item[(b)] $\diff (\cdot ,\cdot )$. This succeeds if and only if
the arguments are codes of sequences which are different.
\item[(c)] $\shorter (\cdot ,\cdot )$. This succeeds if and only
both arguments are codes of sequences and the first sequence is shorter 
than the second sequence.
\item[(d)] $\length (\cdot ,\cdot )$. This succeeds when the 
first argument is a code of a sequence and the second argument is the length
of that sequence.
\item[(e)] $\notincluded (\cdot ,\cdot )$. This succeeds if and only
if both arguments are codes of sequences and the first sequence 
is not an initial segment of the second sequence.
\item[(f)] $\num(\cdot)$. This succeeds if and only if the argument 
is either $0$ or $s^n(0)$ for some $n \geq 1$.
\end{compactdesc}
\end{compactenum}
Now let $P_T^{-}$ be the finite predicate logic program which is the union of
programs $P_{T,0}\cup P_1\cup P_2$.  We denote its language by ${\cal L}^{-}$
and we let $M^{-}$ be the least model of $P_T^{-}$.  By Proposition \ref{aux},
this program $P_T^{-}$ is a Horn program, $M^{-}$ is recursive, and for each
ground atom $a$ in the Herbrand base of $P^{-}$, we can explicitly construct
the set of all $P_T^{-}$-proof schemes of $a$. In particular, $\tree (n) \in
M^{-}$ if and only if  $n$ is the code of node in $T$. 

Our final program $P_T$ will consist of $P_T^{-}$ plus clauses (1)-(7) given
below. We assume no predicate that appears in the head of any of these clauses
is in the language ${\cal L}^{-}$. However, we do allow predicates from the
language of $P_T^{-}$ to appear in the body of clauses (1) to (7). It follows
that for any stable model of the extended program, its intersection with the
set of ground atoms of ${\cal L}^{-}$ will be $M^{-}$. In particular, the
meaning of the predicates listed above will always be the same.

We are ready now to write the additional clauses which, together with the
program $P_T^{-}$, will form the desired program $P_T$.  First of all, we
select the following three new unary predicates.
\begin{compactenum}
\item[(i)] $\ipath (\cdot )$, whose intended interpretation in any given stable
model $M$ of $P_T$ is that it holds only on the set of codes of sequences that
lie on infinite path  through $T$.  This path  will correspond to the path
encoded by the stable model of $M$;
\item[(ii)] $\notpath (\cdot )$, whose intended interpretation in any stable
model $M$ of $P_T$ is the set of all codes of sequences which are in $T$ but do
not satisfy $\ipath( \cdot)$. 
\item[(iii)] $\control (\cdot )$, which will be used to ensure that $\ipath(
\cdot)$ always encodes an infinite path through $T$.
\end{compactenum}

This given, the final 7 clauses  of our program are the following.\\
\ \\
(1) $\ipath (X) \longleftarrow \tree(X),\ \neg \notpath (X)$\\
(2) $\notpath (X) \longleftarrow \tree(X),\ \neg \ipath (X)$\\
(3) $\ipath (0) \longleftarrow$ \quad\quad\quad /* Recall 0 is the code of the 
empty sequence */\\
(4) $\notpath (X) \longleftarrow \tree(X),\ \ipath (Y),\tree (Y), \samelength
(X,Y), \diff (X,Y)$\\ (5) $\notpath (X) \longleftarrow \tree(X),\  \tree(Y),\
\ipath (Y), \ \shorter (Y,X), \notincluded (Y,X)$\\
(6) $\control (X)\longleftarrow \ipath (Y),\ \length (Y,X)$\\
(7) $\control (X) \longleftarrow  \neg \control (X),\num(X)$\\
\ \\
Clearly, $P_T  = P_T^{-} \cup \{ (1),\ldots ,(7)\}$ is a finite program. 

We should note that technically, we must insure that all the predicates that we
use in our finite normal predicate logic program $P_T$ come from our fixed
recursive language $\mathcal{L}$.  The predicates we have used in $P_T$ were
picked mainly for  mnemonic purposes, but since we are assuming that
$\mathcal{F}$ has infinitely many constant symbols and infinitely  many $n$-ary
relation symbols and $n$-ary functions symbols for each $n$, there is no
problem to substitute our predicate names by corresponding predicate names that
appear in $\mathcal{L}$. 

Our goal is to prove the following. 
\begin{compactenum}
\item[(A)] $T$ is a finitely branching recursive tree if and only if 
every element of $H(P_T)$ has only finitely many minimal proof schemes.
Thus, $T$ is finitely branching if and only if $P_T$ has the $\FS$ property.
\item[(B)] $T$ is highly recursive if and only if for every atom $a$ in 
$H(P_T)$, we can effectively find the set of all minimal $P_T$-proof schemes 
of $a$.
\item[(C)] There is a one-to-one degree preserving correspondence
between $[T]$ and $Stab(P_T)$.
\end{compactenum} 

First we prove (A) and (B).  When we add clauses (1)-(7), we note that no atom
of ${\cal L}^{-}$ is in the head of any of these new clauses. This means
that no ground instance of such a clause can be present in a minimal
$P_T$-proof scheme with conclusion being any atom of ${\cal L}^{-}$. This means
that minimal $P_T$-proof schemes with conclusion an atom $p$ of ${\cal L}^{-}$
can involve only clauses from $P_T^{-}$.  Thus, for any ground atom $a$ of
${\cal L}^{-}$, $a$ will have no minimal $P_T$-proof scheme if $a \notin M^{-}$
and we can effectively compute the finite set of $P_T$-proof schemes for $a$ if
$a \in M^{-}$.  Next consider the atoms appearing in the heads of clauses
(1)-(7). These are atoms of the following three forms:\\
(i) $\ipath (t)$,\\
(ii) $\notpath (t)$, and \\
(iii) $\control (t)$

The ground terms of our language are of form $n$, where $n\in \omega$, that is,
of the form $0$ or $s^n(0)$ for $n \geq 1$.  Note that all clauses that have
$\ipath(X)$ or $\notpath(X)$ have  in the body an occurrence of the atom
$\tree(X)$.  Thus for atoms of the form $\ipath (t)$ and $\notpath (t)$, the
only ground terms which possess a $P_T$-proof scheme must be those for which
$t$ is a code of a sequence of natural numbers belonging to $T$. The reason for
this is that predicates of the form $\tree(t)$ from ${\cal L}^{-}$ fail if $t$
is not the code of sequence in $T$. The only exception is clause (3) whose head
is $\ipath (0)$ and 0 is the code of the empty sequence which is in every tree
$T$ by definition.  This eliminates from our consideration ground atoms of the
form $\ipath (t)$ and $\notpath (t)$  with $t \notin T$.  Similarly, the only
ground atoms of the form $\control (t)$ which possess a proof scheme are atoms
of the form $\control (n)$ where $n$ is a natural number. 

Thus we are left with these cases:\\
(a) $\ipath (c(\sigma) )$  where $\sigma\in T$, \\
(b) $\control ( n )$  where $n \in \omega$, and \\
(c) $\notpath (c(\sigma) )$ where $\sigma\in T$.\\
\ \\
{\bf Case (a).} Atoms of the form $\ipath(c(\sigma) )$ where $\sigma\in T$.\\
There are only two type ground clauses $C$ with $\ipath(\cdot)$ in the head,
namely, those that are ground instances of clauses of type (1) and (3). Clause
(3) is a Horn clause.  This implies that a minimal $P_T$-proof scheme which
derives $\ipath(0)$ and uses (3) must be of the form $\langle \langle \ipath
(0), (3) \rangle, \emptyset  \rangle$.  Next consider a minimal $P_T$-proof
scheme $\PS$ of $\ipath(c(\sigma))$ which contains clause (1).  In such a case,
$\PS$ will consist of the sequence  of pairs of a minimal $P_T^{-}$-proof
scheme of $\tree(c(\sigma))$ which will have empty support followed by the
pair   $\langle \ipath (c(\sigma) ) , (1)^* \rangle $ where $(1)^*$ is the
result of substituting $c (\sigma)$ for $X$ in clause (1). The support of $\PS$
will be $\{ \notpath (c(\sigma) )\}$.  Since we are assuming that
$\tree(c(\sigma))$ has only finitely many $P_T^{-}$-proof schemes and we can
effectively find them, it follows that $\ipath(c(\sigma))$ has only finitely
many minimal $P_T$-proof schemes and we can effectively find all of them.\\
{\bf Case (b)}. Atoms of the form $\control (n)$ where $n \in \omega$. \\
There are only two types of ground instances of clauses with the atom
$\control(n)$ in the head, namely, ground instances of clauses (6) and (7). The
only minimal $P_T$-proof schemes of $\control(n)$ that use a ground instance of
clause (7) must consist of the sequence of pairs in  a minimal $P_T^{-}$-proof
scheme of $\num(n)$ followed by the pair  $ \langle \control(n), (7)^*\rangle$ 
where $(7)^*$ is the result of substituting $n$ for $X$ in $(7)$.  Thus the support
of such a minimal $P_T$-proof scheme is $ \{ \control(n)\}$ Since we are
assuming that $\num(n)$ has only finitely many minimal $P_T^{-}$-proof schemes
and we can effectively find them, we can effectively find all minimal
$P_T$-proof schemes of $\control(n)$ that uses a ground instance of (7).  If we
have a minimal $P_T$-proof scheme $\PS$ with conclusion $\control(n)$ that uses
a ground instances of clause (6), then the last term of $\PS$ must be of the
form 
\[
\langle \control(n), \control(n) \leftarrow
\ipath(c(\tau)),\length(c(\tau),n)\rangle
\]
where $c(\tau)$ is the code of node in $T$ of length $n$.  Moreover, in $\PS$,
this triple must be preceded by some interweaving of the sequences of pairs in
minimal $P_T$-proof schemes for $\ipath(c(\tau))$ and $\length(c(\tau),n)$.  
Now we effectively find the finite set of minimal $P_T^{-}$-proof schemes for
$\length(c(\tau),n)$ and we can effectively find the set of all $P_T$-minimal
proof schemes for $\ipath(c(\tau))$.  Moreover, it must be the case that
support of $\PS$ is $A$ the support of the minimal $P_T$-scheme of
$\ipath(c(\tau))$ that was inter-weaved with one of the minimal proof schemes
for $\length(c(\tau)),n)$ to create $\PS$. Since the support of any proof
scheme for $\ipath(c(\tau))$ where $|\tau| \geq 1$ is just
$\{\notpath(c(\tau))\}$, it follows that $A = \{\notpath(c(\tau))\}$ 
if  $|\tau| \geq 1$ and $A = \emptyset$ if $|\tau| =0$. 
Now, if $T$ is finitely branching, there will only be finitely many choices for
$\tau$ since to derive $\ipath(c(\tau))$, $\tau$ must be in $T$.  Hence there
will be only finitely many choices of $\PS$.  On the other hand, if $T$ is not
finitely branching, then there will be an $n$ such that there are infinitely 
many nodes $\tau \in T$ of length $n$ for some $n > 0$ so that there will 
be infinitely many different supports of minimal $P_T$-proof schemes for
$\control(n)$.  If $T$ is highly recursive, then we can effectively find all
$\tau \in T$ of length $n$ so that we can effectively find all such proof
schemes $\PS$.  Similarly, if $P_T$ has the rec. $\FS$ property, then for $n >
0$, we can read off all the nodes in $T$ of length $n$ from the supports of the
minimal $P_T$-proof schemes of $\control(n)$ so that $T$ will be highly
recursive.  Thus $T$ is finitely branching if and only if there are finitely
many minimal $P_T$-proof schemes for $\control(n)$ for each $n \geq 0$.
Similarly, if $T$ is highly recursive, then we can effectively find all the
minimal $P_T$-proof schemes for $\control(n)$ for each $n \geq 0$ and if $P_T$
has the rec. $\FS$ property, then $T$ is highly recursive. \\
\ \\
{\bf Case (c).} Atoms of the form $\notpath (c(\sigma) )$.\\
Here we have to take into account clauses (2), (4), and (5). First, consider a
minimal $P_T$-proof scheme $\PS$ of $\notpath(c(\sigma))$ which contains a
ground instance of clause (2).  In such a case, the sequence of pairs in $\PS$
will consist of the sequence of pairs a minimal $P_T^{-}$-proof scheme of
$\tree(c(\sigma))$ which will have empty support followed by the pair  
\[
\langle \notpath (c(\sigma) ) , (2)^*\rangle 
\]
where $(2)^*$ is the result of substituting $c(\sigma)$ for $X$ in $(2)$.  The
support of $\PS$ is $\{ \ipath (c(\sigma) )\}$.  Since we are assuming that
$\tree(c(\sigma))$ has only finitely many minimal $P_T^{-}$-proof schemes and
we can effectively find them, it follows that $\notpath(c(\sigma))$ has only 
finitely many minimal $P_T$-proof schemes that use a ground instance of clause
(2) and we can effectively find them. 

Next, consider a $P_T$-proof scheme $\PS$ with conclusion $\notpath(c(\sigma))$
which contains a ground instance of clause (4). Then there must exists a $\tau
\in T$ of length $|\sigma|$ such that the last pair in the proof scheme is of
the form 
\begin{equation}\label{4:tau}
\langle c(\sigma), (4)^*\rangle
\end{equation}
where $(4)^*$ is the result of substituting $c(\sigma)$ for $X$ and $c(\tau)$
for $Y$ in $(4)$.  Then $\PS$ must consist of an interweaving of the sequences
of pairs of the minimal $P_T^{-}$-proof schemes for $\tree(c(\sigma))$,
$\tree(c(\tau))$, $\samelength(c(\sigma), c(\tau))$, and $\diff(c(\sigma),
c(\tau))$ and a minimal $P_T$-proof scheme $\ipath(c(\tau))$ with support $A$.
Then the support of $\PS$ will be $A$.  In each case, there are only finitely
many such minimal $P_T$-proofs schemes of these atoms and we can effectively
find them. Thus for each $\tau \in T$ of length $|\sigma|$, we can effectively
find all the minimal $P_T$-proof schemes of $\notpath(c(\sigma))$ that end in a
triple of the form of  (\ref{4:tau}).  Now if $T$ is finitely branching, it
follows that there will be only finitely many minimal $P_T$-proof schemes that
use a ground instance of clause $(4)$ and, if $T$ is highly recursive, then we
can effectively find all $\tau \in T$ of length $|\sigma|$ so that we can
effectively find all minimal $P_T$-proof schemes that use a ground instance of
clause $(4)$. 

Finally  let us consider a $P_T$-proof scheme $\PS$ with conclusion
$\notpath(c(\sigma))$ which contains ground instance of clause (5). Then there
must exists a $\tau \in T$ whose length is less than the length of $\sigma$ and
which is not an initial segment of $\sigma$ such that the last pair in the
proof scheme is of the form 
\begin{equation}\label{5:tau}
\langle c(\sigma), (5)^*\rangle
\end{equation}
where $(5)^*$ is the result of substituting $c(\sigma)$ for $X$ and $c(\tau)$
for $Y$ in $(5)$.  Then $\PS$ must consist of an interweaving of sequences of
pairs in the minimal $P_T^{-}$-proof schemes for $\tree(c(\sigma))$,
$\tree(c(\tau))$, $\shorter(c(\tau), c(\sigma))$, and $\notincluded(c(\tau),
c(\sigma))$ and a minimal $P_T$-proof scheme of $\ipath(\tau)$ with support
$A$. Then the support of $\PS$ is $A$.  In each case, there are only finitely
many minimal $P_T$-proofs schemes of these atoms and we can effectively find
them. Thus for each $\tau$ whose length is less than the length of $\sigma$ and
which is not an initial segment of $\sigma$, we can effectively find all the
minimal $P_T$-proof schemes of $\notpath(c(\sigma))$ that end in a pair of the
form of  (\ref{5:tau}).  Now if $T$ is finitely branching, it follows that
there will be only finitely many minimal $P_T$-proof schemes that use a ground
instance of clause $(5)$ and, if $T$ is highly recursive, then we can
effectively find all $\tau \in T$ of length $|\sigma|$ so that we can
effectively find all minimal $P_T$-proof schemes that use a ground instance of
clause $(5)$. 

Thus, we have proved that if $T$ is finitely branching, then every ground atom
possesses only finitely many minimal $P_T$-proof schemes and if $T$ is highly
recursive, then for every ground atom $a \in H(P_T)$, we can effectively find
the set of all minimal $P_T$-proof schemes of $a$. Thus if $T$ is finitely
branching, then $P_T$ has the $\FS$ property and if $T$ is highly recursive,
then $T$ has the rec. $\FS$ property.  On the other hand, we have shown by our
analysis in  (b) that if $P_T$ has the $\FS$ property, then $T$ must be
finitely branching and if $P_T$ has the rec. $\FS$ property, then $T$ is highly
recursive. This proves  (A) and (B) and establishes parts (2) and (3) of
Theorem \ref{tree2prog}. 

To prove (C), we shall establish a ``normal form'' for the stable models of
$P_T$. Each such model must contain $M^{-}$, the least model of $P_T^{-}$. In
fact, the restriction of a stable model of $P_T$ to $H(P_T^{-})$ is $M^{-}$. 
Given any $\beta = (\beta{(0)}, \beta{(1)}, \ldots ) \in \omega^{\omega}$,
recall that  $\beta \res n = (\beta(0),\beta(1),\dots,\beta(n-1))$. Then we let
\begin{eqnarray}\label{Mbeta}
M_{\beta} =&& M^{-} \cup \{ \control (n) : n\in \omega\} \cup 
\{\ipath(0)\} \cup 
\{ \ipath (c(\beta \res {n}) : n \in \omega \} \cup\ \nonumber \\
&&\{ \notpath (c(\sigma) ) :\sigma \in T \ \mbox{and} \ 
\sigma \not \prec \beta \}. 
\end{eqnarray}
We claim that $M$ is a stable model of $P_T$ if and only if 
$M = M_{\beta}$ for some $\beta \in [T]$. 

First, assume that $M$ is a stable model of $P_T$. Thus $M$ is the least model
of the Gelfond-Lifschitz transform $(ground(P_T))_M$. We know that the atoms
of  ${\cal L}^{-}$ in  $M$ constitute $M^{-}$.  Let us observe that since the
clause (3) belongs to our program, $\ipath(0) \in M$.  Thus we can not use
clause (2) to derive that $\notpath(0)$ is in $M$. Moreover, it is easy to see
that we cannot use clauses of the form (4) or (5) to derive that $\notpath(0)$
is in $M$ so that it must be the case that $\notpath(0) \notin M$.  Next,
suppose that $\sigma \in T$ and length of $\sigma$ is greater than or equal to
1. It is easy to see from clauses (1) and (2) that it cannot be the case that
neither $\ipath(c(\sigma))$ and $\notpath(c(\sigma))$ are in $M$.  Since
clauses of the form of (1) are the only clauses that we can use to derive that
the atom $\ipath(c(\sigma))$ is in the least model of $(ground(P_T))_M$ when
$|\sigma| \geq 1$, it follows that it cannot be the case that both
$\ipath(c(\sigma))$ and $\notpath(c(\sigma))$ are in $M$. Thus exactly one of
$\ipath(c(\sigma))$ and $\notpath(c(\sigma))$ must be in $M$ for all $\sigma
\in T$.  Next we claim that $\control(n) \in M$ for all $n$.  That is, if 
$\control(n) \notin M$ for some $n$, then the Gelfond-Lifschitz transform of
the ground clause $\control(n) \leftarrow \neg \control(n),\num(n)$ from (7)
would be $\control(n) \leftarrow num(n)$ which would force $\control(n)$ to be
in $M$.  Since $\control(n) \in M$, the only way that one could derive that
$\control(n)$ is in the least model of $(ground(P_T))_M$ is via a proof scheme
that uses a ground instance of clause  (6). This means that for each $n \geq
0$, there must be a $\tau^{(n)} \in T$ of length $n$ such that
$\ipath(c(\tau^{(n)})) \in M$. But then we can use clause (4) to show that if
$\sigma$ is a node in $T$ of length $n$ which is different from $\tau^{(n)}$,
then $\notpath(c(\sigma)) \in M$.  But now  the clauses of type (5) will force
that it must be the case that if $m < n$, then $\tau^{(m)}$ must be an initial
segment of $\tau^{(n)}$.  Thus the path $\tau$ where $\tau^{(n)} \sqsubseteq
\tau$ for all $n$  is  an infinite path through $T$ and $M = M_\tau$.  Note
that this shows that if $[T]$ is empty, then  $P_T$ has no stable model.

To complete the argument for (C), we have to prove that $\beta\in [T]$ implies
that $M_{\beta}$ is a stable model of $P_T$. Let $lm(M_\beta)$ be the least
model of $(ground(P_T))_{M_\beta}$.  The presence of clauses (1) and (2) in
$P_T$ implies that $\{ \ipath (c(\beta\res {(n)}) : n\in \omega \} \cup
\{ \notpath (c(\sigma) ): \sigma \in T \setminus \{ \beta\res {(n)} : n\in
\omega\}\} \subseteq lm(M_{\beta })$.  Then clause (6) can be used to show that
for all $n$, $\control (n)$ also belongs to $lm(M_{\beta})$.  Since
$M^{-}\subseteq lm(M_{\beta})$, it follows that $M_{\beta}\subseteq
lm(M_{\beta})$. 

Next we must prove that $lm(M_{\beta})\subseteq M_{\beta}$.  We know that since
none of the heads of rules (1)-(7) involve predicates in $H(P_T^{-})$, it must
be the case that $lm(M_\beta) \cap H(P_T^{-}) =M^{-}$.  The only ground 
clauses from (1) that are in $(ground(P_T))_{M_\beta}$ are clauses of the form
\[ 
\ipath(c(\beta\res {n})) \leftarrow \tree(c(\beta\res {n})). 
\]
These are the only clauses of  $(ground(P_T))_{M_\beta}$ which have
$\ipath(c(\sigma))$ in the head for $\sigma \in T$ so that
$\{\ipath(c(\sigma)):\sigma \in T\}\cap lm(M_\beta) \subseteq M_\beta$.  Since
$M_\beta$ contains all ground clauses of the form $\control(n)$, the only
clauses that we have to worry about are clauses with the ground atom
$\notpath(c(\sigma))$ in the head for $\sigma \in T$.  The only ground clause
from (2) that are in $(ground(P_T))_{M_{\beta}}$ are clauses of the form 
$$\notpath(c(\sigma)) \leftarrow \tree(c(\sigma))$$ where $\sigma \in T-
\{\beta^{(n)}: n \geq 0\}$. Thus the conclusion of all such clauses are in
$M_\beta$.  Thus we are reduced to considering ground clauses of the form (4)
and (5).  Since all such clauses must have an atom $\ipath(c(\tau))$ in 
the body, the only way we can use these clauses is to derive 
$\notpath(c(\sigma))$ in its  head and this happens if $\tau \in \{\beta^{(n)}:
n \geq 0\}$. But then it easy to see that this forces $\sigma \notin
\{\beta^{(n)}: n \geq 0\}$.  Thus the only atoms  $\notpath(c(\sigma)) \in
lm(M_\beta)$ are those with $\sigma \in T-\{\beta^{(n)}: n \geq 0\}$. Thus 
$lm(M_\beta) \subseteq M_\beta$. This proves part (1) of Theorem
\ref{tree2prog}. 

Finally, consider part (4) of Theorem \ref{tree2prog}.  By part (3), we know
that $T$ is highly recursive if and only if $P_T$ has the rec. $\FS$ property.
We must show that if $T$ is decidable and recursively bounded, then $P_T$ is
decidable.  So suppose we are given a set of ground atoms $\{a_1, \ldots,
a_n\}$ and corresponding minimal $P_T$-proof schemes $\PS_i$ of $a_i$.  For
these atoms to to belong to a stable model  $M$ of $P_T$, it must be the case
that the ground atoms in the language of $P_T^{-}$ must all be in $M^{-}$ and
there corresponding proof schemes must be the least minimal proofs schemes for
$P_T^{-}$.  This we can check recursively.  The remaining atoms are of the form
$\ipath(c(\sigma))$, $\notpath(c(\tau))$, and $\control (n)$. It must be the
case that atoms of the form $\ipath(c(\sigma))$ and $\notpath(c(\tau))$ among 
$\{a_1, \ldots, a_n\}$ must be consistent with being the initial segment of the
path through $T$.  If that is not the case, then it is clear that $\{a_1,
\ldots, a_n\}$ is not contained in a stable model of $P_T$. If it is the case,
let $\alpha$ be the longest string $\sigma$ such that $\ipath(c(\sigma)) \in
\{a_1, \ldots, a_n\}$.  Now if $\alpha \notin Ext(T)$, then again 
$\{a_1, \ldots, a_n\}$ is not contained in a stable model of $P_T$. If it is,
then let $m$ be the maximum of all $n$ such that $\control(n) \in \{a_1,
\ldots, a_n\}$ and $|\tau|$ such that $\notpath(c(\tau)) \in \{a_1, \ldots,
a_n\}$.  Since $T$ is recursively bounded, then we can effectively find all
strings of length $m$ which extend $\alpha$. Now if there is a string $\beta$
of length $m$ such that $\alpha \prec \beta$, $\beta \in Ext(T)$,  and there is
no initial segment $\gamma$ of $\beta$ such that $\notpath(c(\gamma)) \in
\{a_1, \ldots, a_n\}$, then it will be the case that $\{a_1, \ldots, a_n\}$ is
contained in a stable model.  For each such $\beta$ and all $\delta \in T$ of
length less than or equal to $m$, the only minimal proof schemes of ground
atoms of the form $\ipath(c(\delta))$, $\notpath(c(\delta)$, and $\control(n)$
for $n \leq m$ depend only on the ground atoms $\ipath(c(\gamma))$ for $\gamma$
contained in $\beta$. Thus by our analysis of Cases {\bf (a)}-{\bf (c)} above,
we can compute the appropriate minimal proofs schemes and then check if the
corresponding minimal $P_T$-proof schemes equals $\{\PS_1, \ldots, \PS_n\}$.
Thus $P_T$ is decidable. 

On the other hand, suppose that $P_T$ has the rec. $\FS$ property and $P_T$ is
decidable.  Then given a node $\beta =(\beta_1, \ldots, \beta_n) \in T$, it is
easy to see that for any path $\pi \in \omega^{\omega}$ which extends $\beta$,
the elements of $M_{\pi}$ which mention only $\beta$, nodes of length $\leq
|\beta|$, and the elements $0,s^1(0) \ldots, s^{|\beta|}(0)$ are the same.
Thus let 
\begin{eqnarray}\label{Mbeta0}
M_{\beta} =&& M^{-} \cup \{ \control (n) : n \leq |\beta|\} \cup 
\{\ipath(0)\} \cup 
\{ \ipath (c(\alpha) : \alpha \sqsubseteq \beta \} \cup\ \nonumber \\
&&\{ \notpath (c(\sigma) ) :\sigma \in T, |\sg| \leq |\beta|, \ \mbox{and} \ 
\sigma \not \prec \beta \}. 
\end{eqnarray}
Then $M_\beta$ is finite and our analysis shows that we can effectively find
all the minimal $P_T$-proofs schemes $\PS_1, \ldots, \PS_r$ which mention only
$\beta$, nodes of length $\leq |\beta|$, and the elements $0,s^1(0) \ldots,
s^{|\beta|}(0)$ which  have conclusions in $M_\beta$.  By the decidability of
$P_T$, we know whether there is a stable model which contains $M_\beta$ and has
$\PS_1, \ldots, \PS_r$ has the corresponding minimal  $P_T$-proof schemes for
elements in $M_\beta$.  If there is such a stable model, then $\beta$ must be
an initial segment of some $\pi \in [T]$ so that $\beta \in Ext(T)$. If there
is no such stable model, then there is no infinite path $\pi \in [T]$ such that
$\beta \sqsubseteq \pi$ so that $\beta \not \in Ext(T)$.  Thus if $P_T$ is
decidable and has the rec. $\FS$ property, then $T$ is decidable and highly
recursive. This completes the proof of Theorem \ref{tree2prog}.

\subsection{Proof of Theorem \ref{prog2trees}}

Suppose that we are given a finite normal predicate logic program $P$. Then  by
our remarks in the previous section, the Herbrand base $H(P)$ will be primitive
recursive, $ground(P)$ will be a primitive recursive program and, for any atom
$a \in H(P)$, the set of minimal $P$-proof schemes with conclusion $a$ is
primitive recursive.  We should note, however, that it is not guaranteed that
the $Support(a)$ which is the set of $can(X)$ such that $X$ is the support of a
minimal $P$-proof scheme of $a$ is recursive.  However, it is the case that
$Support(a)$ is an r.e. set. 

Our basic strategy is to encode a stable model $M$ of $ground(P)$ by a path
$f_M = (f_0,f_1,\ldots )$ through the complete $\omega$-branching tree
$\omega^{<\omega}$ as follows.\\
(1) First, for all $i\geq 0$, $f_{2i} = \chi_M(i)$. That is, at the stage $2i$,
we encode the information about whether or not  the atom encoded by $i$ belongs
to $M$. Thus, in particular, if $i$ is not the code of ground atom in $H(P)$,
then $f_{2i} =0$.\\
(2) If $f_{2i} = 0$, then we set $f_{2i+1} = 0$.  But if $f_{2i} = 1$ so that
$i\in M$ and $i$ is the code of a ground atom in $H(P)$, then we let $f_{2i+1}$
equal $q_M(i)$ where $q_M(i)$ is the least code for a minimal $P$-proof scheme
$\PS$ for $i$ such that the support of $\PS$ is disjoint from $M$.  That is, we
select a minimal $P$-proof scheme $\PS$ for $i$, or to be precise for the atom
encoded by $i$, such that $\PS$ has the smallest possible code of any minimal
$P$-proof scheme $\mathbb{T}$ such that $supp(\mathbb{T})\cap M = \emptyset$.
If $M$ is a stable model, then, by Proposition \ref{keystable}, at least one
such minimal $P$-proof scheme exists for $i$.

Clearly $M \leq_T f_M$ since it is enough to look at the values of $f_M$ at
even places to read off $M$.  Now, given an $M$-oracle, it should be clear that
for each $i\in M$, we can use an $M$-oracle to find $q_M(i)$ effectively. This
means that $f_M\leq_T M$. Thus the correspondence $M\mapsto f_M$ is an
effective degree-preserving correspondence.  It is trivially one-to-one.

Next we construct a primitive  recursive tree $T_P \subseteq \omega^{\omega}$
such that $[T_P] = \{f_M: M\in stab(P)\}$.  Let $N_k$ be the set of all codes
of minimal $P$-proof schemes $\PS$ such that all the atoms appearing in all the
rules used in $\PS$ are smaller than $k$. Obviously $N_k$ is finite. 
It follows from our remarks in the previous section that since $P$ is a finite
normal predicate logic program,  the predicate ``minimal $P$-proof scheme''
which holds only for codes of minimal $P$-proof schemes is a primitive
recursive predicate. This means that there is a primitive recursive function
$h$ such that $h(k)$ equals to the canonical index for $N_k$.  Moreover, given
the code of sequence $\sigma = (\sigma (0),\ldots ,\sigma (k)) \in
\omega^{<\omega}$, there is a primitive recursive function which will produce
canonical indexes of the sets  $I_{\sigma} = \{ i: 2i \leq k \land \sigma (2i)
= 1 \}$ and $O_{\sigma} = \{ i: 2i \leq k \land \sigma (2i) = 0 \}$. 

For any given $k \geq 2$, we let  $\overline{k} = max(\{2j+1:2j+1 < k\}$ and if
$\sigma = (\sigma(0), \ldots, \sg(k))$ is an element of $\omega^{< \omega}$,
then we let $\overline{\sg} = (\sg(0),\ldots, \sg(\overline{k}))$. If  $k =1$
and $\sg =(\sg(0))$, then we let $\overline{k} =0$ and $\overline{\sg} =
\emptyset$. In what follows, we shall identify each atom in $H(P)$ with its
code.  Then we  define $T_P$ by putting a node $\sigma = (\sigma(0), \ldots,
\sg(k))$ into $T_P$ if and only if the following five conditions are met:
\begin{compactenum}
\item[(a)] 
If $2i+ 1 \leq \bar{k}$ and $\sigma (2i) = 0 $ then $\sigma (2i+1 ) = 0$; 
\item[(b)] 
then $\sigma (2i+1 ) = q$, where $q$ is a code for a minimal $P$-proof
scheme $\PS$ such that $concl(\PS ) = i$, $supp (\PS ) \cap I_{\overline{\sg}}
= \emptyset$, and there is no number $j < \sigma (k)$
such that $j$ is a code for a minimal $P$-proof scheme $\mathbb{T}$ 
with conclusion $i$ such that $supp(\mathbb{T}) = supp(\PS)$;
\item[(c)] If $2i+ 1 \leq \overline{k}$ and $\sigma (2i) = 1$ then there is no
code $c \in N_{\lfloor k/2 \rfloor}$ of a minimal $P$-proof scheme $\PS$ 
such that $conc (\PS ) = i$, $supp (\PS ) \subseteq O_{\overline{\sg}}$
and $c < \sigma (2i+1)$    (Here $\lfloor \cdot\rfloor$ is the 
number-theoretic ``floor'' function); 
\item[(d)] If $2i +1 \leq  \overline{k}$ and $\sigma (2i) = 0$ then there is no
code $c\in N_{\lfloor k/2\rfloor}$ of a minimal $P$-proof scheme $\mathbb{T}$
such that $concl(\mathbb{T} ) = i$ and $supp (\mathbb{T} ) \subseteq
O_{\overline{\sg}}$; and 
\item[(e)] If $\overline{k} = 2i+1$ and $\sigma(2i) = 1$, then $\sigma(2i+1)=q$
where $q$ is a code for a minimal $P$-proof scheme $\PS$ such that $concl(\PS )
= i$ and there is no number $j < \sigma (k)$ such that $j$ is a code for a
minimal $P$-proof scheme $\mathbb{T}$ with conclusion $i$ such that $supp(\PS)
= supp(\mathbb{T})$. 
\end{compactenum}

The first thing to observe is that each of the conditions (a)-(e) requires that
we check only a bounded number of facts about codes that have an explicit bound
in terms of the code of $\sigma$.  This implies that $T_P$ has a primitive
recursive definition. It is immediate from our conditions defining $T_P$   that
if $\sigma \in T_P$ and $\tau \prec \sigma$, then $\tau\in T_P$. Thus $T_P$ is
a primitive recursive tree. Conditions (a) and (b) ensure that the set of 
all paths $\pi$ through $T_P$ meet the minimal conditions to be of the form
$f_M$ for some stable model.  That is, condition (a) ensures that if $\pi(2i) =
0$, then $\pi(2i+1) =0$.  Condition (b) ensures that if $\pi(2i) =1$, then 
$\pi(2i+1)$ is the code of a minimal $P$-proof scheme with conclusion $i$ 
and there is no smaller code of a minimal $P$-proof scheme of $i$ with the same
support. Conditions (c), (d) and (e) are carefully designed to ensure that
$T_P$ has the properties that we want. First, condition (c) limits the possible
infinite paths through $T_P$.  We claim that if  $\pi$ is an infinite path
through $T_P$ and $\pi(2i)=1$, then $\pi(2i+1)=r$ where $r$ is smallest code 
of minimal $P$-proof scheme with conclusion $i$ whose support does not
intersect $M_{\pi} = \{j:\pi(2j) =1\}$.  That is, if $\pi(2i+1)$ is the code of
minimal $P$-proof scheme with conclusion $i$ whose support is disjoint from
$M_\pi$ which is greater than $r$, then there will be some $k > 2i+1$ such 
that $c \in N_{\lfloor k/2 \rfloor}$ in which case condition (d) 
would not allow $(\pi(0), \ldots, \pi(k+2))$ to be put into $T_P$. Similarly,
if $\pi(2i+1)$ is the code of minimal $P$-proof scheme $\PS$ with conclusion
$i$ whose support is not disjoint from $M_\pi$, then there will be some $k >
2i+1$ such that $supp(\PS) \cap I_{(\pi(0), \ldots, \pi(k))} \neq \emptyset$ in
which case condition (b) would not allow  $(\pi(0), \ldots, \pi(k+2))$ to 
be put into $T_P$.  Likewise, condition (d) ensures that if $\pi(2i)=0$, there
can be no minimal $P$-proof scheme $\PS$ with conclusion $i$ whose support is
disjoint from $M_\pi$ since otherwise for large enough $k$, condition (e) would
not allow $(\pi(0), \ldots, \pi(k))$ to be put into $T_P$. Finally, condition
(e) is designed to ensure that $T_P$ is finitely branching if and only if $P$
has the $\FS$ property or has an explicit initial blocking set.  We note that 
for a node $(\sg(0),\ldots, \sg(2i),\sg(2i+1))$ where $\sg(2i) =1$, $\sg(2i+1)$
can be the code of {\em any} minimal $P$-proof scheme $\PS$ with conclusion $i$
for which there is no smaller number which codes a proof scheme with the same
conclusion and same support.  For example, we do not require  
$supp (\PS ) \cap I_{\sigma} = \emptyset$. However, if  $supp (\PS ) \cap
I_{\sigma} \neq \emptyset$, then condition (b) will ensure that there are no
extensions of $\sg$ in $T$. 

Our next goal is to show that every $f \in [T_P]$ is of the form $f_M$ for a
suitably chosen stable model $M$ of $P$.  It is clear that if $M$ is stable
model of $P$, then for all $k$, $(f_M(0), \ldots, f_M(k))$ satisfies conditions
(a)-(e) so that $f_M \in [T_P]$. Thus $\{f_M: M \in Stab(P)\} \subseteq [T_P]$. 

Next, let us assume that $\beta = (\beta(0),\beta (1),\ldots )$ is an
infinite path through $T_P$ and $M_{\beta} = \{ i: \beta (2i) = 1\}$.
Then we must prove that\\
(I) $M_{\beta}$ is a stable model of $P$ and \\
(II) $f_{(M_{\beta})} = \beta$.

For (I), suppose that $M_{\beta}$ is not a stable model of $P$.  Let
$lm(M_\beta)$ be the least model of Gelfond-Lifschitz transform
$ground(P)_{M_\beta}$ of $ground(P)$ relative to $M_\beta$. Then by Proposition
\ref{keystable}, it must be the case that either 
\begin{compactdesc}
\item[(i)] there is $j\in M_{\beta} \setminus lm(M_{\beta})$, or 
\item[(ii)] there is $j\in lm(M_{\beta})\setminus M_{\beta}$.
\end{compactdesc}

If {\bf (i)} holds, then let $i$ be the least $j\in M_{\beta} \setminus
lm(M_{\beta})$ and consider the string $\beta\res {(2i+3)}= (\beta (0),\ldots
,\beta (2i+3) )$.  For $\beta\res {(2i+3)}$ to be in $T$, it must be the case
that $\beta (2i+1)$ is a code of a minimal proof scheme $\PS$ such that
$concl(\PS ) = i$ and $supp (\PS ) \cap I_{\beta\res {(2i+1)}} = \emptyset$.
But since $i\notin lm(M_{\beta})$, there must be some $n$ belonging to
$M_{\beta}\cap supp(\PS)$. Clearly, it must be the case that $n > i$. Choose
such an $n$. Then $\beta\res 2n\notin T$ because $supp(\PS )\cap I_{\beta\res
2n} \neq \emptyset$, which contradicts our assumption that $\beta \in [T]$.
Thus {\bf (i)} cannot hold. 
 
If {\bf (ii)} holds, then let $i$ be the least $j\in lm(M_{\beta})  \setminus
M_{\beta}$ and consider again  $\beta\res {(2i+3)}$. Since $ i \in
lm(M_\beta)$, there must be a proof scheme $\mathbb{T}$ such $concl(\mathbb{T}
) = j$ and $supp(\mathbb{T}) \cap M_\beta = \emptyset$.  But then there is an
$n > 2i+1$ large enough so that $supp (\mathbb{T} ) \subseteq O_{\beta\res n}$.
But then $\beta\res n$ does not satisfy the condition (e) of our definition to
be in the tree which again contradicts our assumption that $\beta \in [T]$.
Thus {\bf (ii)} also cannot hold so that $M_\beta$ must be a stable model of
$P$. 

Thus we need only to verify claim (II), namely, that $\beta = f_{(M_{\beta})}$.
Now if $\beta \neq f_{(M_{\beta})}$, then it must be that case that for some
$i\in M_{\beta}$, there is a code $c$ of a minimal proof scheme $\PS$ such that
$concl(\PS ) = i$, $supp (\PS )\cap M_{\beta} = \emptyset$ and $c < \beta
(2i+1)$. But then there is an $n > 2i+1$ large enough so that $supp(\PS
)\subseteq O_{\beta\res n }$ and hence $\beta\res n$ does not satisfy condition
$(d)$ of our definition to be in $T$. Hence, if $\beta \neq f_{(M_{\beta})}$,
then $\beta\res n \notin T_P$ for some $n$ and so $\beta \notin [T_P]$. This
completes the proof of (II) and hence part (1) of the theorem holds. 

Next consider parts (2)-(10).  Note that the tree $T_P$ has the property that
if $\beta \in T$ where $\beta$ has length $n$, then \\
($\dag$) for every $i$ such that $2i \leq n$, $\beta (2i) \in
\{ 0,1\}$ and \\
($\ddag$) for every $i$ such that $2i+1 \leq n$, $\beta (2i+1)$ is either $0$
or it is a code of a minimal $P$-proof scheme $\PS$ such that $concl(\PS ) = i$
and no $j < \beta (2i+ 1)$ is the code of a minimal $P$-proof scheme of $i$
with the same conclusion and the same support.

Thus if $P$ has a finite number of supports of minimal $P$-proof schemes for
each $i$, then $T_P$ will automatically be finitely branching.  Next suppose
that  $P$ has the additional property that there is a recursive function $h$
whose value at $i$ encode all the supports of minimal $P$-proof schemes for
$i$.  Say, the possible support of minimal $P$-proof schemes for $i$ are 
$S^i_1, \ldots, S^i_{\ell_i}$. Then for each $1 \leq j \leq \ell_i$, we can
effectively find the smallest code $c^i_j$ of a minimal $P$-proof scheme for 
$i$ with support $S_i$. Thus for each $i$, we can use $h$ to compute $c^i_1,
\ldots, c^i_{\ell_i}$. But then we know  that the possible values of
$\sg(2i+1)$ for any $\sg \in T_P$ must come from $0,c^i_1, \ldots,
c^i_{\ell_i}$ so that  $T_P$ is recursively bounded.  Next observe that 
if $P$ has the $a.a.$ $\FS$ support property, then it will be the case that for
all sufficiently large $i$, there will be only a finite number of supports of
minimal $P$-proof schemes of $i$ so that $T_P$ will be nearly bounded.
Similarly, if $P$ has the $a.a.$ rec. $\FS$ support property, then it will 
be the case that for all sufficiently large $i$, we can effectively  find the
supports of all minimal $P$-proof schemes of $i$ so that as above, we can
effectively find the possible values of $\sg(2i+1)$ and, hence, $T_P$ will be
nearly recursively bounded.

Next, suppose that $P$ does not have the $\FS$ property. Let $i$ be the least
atom such that there exist infinitely many supports of $P$-proof schemes with
conclusion $i$. Now suppose that there is a node $\sg =(\sg(0), \ldots
\sg(2i+1))$ of length $2i+1$ in $T_P$.  It is easy to check that it will also
be the case that $\sg^* =(\sg(0), \ldots, \sg(2i-1),1, r)$ is a node in $T_P$
where $r$ is any code of a minimal $P$-proof scheme $\PS$ of $i$ such that
there is no smaller code $q$ of a minimal $P$-proof scheme $\mathbb{T}$ of $i$
such that $supp(\PS) = supp(\mathbb{T})$.  Thus if $T_P$ has a node of length
$2i+1$, then $T_P$ will not be infinitely branching.  Let us note that if $P$
has a stable model, then $T_P$ has a node of length $2i+1$ so that $T_P$ is
finitely branching if and only if $P$ has the $\FS$ property.  If $T_P$ does
not have any node of length $2i+1$, then it is easy to check that our
conditions ensure that $\{0, \ldots, i-1\}$ is an  explicit initial blocking
set for $P$. Thus $T_P$ is finitely branching if and only $P$ has the $\FS$
property or $P$ has an explicit initial blocking set. 

Let us now suppose that $P$ does not have the $a.a.$ $\FS$ property. Then there
will be infinitely many  $i$ which are codes of ground atoms of $P$ such that
there exist infinitely many supports of $P$-proof schemes with conclusion $i$.
Now suppose that there is a node $\sg =(\sg(0), \ldots \sg(2i+1))$ of length
$2i+1$ in $T_P$.  Then again, $\sg^* =(\sg(0), \ldots, \sg(2i-1),1, r)$ is a
node in $T_P$ where $r$ is any code of a minimal $P$-proof scheme $\PS$ of $i$
such that there is no smaller code $q$ of a minimal $P$-proof scheme
$\mathbb{T}$ of $i$ such that $supp(\PS) = supp(\mathbb{T})$.  Thus if $T_P$
has a node of length $2i+1$, then $T_P$ will have a node of length $2i$ which
has infinitely many successors in $T_P$.  Note that if $P$ has a stable model, 
then $T_P$ has a node of length $2i+1$ for all $i$ so that $T_P$ is nearly
bounded  if and only if $P$ has the $a.a.$ $\FS$ property.  If $T_P$ does not
have any node of length $2i+1$, then it is easy to check that our conditions 
ensure that $\{0, \ldots, i-1\}$ is an initial blocking set for $P$. Thus $T_P$
is nearly bounded if and only $P$ has the $a.a.$ $\FS$ property or $P$ has an
initial blocking set.

Next, assume that $T_P$ is finitely branching. By K\"{o}nig's lemma, either
$T_P$ is finite or $T_P$ has an infinite path.  If $T_P$ has an infinite path, 
then there will be nodes of length $2i+1$ in $T_P$ for all $i$. Hence for each
$i$, there will be nodes of the form $\sg^* =(\sg(0), \ldots, \sg(2i-1),1, r)$
in $T_P$ where $r$ is any code of a minimal $P$-proof scheme $\PS$ of $i$ such
that there is no smaller code $q$ of a minimal $P$-proof scheme $\mathbb{T}$ of
$i$ such that $supp(\PS) = supp(\mathbb{T})$. Thus if $T_P$ is highly
recursive, then for all $i$, we can find all the codes  $r$ of  minimal
$P$-proof schemes $\PS$ of $i$ such that there is no smaller code $q$ of a
minimal $P$-proof scheme $\mathbb{T}$ of $i$ such that $supp(\PS) =
supp(\mathbb{T})$ because we can compute the set of nodes of length $2i+1$ as a
function of $i$.  Thus $T_P$ is highly recursive if and only if $P$ has the rec.
$\FS$ property or $P$ has an explicit initial blocking set.  Similarly, if $P$
has a stable model, then  $T_P$ must have an infinite path so that $T_P$ is
highly recursive if and only if $P$ has the rec. $\FS$ property. 

Next, assume that $T_P$ is nearly bounded. Thus there is an $m \geq 0$ such
that each node of length greater than or equal to $m$ has only finitely many
successors in $T_P$.  If $T_P$ has nodes of length $2i$ for all $i \geq 0$,
there will be nodes of the form $\sg^* =(\sg(0), \ldots, \sg(2i-1),1, r)$  
in $T_P$ where $r$ is a code of a minimal $P$-proof scheme $\PS$ of $i$ such
that there is no smaller code $q$ of a minimal $P$-proof scheme $\mathbb{T}$ of
$i$ such that $supp(\PS) = supp(\mathbb{T})$. Hence if $2i \geq m$, then it
must be the case that there are only finitely many supports of minimal
$P$-proof schemes of the atom $a$ coded by $i$.  Clearly, if $T_P$ has an
infinite path, then there will be nodes of length $2i$ for all $i$, so that $P$
must have the $a.a.$ $\FS$ property. Similarly, if $T_P$ is nearly recursively
bounded and $T_P$ has nodes of length $2i$ for all $i$, then $P$ will have the
$a.a.$ rec. $\FS$ property.  Thus if $T_P$ is nearly bounded, then either there
will be some fixed $n$ such that $T_P$ has no nodes of length $2n$ in which
case $T_P$ has an initial blocking set or $T_P$ has nodes of length $2n$ for
all $n \geq 0$ in which  case $P$ has the $a.a.$ $\FS$ property. Similarly, 
if $T_P$ is nearly recursively bounded, then either there will be some fixed
$n$ such that $T_P$ has no nodes of length $2n$ in which case $T_P$ has an
initial blocking set or $T_P$ has nodes of length $2n$ for all $n \geq 0$ in
which case $P$ has the $a.a.$ rec. $\FS$ property.  Thus $T_P$ is nearly bounded
(nearly recursively bounded) if and only if $P$ has an  initial blocking set or
$P$ has the $a.a.$ $\FS$ property ($a.a.$ rec. $\FS$ property).  In particular,
if $P$ has a stable model, then $T_P$ is nearly bounded (nearly recursively
bounded) if and only if  $P$ has the $a.a.$ $\FS$ property ($a.a.$ rec. $\FS$
property).  Thus parts (2)-(9) of the theorem hold. 

For (10), note that if $P$ is decidable, then for any finite set of ground
atoms $\{a_1, \ldots, a_n\} \subseteq H(P)$ and any finite set of minimal
$P$-proof schemes $\{\PS_1, \ldots, \PS_n\}$ such that $concl(S_i) =a_i$, we
can effectively decide whether there is a stable model of $M$ of $P$ such
that 
\begin{compactdesc}
\item[(A1)] $a_i \in M$ and $\PS_i$ is the smallest minimal $P$-proof scheme $\PS$ 
for $a_i$ 
such that $supp(\PS) \cap M = \emptyset$; and 
\item[(A2)] for any ground atom $b \notin  \{a_1, \ldots, a_n\}$ such 
that the code of $b$ is strictly less than the maximum of the 
codes of $a_1, \ldots, a_n$, $b \notin M$. 
\end{compactdesc}
But this is precisely what we need to decide to determine whether a given node
in $T_P$ can be extended to an infinite path through $T_P$.  Thus if $P$ is
decidable, then $T_P$ is decidable. On the other hand, suppose $T_P$ is
decidable and we are given a set of atoms $\{a_1 < \ldots < a_n\} \subseteq
H(P)$ and any finite set of minimal $P$-proof schemes $\{\PS_1, \ldots,
\PS_n\}$ such that $concl(\PS_i) =a_i$. Then let $\sg =(\sg(0), \ldots,
\sg(2a_n+3))$ be such that $\sg(2a_n+2) = \sg(2a_n+3) =0$ and for $i \leq a_n$,
$\sg(2i) = \sg(2i+1) =0$ if $i \not \in \{a_1 < \ldots < a_n\}$ and $\sg(2i)
=1$ and $\sg(2i+1) = c(\PS_i)$.  Then there is an infinite path of $T_P$ that
passes through $\sg$ if and only if there is a a stable model of $M$ of $P$
such that the conditions {\bf (A1)} and {\bf (A2)} hold.
Thus $P$ is decidable if and only if $T_P$ is decidable. 

\section{Complexity of index sets for finite normal \\ predicate  
logic programs.}
\label{compl}

In this section, we shall prove our results on the complexity of index sets
associated with various properties of finite normal predicate logic programs,
finite normal predicate logic programs which have the $\FS$ property, and
finite normal predicate logic programs which have the recursive $\FS$
property.  We will sometimes call them FSP programs and rec. FSP programs,
respectively.

\begin{theorem} \label{thm:irb}
\begin{compactenum}
\item[(a)]   $\{e: Q_e\ \text{has an initial blocking set}\}$ 
and \\
$\{e: Q_e\ \text{has an explicit initial blocking set}\}$  are $\Sigma^0_2$
complete.

\item[(b)] $\{e: Q_e\ \text{has the rec. $\FS$ property}\}$ 
is $\Sigma^0_3$-complete.
\item[(c)] $\{e: Q_e\ \text{has the $\FS$ property}\}$ is $\Pi^0_3$-complete.
\item [(d)]
$\{e: Q_e\ \text{has the rec. $\FS$ property and is decidable}\}$ 
is $\Sigma^0_3$-complete.
\end{compactenum}
\end{theorem}
\begin{proof}
In each case, it easy to see that the index set is of the required 
complexity by simply writing out the definition. 

Let $A= \{e: Q_e\ \text{has an explicit initial blocking set}\}$ and 
$Fin$ is the set $\{e:W_e \ \mbox{is} \ \mbox{finite}\}$.  We know that $Fin$
is is $\Sigma^0_2$-complete, \cite{Soare}. Thus to show that $A$ is
$\Sigma^0_2$-complete, we need only to show that $Fin$ is many-one reducible to
$A$. Recall that $W_{e,s}$ is the set of all elements $x$ less than or equal to
$s$ such that $\phi_e(x)$ converges $s$ or fewer steps. It follows that for any
$e$, $N_e = \{s: W_{e,s}-W_{e,s-1} \neq \emptyset\}$ and the set $S_e$ of all 
codes of pairs $(x,y)$ such that $x,y \in N_e$, $x < y$, and there is no $z \in
N_e$ such that $x < z <y$, are recursive sets. Then by Proposition \ref{aux}, we 
can uniformly construct a finite normal predicate logic  Horn program $P_e^{-}$
whose set of atoms is $\{s^n(0): n\geq 0\}$ and which contains two predicates
$N(x)$ and $S(x,y)$ such that $N(s^x(0))$ holds if and only if $x \in N_e$ and 
$S(s^x(0),s^y(0))$ holds if and only if $[x,y] \in S_e$.  Let $E$ be a unary
predicate symbol that does not appear in $P_e^{-}$. Then we let $P_e$ be the
finite normal predicate logic program that consists of 
$P_e^{-}$ and the following two predicate logic clauses:
\begin{compactdesc}
\item[(a)] $E(x) \leftarrow N(x),\neg E(x)$  and 
\item[(b)] $E(x) \leftarrow N(y),S(x,y)$.
\end{compactdesc}
The clauses in {\bf (a)} and {\bf (b)} generate, when grounded, the following 
clauses in $ground(P_e)$:
\begin{compactdesc}
\item[(A)] $E(s^n(0)) \leftarrow N(s^n(0)),\neg E(s^n(0))$ for all $n \geq 0$;
and 
\item[(B)] $E(s^m(0)) \leftarrow N(s^n(0)),S(s^m(0),s^n(0))$ for all 
$m,n \geq 0$.
\end{compactdesc}
Now suppose that $W_e$ is infinite and $N_e = \{n_0 < n_1 < \ldots \}$.  Then
we claim that $P_e$ has a stable model $M_e$ which consists of 
the least model of $P_e^{-}$ plus $\{E(s^{n_i}(0)): i \geq 0\}$.  That is, the
presence of $N(s^n(0))$ in the body of the clauses in ({\bf A}) and the
presence of $N(s^n(0))$ and $S(s^m(0),s^n(0))$ in the body of the clauses in
({\bf B}) ensures that the only atoms of the form $E(a)$ that can possibly be
in any stable model of $P_e$ are of the form $E(s^n(0))$ where $n \in N_e$.
But if $W_e$ is infinite, then the Horn clauses of type ({\bf B}) ensure that 
$\{s^{n_i}(0): i \geq 0\}$ will be in every stable model of $P_e$. This, in 
turn, means  that none of the clauses of type ({\bf A}) for $n \in N_e$ will
contribute to the Gelfond-Lifschitz reduct $(P_e)_{M_e}$. It follows that
$(P_e)_{M_e}$ consists of $P_e^{-}$ plus all the clauses in ({\bf B}) plus all
the clauses of the form $s^n(0) \leftarrow N(s^n(0))$ such that $n \notin N_e$. 
It is then easy to see that $M_e$ is the least model of $(P_e)_{M_e}$ 
so that $M_e$ is a stable model of $P_e$. Thus if $W_e$ is infinite, then 
$P_e$ does not have an explicit initial blocking set. 

Next, suppose that $W_e$ is finite.  Then $N_e$ is finite, say  $N_e = \{n_0 <
\ldots < n_r\}$. Then we will not be able to use a clause of type ({\bf B}) to
derive $E(s^{n_r}(0))$.  Thus the only clause that could possibly derive
$E(s^{n_r}(0))$ would be the clause $$C= E(s^{n_r}(0)) \leftarrow
N(s^{n_r}(0)),\neg E(s^{n_r}(0)).$$ But then there can be no stable model $M$
of $P_e$. That is, if $s^{n_r}(0) \in M$,  then clause $C$ will not be in
$(P_e)_M$ so that there will be no way to derive $E(s^{n_r}(0))$ from
$(P_e)_M$.  On the other hand, if $E(s^{n_r}(0)) \notin M$,  then clause $C$
will contribute the clause $E(s^{n_r}(0)) \leftarrow N(s^{n_r}(0))$ to
$(P_e)_M$ so that $E(s^{n_r}(0))$ will be in the least model of $(P_e)_M$.  It
follows that $\{E(0),E(s(0)), \ldots, E(s^{n_r}(0))\}$ together with all that
atoms of $P_e^{-}$ whose code is less than the code of $E(s^{n_r}(0))$ will be
an explicit initial blocking set for $P_e$.

Thus we have shown that $P_e$ has an explicit initial blocking set if and only
if $W_e$ is finite.  Hence, the recursive function $f$ such that $Q_{f(e)}
=P_e$ shows that $\mathit{Fin}$ is many-one reducible to $A$ and, hence, $A$ is
$\Sigma^0_2$-complete.  The same proof will show that 
$B= \{e: Q_e\ \text{has an initial blocking set}\}$ is $\Sigma^0_2$-complete.

We claim that the completeness of the remaining parts of the theorem 
are all consequences of Theorem \ref{tree2prog}.  That is, recall 
that $T_0,T_1, \ldots $ is an effective list of all primitive recursive 
trees. Then let $g$ be the recursive function such that $Q_{g(e)} = P_{T_e}$
where $P_{T_e}$ is the finite normal predicate logic program constructed from 
$T_e$ as in the proof of Theorem \ref{tree2prog}. Then $g$ shows that 
\begin{compactenum}
\item $\{e:T_e \ \mbox{is r.b.}\}$ is many-one reducible to 
$\{e:Q_e \ \mbox{has the rec. $\FS$ property}\}$; 
\item $\{e:T_e \ \mbox{is bounded}\}$ is many-one reducible to 
$\{e:Q_e \ \mbox{has the $\FS$ property}\}$;
\item $\{e:T_e \ \mbox{is r.b. and decidable}\}$ is many-one reducible to 
$\{e:Q_e \ \mbox{has the rec.}$ $\FS$ property and is decidable$\}$.
\end{compactenum}

Hence the completeness results for parts (b), (c), and (d) immediately 
follow from our completeness results 
for $\{e:T_e \ \mbox{is r.b.}\}$, $\{e:T_e \ \mbox{is bounded}\}$, 
and  $\{e:T_e \ \mbox{is r.b. and decidable}\}$ given in Section \ref{classes}. 
\end{proof}

It is not always the case that the complexity results for finite normal
predicate logic programs match the corresponding complexity for trees.  For
example, K\"{o}nig's Lemma tells us that an infinite finitely branching tree
must have an infinite path through it. It follows that $[T] = \emptyset$ holds
for a primitive recursive finitely branching tree $T$ if and only if $T$ is
finite.  This means the properties that $T$ is bounded and empty and $T$ is
recursively bounded and empty are $\Sigma^0_2$ properties since $T$ being
finite is a $\Sigma^0_2$ predicate for primitive recursive trees. 
K\"{o}nig's Lemma is a form of the Compactness Theorem for propositional 
logic which, we have observed, fails for normal propositional logic programs. 
Indeed, given any finite normal predicate logic program $Q_e$, we can simply
take an atom $a$ which does not occur in $ground(Q_e)$ and add the clause $C =
a \leftarrow \neg a$. Then the program $Q_e \cup \{C\}$ does not have a stable
model but will have the $\FS$ property if and only if $Q_e$ has the $\FS$
property and will have the rec. $\FS$ property if and only if  $Q_e$ has the
rec. $\FS$ property.  Thus there is a recursive function $h$ such that 

\begin{compactenum}
\item $Q_{h(e)}$ does not have stable model, 
\item $Q_e$ has the $\FS$ property if and only if $Q_{h(e)}$ has the $\FS$
property, and 
\item $Q_e$ has the rec. $\FS$ property if and only if $Q_{h(e)}$ has the rec.
$\FS$ property.  
\end{compactenum}

It follows that 
$\{e: Q_e \ \mbox{has the $\FS$ property}\}$ is many-one reducible to 
$\{e: Q_e \ \mbox{has the $\FS$ property}$ and $\mathit{Stab}(Q_e)$ $=
\emptyset\}$ and $\{e: Q_e \ \mbox{has the rec. $\FS$ pro}$-perty$\}$ is
many-one reducible to 
$\{e: Q_e \ \mbox{has the rec. $\FS$ property and}\  \mathit{Stab}(Q_e) =
\emptyset\}$.  Thus it follows from  Theorem \ref{thm:pirb} that   

\begin{compactenum} 
\item $\{e: Q_e \ \mbox{has the $\FS$ property and 
$\mathit{Stab}(Q_e) = \emptyset$} \}$ is $\Pi^0_3$-complete 
and
\item $\{e: Q_e \ \mbox{has the rec. $\FS$ property and}\ \mathit{Stab}(Q_e) =
\emptyset \}$ is $\Sigma^0_3$-complete.
\end{compactenum}

To see that $\{e: Q_e \ \mbox{has the rec. $\FS$ property and} \
\mathit{Stab}(Q_e) = \emptyset \}$ is
$\Pi^0_3$, we can appeal to Theorem \ref{prog2trees} which constructs 
a finitely branching tree $T_{Q_e}$ such that there is a one-to-one effective 
degree preserving correspondence between the stable models of $Q_e$ and
$[T_{Q_e}]$.  It follows that $Q_e$ has the $\FS$ property and no stable models
if and only if $Q_e$ has the $\FS$ property and $T_{Q_e}$ is finite.  This
latter predicate is a  $\Pi^0_3$ predicate because $Q_e$ having the $\FS$
property is $\Pi^0_3$ predicate and $T_{Q_e}$ being finite is a $\Sigma^0_2$
predicate. Similarly, $Q_e$ has the rec. $\FS$ property and has no stable
models if and only if $Q_e$ has the rec. $\FS$ property and $T_{Q_e}$ is finite
which is a  $\Sigma^0_3$ predicate because $Q_e$ having the rec. $\FS$ property
is $\Sigma^0_3$ predicate and $T_{Q_e}$ being finite is a $\Sigma^0_2$
predicate. Thus we have proved the following theorem. 

\begin{theorem}\label{thm:irbb}
\begin{compactenum}
\item[(a)] $\{e: Q_e \ \text{has the rec. $\FS$ property and } \mathit{Stab}(Q_e) =
\emptyset \}$ is
$\Sigma^0_3$-complete; and  
\item[(b)] $\{e: Q_e \ \text{has the $\FS$ property and} \ \mathit{Stab}(Q_e) =
\emptyset \}$ is $\Pi^0_3$-complete.  
\end{compactenum}
\end{theorem}

The method of proof for parts (b), (c), and (d) in Theorem \ref{thm:irb} can be
used to prove many  results about properties of stable models of finite normal
predicate logic programs $Q_e$ where $Stab(Q_e)$ is not empty.  That is, one
can prove that the desired index set is in the proper complexity class by
simply writing out the definition or by using Theorem \ref{prog2trees}. For
example, Theorem 2.6 (b) says that $\{e: T_e \ \text{is $r.b.$ and $[T_e] \neq
\emptyset$}\}$ is $\Sigma^0_3$-complete.  We claim that this theorem
immediately implies that $\{e: Q_e \ \text{has the rec. $\FS$ property  and }
\mathit{Stab}(Q_e) \neq \emptyset\}$ is also  $\Sigma^0_3$-complete. First we
claim that the fact that $\{e: Q_e$ has the rec. $\FS$ property and
$\mathit{Stab}(Q_e) \neq \emptyset \}$ is  $\Sigma^0_3$ follows from Theorem
\ref{prog2trees}.  That is, by Theorem \ref{prog2trees} $Q_e$ has the rec.
$\FS$ property and $\mathit{Stab}(Q_e) \neq \emptyset$ if and only if $T_{Q_e}$
is $r.b.$ and $[T_{Q_e}]$ is nonempty. But this latter question is $\Sigma^0_3$
question so the former question is a $\Sigma^0_3$ question. Thus Theorem
\ref{prog2trees} allows us to reduce complexity bounds  about finite normal
predicate logic programs $P$ which have stable models to complexity bounds of
their corresponding trees $T_P$ where $[T_P]$ is nonempty.  Then we can then
use Theorem \ref{tree2prog} and the theorems on index sets for trees given in
Section \ref{classes} to establish the necessary completeness results. For
example,  to show that $\{e: Q_e \ \text{has the rec. $\FS$ property  and 
$\mathit{Stab}(Q_e) \neq \emptyset$} \}$ is $\Sigma^0_3$-complete,  we use
Theorem \ref{tree2prog} and the fact that $\{e:
T_e \ \text{is $r.b.$ and $[T_e]$ is nonempty}\}$ is $\Sigma^0_3$-complete. 
That is, it follows from the proof of  Theorem \ref{tree2prog} that 
there is a recursive function $f$ such that 
$Q_{f(e)} = P_{T_e}$. Hence 
\begin{multline*}
e \in \{h: T_h \ \text{is $r.b.$ and $[T_h]$ is nonempty}\}
\iff  \\
f(e) \in \{g: Q_g \ \text{has the rec. $\FS$ property  and}\ \mathit{Stab}(Q_g)
\neq \emptyset\}.
\end{multline*}
Thus $\{e: Q_e \ \text{has the rec. $\FS$ property  and $\mathit{Stab}(Q_e) \neq
\emptyset$}\}$ is $\Sigma^0_3$-complete. 

One can use the same techniques to prove that the  following theorem  
follows from the corresponding index sets results on 
trees given in Section \ref{classes}. 

\begin{theorem} \label{thm:ire} 
\begin{compactenum}
\item[(a)] $\{e: Q_e \ \text{has the  rec. $\FS$ property and
$\mathit{Stab}(Q_e) \neq \emptyset$} \}$
is \\ 
$\Sigma^0_3$-complete.
\item[(b)] $\{e: Q_e \ \text{has the $\FS$ property  and 
$\mathit{Stab}(Q_e) \neq \emptyset$} \}$ is $\Pi^0_3$-complete.   
\item[(c)] $\{e: \mathit{Stab}(Q_e) \neq \emptyset \}$ is
$\Sigma_1^1$-complete. 
\end{compactenum}
\end{theorem}

Since  $\{e: Q_e \ \text{$\mathit{Stab}(Q_e)$ is empty} \}$ is the complement 
of the $\Sigma^1_1$-complete set $\{e:  \text{$\mathit{Stab}(Q_e)
\neq \emptyset$} 
\}$, we have the following corollary. 
\begin{corollary} 
$\{e: \ \text{$\mathit{Stab}(Q_e) = \emptyset$} \}$ is $\Pi_1^1$-complete. 
\end{corollary}

Next we want to consider the properties of $Stab(Q_e)$ being infinite or
finite. 

\begin{theorem} \label{thm:ici}
\begin{compactenum}
\item[(a)] $\{e: Q_e\ \text{has the rec. $\FS$ property and $Stab(Q_e)$ is
infinite}\}$ is $D^0_3$-complete 
and $\{e: Q_e\ \text{has the rec. $\FS$ property and $Stab(Q_e)$ is
fini}$-{\em te}$\}$ is $\Sigma^0_3$-complete.  
\item[(b)] $\{e: Q_e\ \text{has the $\FS$ property and $Stab(Q_e)$
is infinite}\}$ is $\Pi^0_4$-complete and \\
$\{e: Q_e\
\text{has the $\FS$ property and $Stab(Q_e)$ finite}\})$ is 
$\Sigma^0_4$-complete. 
\item[(c)] $\{e: Stab(Q_e)\ \text{is infinite}\}$ is 
$\Sigma_1^1$-complete. 
$\{e: Stab(Q_e)\ \text{is
finite}\})$ is $\Pi_1^1$-complete.
\end{compactenum}
\end{theorem}
\begin{proof}
To prove the upper bounds in each case, we do the following. Given a finite
normal predicate logic program $Q_e$, let $a$ and $\bar{a}$ be two atoms which
do not occur in $ground(P)$. Then let $R_e$ be the finite normal predicate
logic program which arises from $Q_e$ by adding $a$ to body of every clause in
$Q_e$ and adding the following two clauses:\\
$C_1 = a \leftarrow \neg \bar{a}$ and \\
$C_2 = \bar{a} \leftarrow \neg a$.\\
Then we claim that exactly one of $a$ or $\bar{a}$ must be in every stable
model $M$ of $R_e$. That is, if neither $a$ or $\bar{a}$ are in $M$, then $C_1$
and $C_2$ will contribute $a \leftarrow$ and $\bar{a} \leftarrow$ to $(R_e)_M$
so that both $a$ and $\bar{a}$ will be in the least model of  $(R_e)_M$. 
If both $a$ and $\bar{a}$ are in $M$, then $C_1$ and $C_2$ will contribute
nothing to  $(R_e)_M$ so that neither $a$ nor $\bar{a}$ will be in the least
model of  $(R_e)_M$ since then there will be no clauses of  $(R_e)_M$ with
either $a$ or $\bar{a}$ in the head of the clause. It follows that $R_e$ will
have two types of stable models $M$, namely $M = \{\bar{a}\}$ or $M = M^* \cup
\{a\}$ where $M^*$ is stable model of $Q_e$. The modified program $R_e$ is
guaranteed to have a finite stable model, and in particular $\mathit{Stab}(R_e)
\neq \emptyset$.  Because of the form of stable models of $R_e$, $Stab(Q_e)$ is
finite if and only if $Stab(R_e)$ is finite. Clearly $Q_e$ has the $\FS$ (rec.
$\FS$) property if and only if $R_e$ has the $\FS$ (rec. $\FS$) property.  
By Theorem \ref{prog2trees}, there is a recursive function $g$ such that
$T_{g(e)} = T_{R_e}$ as constructed in the proof of Theorem \ref{prog2trees}.
Then we know $Stab(R_e)$ is finite if and only if $[T_{g(e)}]$ is finite and 
$R_e$ has the $\FS$ (rec. $\FS$) property if and only if $T_{g(e)}$ is
recursively bounded. Then for example, it follows that $\{e:Q_e $\ has the
$\FS$ property and $\mathit{Stab}(Q_e)$ is finite$\}$ is many-one reducible to
$\{e:T_e \ \mbox{is $r.b$ and $[T_e]$ is finite}\}$ which is $\Sigma^0_3$. In
this way, the complexity bounds follows from the complexity bounds in Theorem
\ref{thm:pici}.

To establish the completeness results in each case, we can proceed as follows.
We can use the construction of Theorem \ref{tree2prog} to construct a finite
normal predicate logic program $P_{T_e}$ such that $[T_e]$ is finite if and
only if $Stab(P_{T_e})$ is finite and $T_e$ is bounded ($r.b.$) if and only if
$P_{T_e}$ has the $\FS$ (rec. $\FS$) property. Thus there is a recursive
function $f$ such that $Q_{f(e)} = P_{T_e}$.   Then, for example, it follows
that $f$ shows that $\{e:T_e \ \mbox{is $r.b$  and $[T_e]$ is finite}\}$ is 
many-one reducible to $\{e:Q_e$ has the $\FS$ property and 
$\mathit{Stab}(Q_e)$ is finite$\}$.  Thus $\{e:Q_e$ has the $\FS$ property and 
$\mathit{Stab}(Q_e)$ is finite$\}$ is $\Sigma^0_3$-complete 
since $\{e:T_e \ \mbox{is $r.b$  and $[T_e]$ is finite}\}$ is 
$\Sigma^0_3$-complete.  In this way, we can use completeness results of Theorem
\ref{thm:pici} to establish the completeness of each part of the theorem. 
\end{proof}

By combining the completeness results of Theorem 2.9 with Theorems
\ref{tree2prog}, and \ref{prog2trees}, we can use the same method of proof to
prove the following theorem. 

\begin{theorem} \label{thm:icu} 
$\{e: Stab(Q_e)\ \text{is uncountable}\}$ is $\Sigma^1_1$-complete and \\ 
$\{e: Stab(Q_e)\ \text{is countable}\}$ and $\{e: Q_e\ \text{is countable
infinite}\}$ are $\Pi^1_1$-complete. \\
The same results hold for rec. FSP and FSP programs.   
\end{theorem}

\begin{theorem} \label{thm:irc}
For every positive integer $c$,
\begin{compactenum} 
\item[(a)] $\{e: Q_e$ has the rec. $\FS$ property and $Card(Stab(Q_e)) > c\}$,\\
$\{e: Q_e$ has the rec. $\FS$ property, and 
$Card(Stab(Q_e)) \leq c\}$, and \\
 $\{e: Q_e$ has the rec. $\FS$ property and $Card(Stab(Q_e)) = c\}$ 
are all \\
$\Sigma^0_3$-complete.
\item[(b)]
$\{e: Q_e$ has the $\FS$ property and $Card(Stab(Q_e)) \leq c\}$ and \\
$\{e: Q_e$ has the $\FS$ property and $Card(Stab(Q_e)) = 1\}$
are both  \\
$\Pi^0_3$-complete.
\item[(c)]
$\{e: Q_e$ has the $\FS$ property and $Card(Stab(Q_e)) > c\}$
and \\
$\{e: Q_e$ has the $\FS$ property and $Card(Stab(Q_e)) = c+1\}$
are both \\
$D^0_3$-complete.
\item[(d)] 
$\{e: Q_e$ has the rec. $\FS$ property, is decidable, and $Card(Stab(Q_e)) >
c\}$, $\{e: Q_e$ has the rec. $\FS$ property, is decidable,  and
$Card(Stab(Q_e)) \leq c\}$, and  $\{e: Q_e$ has the rec. $\FS$ property, is
decidable,  and $Card(Stab(Q_e))$ $= c\}$ are all $\Sigma^0_3$-complete.
\item[(e)]
$\{e:Card(Stab(Q_e)) > c\}$ is $\Sigma_1^1$-complete, while
$\{e:Card(Stab(Q_e)) \leq c\}$ and $\{e:Card(Stab(Q_e)) = c\}$
are both $\Pi^1_1$-complete.
\end{compactenum} 
\end{theorem}

\begin{proof}
The proofs for this theorem are divided into two cases.  For the cases  where
we are trying to establish the complexity results for properties where
$Card(Stab(Q_e))=c$ or $Card(Stab(Q_e)) \geq c$, we can directly use Theorems
\ref{prog2trees} and \ref{tree2prog}. For example, Theorem \ref{prog2trees}
says that $Q_e$ has $c$ ($> c$) stable models if and only if the tree $T_{Q_e}$
constructed in the proof of Theorem \ref{prog2trees} has $c$ ($> c$) infinite
paths. Moreover, $Q_e$ has the $\FS$ (rec. $\FS$) property if and only if
$T_{Q_e}$ has the $\FS$ (rec. $\FS$) property.  Let $f_1$ be recursive function 
such that $T_{f_1(e)} = T_{Q_e}$. Then, for example,  $f_1$ shows that 
\begin{multline*}
 A=\{e: Q_e \ \text{has the rec. $\FS$ property and is decidable and}\\
 Card(Stab(Q_e)) = c\} 
\end{multline*}
 is many-one reducible to 
\[
B=\{h: T_h \ \text{is $r.b$. and is decidable and}\ Card(Stab(Q_h)) = c\}.
\]
By Theorem \ref{thm:pirc}, we know that $B$ is  is $\Sigma^0_3$ so that $A$ is
$\Sigma^0_3$.  Thus we can reduce the problem of the complexity bounds for the
properties involving $Stab(Q_e) = c$, and $Stab(Q_e) > c$ to the corresponding
properties of trees that $[T_e] =c$ and $[T_e] >c$ that appear in Theorem
\ref{thm:pirc}. 

To establish completeness in each case, we can use Theorem \ref{tree2prog}.
That is, there is a recursive function $f_2$ such that $Q_{f_2(e)} = P_{T_e}$
as constructed in Theorem \ref{tree2prog}.  Then, for example, $f_2$ shows that
\[
C=\{e: P_e \ \text{is $r.b.$ and is decidable and}\ Card([T_e]) = c\}
\]
is many-one reducible to 
\begin{multline*}
D=\{e: Q_e \ \text{has the rec. $\FS$ property and is decidable} \\
\text{and}\ Card(Stab(Q_e)) = c\}. 
\end{multline*}
We know by Theorem \ref{thm:pirc} that  $C$ is $\Sigma^0_3$-complete so that
$D$ is complete for  $\Sigma^0_3$ sets. Thus it follows that   $\{e: Q_e$ has
the rec. $\FS$ property and is decidable and $Card(Stab(Q_e)) = c\}$ is
$\Sigma^0_3$-complete. 

One has to be a bit more careful for the properties that involve the condition
that $Stab(Q_e) \leq c$.   In this case, we can use the techniques of the proof
of Theorem \ref{thm:ici} so we shall use the same notation and definitions as
in the proof of Theorem \ref{thm:ici}.  That is, it is easy to see that 
$Stab(Q_e) \leq c$ if and only if $Stab(R_e) \leq c+1$ and that $Stab(R_e) \leq
c+1$ if and only if $[T_{R_e}] \leq c+1$. But since $Stab(R_e) \neq \emptyset$
by construction, we see that $Q_e$ has the $\FS$ (rec. $\FS$) property if and
only if $R_e$ has the $\FS$ (rec. $\FS$) property if and only if $T_{r_e}$ is
bounded ($r.b.$). Let $g$ be the recursive function such that $T_{g(e)} =
T_{R_e}$. Then, for example, $g$ shows that 
\[
E =\{e: Q_e \ \text{has the $\FS$ property and}\ Card(Stab(Q_e)) \leq c\} 
\]
is many-one reducible to 
\[
F= \{e:T_e \ \mbox{is bounded and} \ [T_e] \leq c+1 \} 
\]
which is $\Pi^0_3$ by Theorem \ref{thm:pirc}. Thus,  $E$ is $\Pi^0_3$.  All the
other complexity bounds that involve the property $Stab(Q_e) \leq c$ can be
proved in a similar manner. 

To establish the corresponding completeness results, we observe that $[T_e]\leq
c$ if and only if $\mathit{Card}(\mathit{Stab}(P_{T_e})) \leq c$ and $T_e$ is
bounded ($r.b.$)  if and only if $P_{T_e}$ has the $\FS$ (rec. $\FS$) property.
Let $h$ be the recursive function such that $Q_{h(e)} = P_{T_e}$. Then $h$
shows that $\{e:T_e \ \mbox{is bounded and} \ [T_e] \leq c\}$ is many-one
reducible to $\{e: Q_e \ \text{has the $\FS$ property and}\ Card(Stab(Q_e))
\leq c\}$. Thus $\{e: Q_e \ \text{has the $\FS$ property and}\ Card(Stab(Q_e))
\leq c\}$ is $\Pi^0_3$-complete. All the other completeness results  that
involve the property $Stab(Q_e) \leq c$ can be proved in a similar manner. 
\end{proof}

Next, we give some index set results concerning the number of recursive stable
models of a finite normal predicate logic program $Q_e$. Here we say that
$Stab(Q_e)$ is {\em recursively empty} if $Stab(Q_e)$ has no recursive elements
and is  {\em recursively nonempty} if $Stab(Q_e)$ has at least one recursive
element. Similarly, we say that a $Stab(Q_e)$ has  {\em recursive cardinality
equal to $c$} if $Stab(Q_e)$ has exactly $c$ recursive members. 

\begin{theorem} \label{thm:rbce}
\begin{compactenum}
\item[(a)] $\{e: Q_e$ has the rec. $\FS$ property and $Stab(Q_e)$ is
recursively nonempty$\}$ is $\Sigma^0_3$-complete, $\{e: Q_e$ has the rec.
$\FS$ property and $Stab(Q_e)$ is recursively empty$\}$ is $D^0_3$-complete,
and $\{e: Q_e$ has the rec. $\FS$ property and $Stab(Q_e)$ is nonempty and
recursively empty$\}$ is $D^0_3$-complete.
\item[(b)] $\{e: Q_e$ has the $\FS$ property and $Stab(Q_e)$ is recursively
nonempty$\}$ is $D^0_3$-complete, $\{e: Q_e$ has the $\FS$ property and
$Stab(Q_e)$ is recursively empty$\}$ is $\Pi^0_3$-complete, and 
$\{e: Q_e$ has the $\FS$ property and $Stab(Q_e)\neq \emptyset$ 
and recursively empty$\}$ is $\Pi^0_3$-complete.
\item[(c)] $\{e: Stab(Q_e)$ is recursively nonempty$\}$ 
is $\Sigma^0_3$-complete, $\{e: Stab(Q_e)$ is recursively empty$\}$ is
$\Pi^0_3$-complete and $\{e: Stab(Q_e) \neq \emptyset$ and recursively
empty$\}$ is $\Sigma_1^1$-complete.
\end{compactenum}
\end{theorem}
\begin{proof}

We say that a finite normal predicate logic program $Q_e$ has an {\em isolated
stable model} $M$, if there is a finite set of ground atoms $a_1, \ldots, a_n,
b_1, \ldots, b_m$ such that $a_i \in M$ for all $i$ and $b_j \notin M$ for all
$j$ and there is no other stable model $M'$ such $a_i \in M'$ for all $i$ and
$b_j \notin M'$ for all $j$. Thus isolated stable models are determined by a
finite amount of positive and negative information.  We say that a finite
predicate logic program $Q_e$  is {\em perfect} if $Stab(Q_e)$ is nonempty 
and it has no isolated elements.  

To prove the upper bounds in each case, we do the 
following. 
Jockusch and Soare \cite{JS72a} constructed a recursively bounded  primitive
recursive tree such that $[T] \neq \emptyset$ and $[T_e]$ has no recursive
elements. It then follows from Theorem \ref{thm:basis1} that $[T]$ can have no
isolated elements so that $[T]$ is perfect. By Theorem \ref{tree2prog}, there
is a finite normal predicate logic program $U$ such that $U$ has the rec. $\FS$
property and there is a one-to-one degree preserving correspondence between
$[T]$ and $Stab(U)$.  Thus $Stab(U)$ has no recursive or isolated elements. Now 
suppose that we are given a finite normal predicate logic program $Q_e$. Then
make a copy $V$ of the finite normal predicate logic program $U$ such that $V$
has no predicates which are in common with $Q_e$. Let $a$ and $\bar{a}$ be two
atoms which do not appear in either $V$ or $Q_e$ and let $S_e$ be the finite
normal predicate logic program which arises from $U$ and $Q_e$ by adding $a$ to
the body of every clause in $Q_e$, adding $\bar{a}$ to  the body of every
clause in $V$, and adding the following two clauses:\\
$C_1 = a \leftarrow \neg \bar{a}$ and \\
$C_2 = \bar{a} \leftarrow \neg a$.\\
Then, as before, we claim that exactly one of $a$ or $\bar{a}$ must be in every 
stable model $M$ of $S_e$. That is, if neither $a$ or $\bar{a}$ are in $M$,
then $C_1$ and $C_2$ will contribute $a \leftarrow$ and $\bar{a} \leftarrow$ to
$(S_e)_M$ so that both $a$ and $\bar{a}$ will be in the least model of
$(S_e)_M$.  If both $a$ and $\bar{a}$ are in $M$, then $C_1$ and $C_2$ will
contribute nothing to  $(S_e)_M$ so that neither $a$ nor $\bar{a}$ will be in
the least model of  $(S_e)_M$ since then there will be no clauses of  $(S_e)_M$
with either $a$ or $\bar{a}$ in the head of the clause. It follows that $S_e$
will have two types of stable models $M$, namely $M = M_1 \cup \{\bar{a}\}$ or
$M = M_2 \cup \{a\}$ where $M_1$ is stable model of $V$ and $M_2$ is stable
model of $Q_e$.  Since $V$ has the rec. $\FS$ property, is perfect, and 
has no recursive stable models, it follows that  
\begin{compactenum}
\item $Q_e$ has the rec. $\FS$ ($\FS$) property if and only if 
$S_e$ has the rec. $\FS$ ($\FS$) property, 
\item $Q_e$ is perfect if and only if $S_e$ is perfect,  and 
\item the only recursive stable models of $S_e$ are of the form 
$M \cup \{a\}$ where $M$ is a recursive stable model of 
$Q_e$. 
\end{compactenum}

By Theorem \ref{prog2trees}, there is a recursive function $k$ such that
$T_{k(e)} = T_{S_e}$ as constructed in the proof of Theorem \ref{prog2trees}
such that $T_{k(e)}$ is bounded ($r.b.$) if and only if $S_e$ has the $\FS$
(rec. $\FS$) property and there is an effective one-to-one degree preserving
correspondence between $Stab(S_e)$ and $[T_{k(e)}]$.  It follows that
$T_{k(e)}$ is bounded ($r.b.$) if and only if $Q_e$ has the $\FS$ (rec. $\FS$)
property and there is an effective one-to-one degree preserving correspondence
between the recursive elements of $Stab(S_e)$ and the recursive elements of
$[T_{k(e)}]$.

Then for example, it follows that $\{e:Q_e$ has the rec. $\FS$ property and is
recursively empty$\}$ is many-one reducible to $\{e:T_e$ is $r.b$ and $[T_e]$
is recursively empty$\}$ which is $D^0_3$. Thus 
\[
\{e:Q_e \ \mbox{has the rec. $\FS$ property and is recursively empty}\}
\]
is $D^0_3$.  In this way, the upper bounds on the complexity of each index set
in the theorem follow from the corresponding complexity bound of the
corresponding property of trees in Theorem \ref{thm:prbce}. 

The completeness results for each part of the theorem follow 
from Theorem \ref{tree2prog} and the corresponding completeness 
results in Theorem \ref{thm:prbce} as before. 
\end{proof}

The same method of proof can be used to prove the following 
theorems. 

\begin{theorem} \label{thm:bcc} Let $c$ be a positive integer.
\begin{compactenum}
\item[(a)] $\{e: Q_e$ has the rec. $\FS$ property and $Stab(Q_e)$ has
recursive cardinality $> c\}$ is $\Sigma^0_3$-complete, $\{e: Q_e$ has the rec.
$\FS$ property and $Stab(Q_e)$ has recursive cardinality $\leq c\}$ is
$D^0_3$-complete, and $\{e: Q_e$ has the rec. $\FS$ property and $Stab(Q_e)$
has  recursive cardinality $=c \}$ is $D^0_3$-complete. 
\item[(b)] $\{e: Q_e$ has the $\FS$ property and $Stab(Q_e)$ has recursive 
  cardinality $> c\}$ is $\Pi^0_3$-complete, 
$\{e: Q_e$ has the $\FS$ property and $Stab(Q_e)$ has recursive cardinality
 $\leq c\}$ is $D^0_3$-complete, and 
$\{e: Q_e$ has the $\FS$ property and $Stab(Q_e)$ has recursive cardinality
 $=c \}$ is $D^0_3$-complete. 
\item[(c)] $\{e:$ $Stab(Q_e)$ has recursive cardinality $>c\}$ is 
$\Sigma^0_3$-complete, 
$\{e:$ $Stab(Q_e)$ has recursive cardinality $\leq c\}$ is
$\Pi^0_3$-complete, and 
 $\{e:$ $Stab(Q_e)$ has recursive cardinality $= c\}$ 
is $D^0_3$-complete.  
\item[(d)] $\{e: Q_e$ is decidable and has the rec. $\FS$ property and 
$Stab(Q_e)$ has recursive cardinality $> c\}$ is $\Sigma^0_3$-complete, 
$\{e: Q_e$ is decidable and has the rec. $\FS$ property and $Stab(Q_e)$ has
recursive cardinality $\leq c\}$ is $D^0_3$-complete, and  
$\{e: Q_e$ is decidable and has the rec. $\FS$ property and 
$Stab(Q_e)$ has  recursive cardinality $=c \}$ is $D^0_3$-complete. 
\end{compactenum}
\end{theorem}

\begin{theorem} $\{e: \text{$Stab(Q_e)$ has finite recursive
    cardinality}\}$ is $\Sigma^0_4$-complete and  
$\{e: \text{$Stab(Q_e)$ has infinite recursive
    cardinality}\}$ is $\Pi^0_4$-complete.  
The same results are  true for programs which 
have the rec. $\FS$ property  and the $\FS$ property. 
\end{theorem}

\begin{theorem} \label{pthm:ips} 
\begin{compactenum} 
\item[(a)]
$\{e: Q_e\ \text{has the rec. $\FS$ property and $Stab(Q_e)$ is
perf}$-$\text{ect}\}$ is $D^0_3$-complete.
\item[(b)]
$\{e: Q_e\ \text{has the $\FS$ property and $Stab(Q_e)$ is perfect}\}$
is $\Pi^0_4$-complete. 
\item[(c)]
$\{e: Stab(Q_e)\ \text{is perfect}\}$
is $\Sigma^1_1$-complete.
\end{compactenum} 
\end{theorem}

\section{Index set results for $a.a.$ FSP and  $a.a.$ rec. FSP
programs.}\label{aa}

In this section, we shall use our results from the previous section 
to prove results about index sets of $a.a.$ FSP and  $a.a$ rec. FSP programs.
Recall Section \ref{intro}, discussion after Proposition \ref{p.req}) that a
finite predicate logic program $P$ has the {\it almost always finite support
($\afs$) property} if for all but finitely many atoms $a \in H(P)$, there are
only finitely many inclusion-minimal supports of minimal $P$-proof schemes for
$a$. 

First we shall prove index set results for finite normal predicate logic
programs which have the $a.a.$ rec. $\FS$ property.

\begin{theorem}\label{thm:aarFSP}
\begin{compactenum}
\item[(a)] $\{e: Q_e\ \text{has the $a.a.$ rec. $\FS$ property}\}$ is
$\Sigma^0_3$-complete.

\item[(b)] $\{e: Q_e$ has the $a.a.$ rec. $\FS$ property and 
$\mathit{Stab}(Q_e)$ is nonempty$\}$ and $\{e: Q_e$ has the $a.a.$ rec. $\FS$
property and $\mathit{Stab}(Q_e)$ is empty$\}$ are $\Sigma^0_3$-complete.

\item[(c)] $\{e: Q_e \ \text{has the $a.a.$ rec. $\FS$ property and}\
Card(Stab(Q_e)) > c\}$, $\{e: Q_e$ has the $a.a.$ rec. $\FS$ property
and\ $Card(Stab(Q_e)) \leq c\}$, and $\{e: Q_e$ has the $a.a.$ rec. $\FS$
property and $Card(Stab(Q_e)) = c\}$ are all $\Sigma^0_3$-complete.

\item[(d)] $\{e: Q_e$ has the $a.a.$ rec. $\FS$ property and $Stab(Q_e)$ is
infinite$\}$ is $D^0_3$-complete and $\{e: Q_e$ has the $a.a.$ rec. $\FS$
property and $Stab(Q_e)$ is finite$\}$ is $\Sigma^0_3$-complete.

\item[(e)] $\{e: Q_e\ \text{ has the a.a. rec. $\FS$ property and $Stab(Q_e)$
is uncountable}\}$ is $\Sigma^1_1$-complete and $\{e: Q_e\ \text{ has the
a.a. rec. $\FS$ property and $Stab(Q_e)$}$ is countable$\}$ and $\{e: Q_e\
\text{ has the $a.a.$ rec. $\FS$ property and $Stab(Q_e)$}$ is countably
infinite$\}$ are $\Pi^1_1$-complete.

\item[(f)] $\{e: Q_e$ has the $a.a.$ rec. $\FS$ property and $Stab(Q_e)$ is
recursively nonempty$\}$ is $\Sigma^0_3$-complete, $\{e: Q_e$ has the $a.a.$
rec. $\FS$ property and $Stab(Q_e)$ is recursively empty$\}$ is
$D^0_3$-complete, and $\{e: Q_e$ has the $a.a.$ rec. $\FS$ property and 
$Stab(Q_e)$ is nonempty and recursively empty$\}$ is $D^0_3$-complete.

\item[(g)] $\{e: Q_e$ has the $a.a.$ rec. FPS and $Stab(Q_e)$ has
recursive cardinality $ > c\}$ is $\Sigma^0_3$-complete, $\{e: Q_e$ has the
$a.a.$ rec. $\FS$ property and $Stab(Q_e)$ has recursive cardinality $\leq c\}$
is $D^0_3$-complete, and $\{e: Q_e$ has the $a.a.$ rec. $\FS$ property and
$Stab(Q_e)$ has  cardinality $=c \}$ is $D^0_3$-complete. 

\item[(h)] $\{Q_e:$ has the $a.a.$ rec. $\FS$ property and $Stab(Q_e)$ has
$\{e:$ has the $a.a.$ rec. $\FS$ property and $Stab(Q_e)$ has infinite
recursive cardinality$\}$ is $\Pi^0_4$-complete.

\item[(i)] $\{e: Q_e$ has the $a.a.$ rec. $\FS$ property and $Stab(Q_e)$  is
perfect$\}$ are $D^0_3$-complete.
\end{compactenum}
\end{theorem}
\begin{proof}
Let $f$ be the recursive function such that $T_{Q_e}=T_{f(e)}$ where $T_{Q_e}$
is as constructed in the proof of Theorem \ref{prog2trees}.  Then $f$ shows
that $\{e: Q_e$ has the $a.a.$ rec. $\FS$ property and $Stab(Q_e)$ is
nonempty$\}$ is many-one reducible to $\{e: [T_e]$ is nearly $r.b.$ and is
nonempty$\}$ which is $\Sigma^0_3$. Thus $\{e: Q_e\ \text{has the $a.a.$ rec. $\FS$
property and $Stab(Q_e)$ is nonempty}\}$ is $\Sigma^0_3$.  In this way, we can establish
the upper bound on the complexity of the index set for any  property of finite
normal predicate $a.a.$ FSP logic programs $Q_e$ where the property is
restricted to cases where $Stab(Q_e) \neq \emptyset$ from the corresponding
complexity of the corresponding property for nearly recursively bounded trees. 

For the other upper bounds, first, it is easy to see that $A=\{e:Q_e$ has the
$a.a.$ rec. $\FS$ property$\}$ is $\Sigma^0_3$ by simply writing out the
definition.  To see that $B = \{e:Q_e$ has the $a.a.$ rec. $\FS$ property and
$Stab(Q_e)$ is empty$\}$ is $\Sigma^0_3$, note that $e \in B$ if and only if $e
\in A$ and either (i) $Q_e$ has an initial blocking set or (ii) $Q_e$ does not
have an initial blocking set and $T_{Q_e}$ as constructed in the proof of
Theorem \ref{prog2trees} is nearly recursively bounded and $[T_{Q_e}]=
\emptyset$.  Since the predicate `$Q_e$ has an initial blocking set' is
$\Sigma^0_2$ and the predicate `$T_e$ is nearly recursively bounded and
$[T_e]= \emptyset$'  is a $\Sigma^0_3$ predicate, it follows that $B$ is
$\Sigma^0_3$.  To see that $C = \{e:Q_e$ has the $a.a.$ rec. $\FS$ property
and $Card(Stab(Q_e)) \leq c\}$ is $\Sigma^0_3$ for any $c \geq 1$, we can use 
the program $R_e$ constructed in the proof of Theorem \ref{thm:irc}.  That is,
$e \in C$ if and only if $R_e$ has the $a.a.$ rec. $\FS$ property and $Card(S_e)
\leq c+1$. Now by Theorem \ref{prog2trees}, $R_e$ has the $a.a.$ $\FS$ property
and $Card(S_e) \leq c+1$ if and only if $T_{R_e}$ is nearly recursively bounded
and $Card([T_{R_e}]) \leq c+1$.  But $\{e: T$ is nearly $r.b.$ and
$Card([T_{R_e}]) \leq c+1\}$ is $\Sigma^0_3$ so that $C$ is $\Sigma^0_3$. A
similar proof will show that $D = \{e:Q_e$ has the $a.a.$ rec. $\FS$ property
and is finite$\}$ is $\Sigma^0_3$ and $E=\{e:Q_e$ has the $a.a.$ rec. $\FS$
property and is countable$\}$ is $\Sigma^1_1$. 

Finally, for the upper bounds on the complexity for the index sets in parts
(g), (h), and (i), we can use the program $S_e$ constructed from $Q_e$ as in 
the proof of the Theorem \ref{thm:rbce}.  That is, it is easy to see that $Q_e$
has the $a.a.$ rec. $\FS$ property if and only if $S_e$ has the $a.a.$ rec.
$\FS$ property and that the cardinality of the set of recursive stable models of
$Q_e$ equals the cardinality of the  set of recursive stable models of $S_e$.
Moreover, $Stab(Q_e)$ is perfect if and only if $Stab(S_e)$ is perfect.  But $S_e$ has the
$a.a.$ $\FS$ property if and only if the tree $T_{S_e}$ as constructed in the
proof of Theorem \ref{prog2trees} is nearly recursively bounded. Let $g$ be
the recursive function such that $T_{g(e)} = T_{S_e}$.  Then the question
whether $e$ lies in the desired index set in parts (g), (h), and (i),  
can be reduced to the problem of whether $g(e)$ lies in the corresponding index
set for nearly recursively bounded trees. Thus the upper bounds the complexity 
of these index sets  follow from the complexity of the corresponding index 
sets for nearly recursively bounded trees in Section \ref{classes}. 

The completeness for each of the index sets in our theorem can be proved as
follows. Given a finite normal predicate logic program $Q_e$, we construct 
a finite normal predicate logic program $Y_e$ as follows.  Let $L_e$ denote the
underlying language of of $Q_e$ and $L_e^*$ be the language which contains 0,
$s$, and a predicate $R^*(z,x_1, \ldots, x_n)$ for every predicate $R(x_1,
\ldots, x_n)$ and a predicate $A^*(x)$ for every propositional atom $A$ in $L$
where none of $R^*$, and $A^*$ occur in $L_e$.  To ease notation, we shall let
$\bar{0} =0$ and $\bar{n} =s^n(0)$ for each $n \geq 1$.  Then by Proposition
\ref{aux}, there is a finite normal predicate logic Horn program $Q^{-}$ with a
recursive least model $M^{-}$ whose language contains the constant symbol $0$
as well as all the constant symbols of $L_e$ and the function symbol $s$ and
all the function symbols from $L_e$ and whose set of predicate symbols are
disjoint from the language $L_e^*$ which includes the predicates $\num(\cdot)$,
$\noteq(\cdot,\cdot)$, and $term(\cdot)$ such that for any ground terms $t_1$
and $t_2$: 
\begin{compactenum}
\item $\num(t_1)$ holds in  $M^{-}$ if and only if $t_1 = \bar{n}$ for some 
$n \geq 0$,
\item $\noteq(t_1,t_2)$ holds in  $M^{-}$ if and only if there exist 
natural numbers $n$ and $m$ such that $n \neq m$ and $t_1 = \bar{n}$ and 
$t_2 =\bar{m}$, and 
\item $\term(t_1)$ holds in model $M^{-}$ if and only if $t_1$ is a 
ground term in $L_e$. 
\end{compactenum}
Moreover, we can assume that $Q^-$ has the rec. $\FS$  property.  Then let $Y_e$
be the program $Q^{-}$ plus all clauses $C^*(x)$ that arise from clauses $C \in
Q_e$ by adding the predicate $\num(x)$ to the body where $x$ the first variable
of the language that does not occur in $C$, adding the predicate $\term(t)$ to
the body for each term that occurs in $C$, and by replacing each predicate
$R(t_1, \ldots, t_n)$ that occurs in $C$ by $R^*(x,t_1, \ldots, t_n)$ and each
propositional atom $A$ that occurs in $C$ by $A^*(x)$.  The idea is that as $x$
varies over $\{\bar{n}: n \geq 0\}$, these clauses will produce infinitely many
copies of the program $Q_e$. More precisely, we let $Q_e^{n}$ denote the set of
all clauses of the form $C^*(\bar{n})$. $Q_e^{n}$ is essentially an exact copy
of $Q_e$ except that we have extended all predicates and propositional atoms to
have an extra term corresponding to $\bar{n}$ and each clause contains the
predicate $num(\bar{n})$ and $\term(t)$ in the body for each term in the
original  clause.  Since none of the clauses $C^*(x)$ have any predicates from
$Q^{-}$ in their heads, it will be the case that in every stable model $M$ of
$Y_e$, $M$ restricted to the ground atoms of $Q^{-}$ will just be $M^{-}$.
Thus, in particular, 
\begin{compactenum}
\item $\num(t_1)$ holds in $M$ if and only if $t_1 = \bar{n}$ for some 
$n \geq 0$,
\item $\noteq(t_1,t_2)$ holds in $M$ if and only if there exist 
natural numbers $n$ and $m$ such that $n \neq m$ and $t_1 = \bar{n}$ and 
$t_2 =\bar{m}$, and 
\item $\term(t_1)$ holds in $M$ if and only if $t_1$ is a 
ground term in $L_e$. 
\end{compactenum}

Now, if $\PS$ is any $ground(Q_e)$-proof scheme, then we let $\PS^{n}$ be the
result of  adding $\num(\bar{n})$ to each clause in $\PS$ and $\term(t)$ to
each clause if $t$ occurs in $\PS$ and replacing each predicate $R(t_1, \ldots,
t_n)$ that occurs in $\PS$ by $R^*(\bar{n},t_1, \ldots, t_n)$ and each
propositional atom $A$ that occurs in $\PS$ by $A^*(\bar{n})$.  It is easy to
see that the all minimal $ground(Y_e)$-proof schemes that derive atoms outside
of $ground(Q^{-})$ must consist of an interweaving of the pairs from minimal  
$ground(Q^{-})$-proof schemes of $\num(\bar{n})$  and $\term(t)$ for each term
$t$ in $L_e$ that occurs in the proof scheme of the form $\PS^{\bar{n}}$ with
the pairs for some $ground(Q_e)$-proof scheme $\PS^{\bar{n}}$.  It follows that
if  $Q_e$ has the rec. $\FS$ ($\FS$) property, then  $Y_e$ has the rec. $\FS$
($\FS$) property. However, if $Q_e$ does not have the rec. $\FS$ property, then
it cannot be that $Y_e$ has the $a.a$ rec.  $\FS$ property since if we could
effectively find  all the inclusion-minimal supports of minimal $Y_e$-proof
schemes for all but finitely many atoms, then there would be some $n$ in which
we could find all the inclusion-minimal supports of minimal $Q^{\bar{n}}$-proof schemes 
for any atom which contains $\bar{n}$, which
would allow us to effectively find all the inclusion-minimal supports of
minimal $Q_e$-proof schemes for any ground atom of $L$. Similarly, if $Q_e$
does not have the $\FS$ property, then the $Y_e$ does not have the  $a.a.$
$\FS$ property. Thus $Q_e$ has the rec. $\FS$ ($\FS$) property if and only if
$Y_e$ has the $a.a.$ rec. $\FS$ ($\FS$) property.  

Next we want to add a finite number of predicate clauses to $Y_e$ to produce a
finite normal predicate logic program $Z_e$ which restricts the stable models
to be essentially the same relative to the atoms of $ground(Q^{\bar{n}})$ for
all $n \geq 0$. To this end, we let $a$ be an atom that does not appear in
$Y_e$ and for each predicate $R^*(z,x_1, \ldots, x_n)$ of $Y_e$, we add a clause
\begin{multline*}
C_{R^*} = a \leftarrow R^*(y,x_1, \ldots, x_n),
\neg R^*(z,x_1, \ldots, x_n), \noteq(y,z), \\
\term(x_1), \ldots, \term(x_n), \neg a 
\end{multline*}
and for each propositional atom $A$ of $L_e$, we add a clause 
\[
C_{A^*} = a \leftarrow A^*(y),\neg A^*(z), 
\noteq(x,y), \neg a.
\]
First, we observe that $a$ cannot belong to any stable model of $M$ of $Z_e$.
That is, if $a \in M$, that none of the clauses $C_{R^*}$ and $C_{A^*}$ will
contribute anything to $ground(Z_e)_M$. Thus no clauses with $a$ in the head
will be $ground(Z_e)_M$ so that $a$ will not be in the least model of $M$ and
$M \neq ground(Z_e)_M$.

Now suppose that $M$ is a stable model of $Z_e$ and $a \notin M$.  Then it is
easy to see from the form of our rules that for any predicate $R(x_1, \ldots,
x_n)$ of $L_e$, $M$ can only contain ground atoms of the form $R^*(t_0,t_1,
\ldots,t_n)$ where $t_0 = \bar{n}$ for some $n \geq 0$ and $t_1, \ldots, t_n$
are ground terms of $L_e$. Similarly, for each propositional atom $A$ of $L_e$
and ground term $t$, $A(t)$ in $M$ implies $t = \bar{n}$ for some $n \geq 0$.
We claim that for any predicate $R(x_1, \ldots, x_n)$ and any ground terms
$t_1, \ldots, t_n$ in $L_e$, either $D_{R,t_1, \ldots, t_n} =
\{R^*(\bar{n},t_1, \ldots, t_n): n \geq 0\}$ is contained in $M$ or is entirely
disjoint from $M$. That is, if there is an $n \neq m$ such that
$R^*(\bar{n},t_1, \ldots, t_n) \in M$ but $R^*(\bar{m},t_1, \ldots, t_n) \notin
M$, then the clause $C_{R^*}$ will contribute the clause 
\[
\bar{C}_{R^*} = a \leftarrow R^*(\bar{n},t_1, \ldots, t_n),
\noteq(\bar{n},\bar{m})
\]
to  $ground(Z_e)_M$ so that $a$ would be in $M$ since $M$ is a model of
$ground(Z_e)_M$ and, hence, $M$ is not a stable model of $Z_e$.  Similarly, for
each propositional atom $A$ in $L_e$ either $D_{A} = \{A^*(\bar{n}): n \geq
0\}$ is contained in $M$ or is entirely disjoint from $M$. That is, if there is
an $n \neq m$ such that $A^*(\bar{n}) \in M$ but $A^*(\bar{m}) \notin M$, then
the clause $C_{A^*}$ will contribute the clause 
\[
\bar{C}_{A^*} = a \leftarrow A^*(\bar{n}),\noteq(\bar{n},\bar{m})
\]
to  $ground(Z_e)_M$ so that $a$ would  be in $M$ and $M$ is not a stable model
of $Z_e$. It follows that the stable models of $Z_e$ are in one-to-one
correspondence with the stable models of $Q_e$. That is, if $U$ is a stable
model of $Q_e$, then there is a stable model $V(U)$ of $Z_e$ such that
\begin{compactenum}
\item $M^{-} \subseteq V(U)$; 
\item for all predicate symbols $R(x_1, \ldots, x_n)$ in $L_e$, and ground
terms $t,t_1, \ldots, t_n$ in $L_e^*$, $R^*(t,t_1, \ldots, t_n) \in V(U)$ if
and only  if $t = \bar{m}$ for some $m \geq 0$, $t_1,\ldots,t_n \in L_e$, and 
$R(t_1, \ldots, t_n) \in U$; and 
\item for all  propositional atoms $A$ in $L_e$ and ground terms $t$ in
$L_e^*$, $A^*(t) \in V(U)$ if and only $t = \bar{m}$ for some $m \geq 0$ and $A
\in U$. 
\end{compactenum}
In addition, it is easy to prove by induction on the length of proof schemes
that every stable model of $V$ of $Z_e$ is of the form $V(U)$ where 
\begin{compactenum}
\item for all predicate symbols $R(x_1, \ldots, x_n)$ and ground terms $t_1,
\ldots, t_n$ in $L_e$, $R(t_1, \ldots, t_n) \in U$ if and only if
$R(\bar{0},t_1 \ldots, t_n) \in V$; and 
\item for all propositional atoms $A$ in $L_e$, $A \in U$ if and only if 
$A^*(\bar{0}) \in V$. 
\end{compactenum}
It follows that there is an effective one-to-one degree preserving
correspondence between the $Stab(Q_e)$ and $Stab(Z_e)$. Now let $\ell$ be a
recursive function such that $Q_{\ell(e)} = Z_e$.  We observe that our theorem
states that the complexity of every property of finite normal predicate logic
programs which have $a.a.$ rec. $\FS$ property is the same as the corresponding
complexity of the same property of finite normal predicate logic programs with
just the rec. $\FS$ property.  For example, in Section \ref{proofs}, we proved
that for every positive integer $c$, $X= \{e: Q_e$ has the rec. $\FS$ property
and $Card(Stab(Q_e)) = c\}$ is $\Sigma^0_3$-complete while we want to prove 
that $Y=\{e: Q_e$ has the $a.a.$ rec. $\FS$ property and $Card(Stab(Q_e)) =
c\}$ is $\Sigma^0_3$-complete. Now $\ell$ shows that $X$ is many-one reducible
to $Y$ so that, since we have already shown that $Y$ is $\Sigma^0_3$, it must
be the case that $Y$  is $\Sigma^0_3$-complete.  All the other completeness
results follows from the corresponding completeness results in the same manner. 
\end{proof}

\begin{theorem}
\begin{compactenum}
\item[a.] $\{e: Q_e \ \text{has the $a.a.$ $\FS$ property}\}$ is
$\Sigma^0_4$-complete. 
\item[b.] $\{e: Q_e$ has the $a.a.$ $\FS$ property and $Stab(Q_e)$ is empty$\}$
and $\{e: Q_e$ has the $a.a.$ $\FS$ property and $Stab(Q_e)$ is nonempty$\}$
are $\Sigma^0_4$-complete.  
\item[c.] For any positive integer $c$, $\{e: Q_e$ has the $a.a.$ $\FS$
property and $\mathit{Card}(Stab$ $(Q_e)) > c\}$, $\{e: Q_e$ has the $a.a.$
$\FS$ property and $\mathit{Card}(Stab(Q_e)) \leq c\}$, and $\{e: Q_e$ has the
$a.a.$ $\FS$ property and $\mathit{Card}(Stab(Q_e)) =c\}$ are
$\Sigma^0_4$-complete. 
\item[d.] $\{e: Q_e$ has the $a.a.$ $\FS$ property and $Stab(Q_e)$ is
finite$\}$ and $\{e: Q_e$ has the $a.a.$ $\FS$ property and $Stab(Q_e)$ is
infinite$\}$ are $\Sigma^0_4$-complete.  
\item[e.] $\{e: Q_e$ has the $a.a.$ $\FS$ property and $Stab(Q_e)$ is
countable$\}$ and  $\{e: Q_e$ has the $a.a.$ $\FS$ property and $Stab(Q_e)$ is
countably infinite$\}$ are $\Pi^1_1$-complete and $\{e: Q_e$ has the $a.a.$
$\FS$ property and $Stab(Q_e)$ is uncountable$\}$ are $\Sigma^1_1$-complete.  
\item[f.] $\{e: Q_e$ has the $a.a.$ $\FS$ property and $Stab(Q_e)$ is
recursively empty$\}$, $\{e: Q_e$ has the $a.a.$ $\FS$ property and $Stab(Q_e)$
recursively nonempty$\}$, and $\{e: Q_e$ has the $a.a.$ $\FS$ property and
$Stab(Q_e)$ is nonempty and recursively empty$\}$ are $\Sigma^0_4$-complete. 
\item[g.] For every positive integer $c$, $\{e: Q_e$ has the $a.a.$ $\FS$
property and $Stab(Q_e)$ has recursive cardinality c$\}$, $\{e: Q_e$ has the
$a.a.$ $\FS$ property and $Stab(Q_e)$ has recursive cardinality  $\leq c\}$,
and $\{e: Q_e$ has the $a.a.$ $\FS$ property and $Stab(Q_e)$ has recursive
cardinality $= c\}$ are $\Sigma^0_4$-complete.
\end{compactenum}
\end{theorem}
\begin{proof}
To establish the upper bounds for each of the index sets described in the
theorem, we can use the same strategy as we did in Theorem \ref{thm:aarFSP}.
That is, by Theorem \ref{prog2trees}, $Q_e$ has the $a.a.$ $\FS$ property  and
has a stable model if and only if  $T_{Q_e}$ is nearly bounded and   $[T_{Q_e}]
\neq \emptyset$.  Let $f$ be the recursive function such that
$T_{Q_e}=T_{f(e)}$.  Then $f$ shows that 
\[
A= \{e: Q_e \ \mbox{has the $a.a.$ $\FS$ property and $\mathit{Stab}(Q_e)$
is nonempty}\}
\]
 is many-one reducible to 
\[
B= \{h:T_h \ \mbox{is nearly bounded and $[T_h]$ is nonempty}\} 
\]
which is $\Sigma^0_4$. Thus $A$ is $\Sigma^0_4$.  In this way, we can establish
the upper bounded on the complexity of the index set for any  property of
finite normal predicate logic programs $Q_e$ which have the $a.a.$ rec. $\FS$
property where the property is restricted to cases such that $Stab(Q_e) \neq
\emptyset$ from the complexity of the corresponding property for nearly
recursively bounded trees. 

For the other upper bounds, first, it is easy to see that $\bar{A}=\{e:Q_e$ has
the $a.a.$ $\FS$ property$\}$ is $\Sigma^0_4$ by simply writing out the
definition.  To see that $\bar{B} = \{e:Q_e$ has the $a.a.$ rec. $\FS$ property
and $[T_{Q_e}]$ is empty$\}$ is $\Sigma^0_4$, note that $e \in \bar{B}$ if and
only if $e \in \bar{A}$ and either (i) $Q_e$ has an initial blocking set or
(ii) $Q_e$ does not have an initial blocking set and $T_{Q_e}$ as constructed
in Theorem \ref{prog2trees} is nearly bounded and $[T_{Q_e}]= \emptyset$.
Since the predicate `$Q_e$ has an initial blocking set' is $\Sigma^0_2$ and
the predicate `$T_e$ is nearly bounded and $[T_e]= \emptyset$' is a
$\Sigma^0_4$ predicate, it follows that $\bar{B}$ is $\Sigma^0_4$.  To see that 
$\bar{C} = \{e:Q_e$ has the $a.a.$ $\FS$ property and $Card(Stab(Q_e)) \leq c\}$
is $\Sigma^0_4$ for any $c \geq 1$, we can use the program $R_e$ constructed in
the proof of Theorem \ref{thm:irc}.  That is, $e \in \bar{C}$ if and only if
$R_e$ has the $a.a.$ $\FS$ property and $Card(\mathit{Stab}(R_e)) \leq c+1$.
Now by Theorem \ref{prog2trees}, $R_e$ has the $a.a.$ $\FS$ property and 
$\mathit{Card}(\mathit{Stab}(R_e)) \leq c+1$ if and only if $T_{R_e}$ is nearly
bounded and $Card([T_{R_e}]) \leq c+1$.  But $\{e: T$ is nearly bounded and
$Card([T_{R_e}]) \leq c+1\}$ is $\Sigma^0_4$ so that $\bar{C}$ is $\Sigma^0_4$.
A similar proof will show that $\bar{D} = \{e:Q_e$ has the $a.a.$ $\FS$
property and is finite$\}$ is $\Sigma^0_4$ and $\bar{E}=\{e:Q_e$ has the $a.a.$
$\FS$ property and is countable$\}$ is $\Sigma^1_1$. 

Finally, for the upper bounds on the complexity for the index sets in parts (f)
and (g), we can use the program $S_e$ constructed from $Q_e$ in the proof of
Theorem \ref{thm:rbce}.  That is, it is easy to see that $Q_e$ has the $a.a.$
$\FS$ property if and only if $S_e$ has the $a.a.$ $\FS$ property and that the 
cardinality of the set of recursive stable models of $Q_e$ equals the
cardinality of the set of recursive stable models of $S_e$. Moreover, the set
of stable models of $Q_e$ is perfect if and only if the set of stable models of
$S_e$ is perfect.  But $S_e$ has the $a.a.$ $\FS$ property if and only if  the
tree $T_{S_e}$ as constructed in Theorem \ref{prog2trees} is nearly recursively
bounded. Let $g$ be the recursive function such that $T_{g(e)} = T_{S_e}$.
Then the question whether $e$ lies in the desired index set in parts (f), (g),
and (h) can be reduced to the problem of whether $g(e)$ lies in the
corresponding index set for nearly bounded trees. Thus the upper bounds for the
complexity  of these index sets follow from the complexity of the corresponding
index sets for nearly bounded trees in Section \ref{classes}. 

For the completeness results in part (e) of the theorem, we can follow the same
strategy as in the proof of Theorem \ref{thm:aarFSP}.  By Theorem
\ref{thm:icu}, we know that  $X= \{e: Q_e$ has the $\FS$ property and
$Stab(Q_e)$ is uncountable$\}$ is $\Pi^1_1$-complete while we want to prove
that $Y=\{e: Q_e$ has the $a.a.$ $\FS$ property and $Card(Stab(Q_e)$ is
uncountable$\}$ is  $\Pi^1_1$-complete. Now the recursive function $\ell$ such 
that $Z_e = Q_{\ell(e)}$ constructed in the proof of Theorem \ref{thm:aarFSP} 
shows that $X$ is many-one reducible to $Y$ so that $Y$ is  $\Pi^1_1$-complete.
All the other completeness results in part (e)  of our theorem follow from the
corresponding completeness results in Theorem \ref{thm:icu} in the same manner. 

Unfortunately, we cannot  follow that same strategy as in Theorem
\ref{thm:aarFSP} in the remaining parts of theorem because the completeness
results for finite normal predicate logic programs with the $\FS$ property do
not match the completeness results for finite normal predicate logic programs
with $a.a.$ $\FS$ property. Instead we shall outline the modifications that are
needed to prove an analogue of Theorem \ref{tree2prog} that can be used to
prove the completeness result for finite normal predicate logic programs which
have the $a.a.$ $\FS$ property from the corresponding completeness results for 
nearly bounded trees. 

First, let us recall the construction of the trees that we  used to prove part
(d) of Theorem \ref{thm:pirb}.  We defined a primitive recursive function
$\phi(e,m,s) = (least \ n > m)(n \notin W_{e,s} \setminus \{0\})$. For any
given $e$, we let $V_e$ be the tree such that $(m), (m,0),(m,1) \in U_e$ for
all $m \geq 0$ and $(m,s+2) \in V_e$ if and only if $m$ is the least element
such that $\phi(e,m,s+1) > \phi(e,m,s)$.  This is only a slight modification of
the tree $U_e$ defined in that the proof of part (d) of Theorem \ref{thm:pirb}
in that we have ensured that $(m,0),(m,1) \in V_e$ are always in $U_e$ and so
that we are forced to shift the remaining nodes to right by one.  It will still
be that case that if $W_e \setminus \{0\}$ is cofinite, then there is exactly
one node in $V_e$ that has  infinitely many successors and $V_e$ is  bounded
otherwise.  Clearly there is a recursive function $f$ such that $T_{f(e)} =
V_e$. But then 
\[
e \in \omega \setminus \mathit{Cof} \iff T_{f(e)} \ \text{is bounded}.
\]
where $\mathit{Cof} = \{e: \omega \setminus W_e$ is finite$\}$. 

Next let $S$ be an arbitrary complete $\Sigma^0_4$ set and suppose that
$a \in S \iff (\exists k) (R(a,k))$ where $R$ is $\Pi^0_3$. By the usual
quantifier methods, we may assume that $R(a,k)$ implies that $R(a,j)$ for all
$j>k$.  By the $\Pi^0_3$-completeness of the set $\{e: T_e \ \text{is
bounded}\}$, there is a recursive function $h$ such that $R(a,k)$ holds if and
only if $V_{h(a,k)}$ is bounded and such that $V_{h(a,k)}$ is $a.a.$ bounded
for every $a$ and $k$.  Now we can define a recursive function $\psi$ so that
\[
T_{\psi(a,e)} = \{(0)\} \cup \{(k+1)^\smallfrown \sigma: \sigma \in
V_{h(a,k)}\} \cup \{0^\smallfrown \sigma: \sigma \in T_e\}.
\]
Thus we have two parts of the tree $T_{\psi(a,e)}$.  That is,  above the node
(0), we have a copy of $T_e$ and we shall call this part of the tree
$First0(T_{\psi(a,e)})$.  We shall refer to the remaining part of
$T_{\psi(a,e)}$ as $NotFirst0(T_{\psi(a,e)})$. Now if $a \in S$, then
$V_{h(a,k)}$ is bounded for all but finitely many $k$ and is nearly bounded for
the remainder.  Thus $NotFirst0(T_{\psi(a,e)})$  is nearly bounded. If $a
\notin S$, then, for every $k$, $V_{h(a,k)}$ is not bounded, so that
$NotFirst0(T_{\psi(a,e)})$  is not nearly bounded.  Thus $a \in S$ if and only
if $NotFirst0(T_{\psi(a,e)})$ is nearly bounded.  Hence if $T_e$ is $r.b.$ or
bounded, then $a \in S$ if and only if $T_{\psi(a,e)}$  is nearly bounded.

Next we describe a finite normal predicate logic program $Q_{a,e}$ such 
that there is a one-to-one effective correspondence between $Stab(Q_{a,e})$ 
and $[T_{\psi(a,e)}]$.  Our construction will just be a slight modification of 
the construction in Theorem \ref{tree2prog}. First we shall need some
additional predicates on sequences. That is, we let the predicate
$\firstO(c(\sigma))$ be true if and only if $\sigma$ is a sequence which starts
with 0 and the predicate  $\notfirstO(c(\sigma))$ be true if and only if
$\sigma$ is a nonempty sequence which does not starts with 0.  We let the
predicate $\thirdO(c(\sigma))$ be true if and only if $\sigma$ is a sequence of
length $\geq 3$ whose third element is 0 and we let the predicate
$\notthirdO(c(\sigma))$ be true if and only if $\sigma$ is a sequence of length
$\geq  3$ whose third element is not 0.  We shall also require a predicate
$\lengthonetwo( \cdot)$ which holds only on  codes  of sequences of length 1 or
2 and $\length3( \cdot)$ which holds only on codes of  sequences of length 3.
Finally, we shall need a predicate $\Agreeonetwo(\cdot, \cdot) $ which holds
only on pairs of codes $(c(\sigma),c(\tau))$ where $\sigma$ and $\tau$ are 
of length 3 and $\sigma$ and $\tau$ agree on there first two entries. 

As in the proof of Theorem \ref{tree2prog}, there exists the following three 
finite normal predicate logic programs such that the set of ground terms in
their underlying  language are all of the form $s^n(0)$ where $0$ is a constant
symbol and $s$ is a unary function symbol.  We shall use $n$ has an
abbreviation for the term $s^n(0)$. 
\begin{compactenum}
\item[(I)]  A finite predicate logic Horn program $P_0$ such that for a
predicate $\tree(\cdot )$ of the language of $P_0$, the atom $\tree(n)$ belongs
to the least Herbrand model of $P_0$ if and only if $n$ is a code for a finite
sequence $\sigma$ and $\sigma\in T_{\psi(a,e)}$.  
\item[(II)]  A finite predicate logic Horn program $P_1$ such that for a
predicate $seq(\cdot)$ of the language of $P_1$, the atom $seq(n)$ belongs to
the least Herbrand model of $P_1$ if and  only if $n$ is the code of a finite
sequence $\alpha \in \omega^{< \omega}$. 
\item[(III)]  A finite predicate logic Horn program $P_2$ which correctly
computes the following recursive predicates on codes of sequences.
\begin{compactdesc}
\item[(a)] $\samelength (\cdot ,\cdot )$. This succeeds if and only if both
arguments are the codes of sequences of  the same length.
\item[(b)] $\diff (\cdot ,\cdot )$. This succeeds if and only if the arguments
are codes of sequences which are different.
\item[(c)] $\shorter (\cdot ,\cdot )$. This succeeds if and only both arguments
are codes of sequences and the first sequence is shorter than the second
sequence.
\item[(d)] $\length (\cdot ,\cdot )$. This succeeds when the first argument is
a code of a sequence and the second argument is the length of that sequence.
\item[(e)] $\notincluded (\cdot ,\cdot )$. This succeeds if and only if both
arguments are codes of sequences and the first sequence is not the initial
segment of the second sequence.
\item[(f)] $\firstO(\cdot)$. This succeeds if and only if the argument is the
code of a sequence which starts with 0. 
\item[(g)] $\notfirstO(\cdot)$. This succeeds if and only if the argument is
the code of a nonempty sequence which  does not start with 0. 
\item[(h)] $\thirdO(\cdot)$. This succeeds if and only if the argument is the
code of a sequence of length $\geq 3$ whose third element is 0. 
\item[(i)] $\notthirdO(\cdot)$. This succeeds if and only if the argument 
is the code of a sequence of length $\geq 3$ whose third element is not 0. 
\item[(j)] $\Agreeonetwo (\cdot, \cdot )$. This succeeds if and only if the 
arguments are codes of a sequences of length $3$ which agree on the first two
elements. 
\item[(k)] $\lengthonetwo(\cdot )$. This succeeds if and only if the argument
is  a code  of a sequence  of length 1 or 2. 
\item[(l)] $\lengththree(\cdot )$. This succeeds if and only if the argument
is  a code  of a sequence  of length 3.
\item[(m)] $\num(\cdot)$. This succeeds if and only if the argument is either
$0$ or $s^n(0)$ for some $n \geq 1$.
\item[(n)] $\mathit{greater0}(\cdot)$. This succeeds if and only if the
argument is $s^n(0)$ for some $n \geq 1$.  
\end{compactdesc}
\end{compactenum}
Now let $P^{-}$ be the finite normal predicate  logic program which is the
union of programs $P_0\cup P_1\cup P_2$.  We denote its language by ${\cal
L}^{-}$ and we let $M^{-}$ be the least model of $P^{-}$.  By Proposition
\ref{aux}, we can assume that this program $P^{-}$ is a Horn program and for
each ground atom $b$ in the Herbrand base of $P^{-}$, we can explicitly
construct the set of all $P^{-}$-proof schemes of $b$. In particular, $\tree
(n) \in M^{-}$ if and only if $n$ is the code of node in $T_{\psi(a,e)}$. 

Our final program $P_T$ will consist of $P^{-}$ plus clauses (1)-(12) given
below. We assume that these additional clauses do not contain any of predicates
of the language ${\cal L}^{-}$ in the head. However, predicates from ${\cal
L}^{-}$ do appear in the bodies of clauses (1) to (12). Therefore, whatever
stable model of the extended program we consider, its trace on the set of
ground atoms of ${\cal L}^{-}$ will be $M^{-}$. In particular, the meaning of
the predicates of the language ${\cal L}^{-}$ listed above will always be the
same.

We are now ready to write the additional clauses which, together with the
program $P^{-}$, will form the desired program $Q_{a,e}$.  First of all, we
select three new unary predicates:
\begin{compactenum}
\item[(i)] $\ipath (\cdot )$, whose intended interpretation in any given stable
model $M$ of $Q_{a,e}$ is that it holds only on the set of codes of sequences
that lie on an infinite path  through $T_{\psi(a,e)}$ that starts with 0.  
This path  will correspond to the path encoded by the stable model of $M$,
\item[(ii)] $\notpath (\cdot )$, whose intended interpretation in any 
stable model $M$ of $Q_{a,e}$ is the set of all codes of sequences which are in
$T_{\psi(a,e)}$ but do not satisfy $\ipath( \cdot)$, and 
\item[(iii)] $\control (\cdot )$, which will be used to ensure that 
$\ipath( \cdot)$ always encodes an infinite path through $T_{\psi(a,e)}$.
\end{compactenum}
Next we include the same seven sets of clauses as we did in Theorem
\ref{tree2prog} to make sure that stable models $Q_{a,e}$ code paths through
the tree $T_e$ which sit above the node 0.  This requires that we modify those
clauses so that we restrict ourselves to the sequences that satisfy
$\firstO(X)$. 

This given, the first seven clauses  of our program are the following.\\
\ \\
(1) $\ipath (X) \longleftarrow \firstO(X),\tree(X),\ \neg \notpath (X)$\\
(2) $\notpath (X) \longleftarrow \firstO(X),\tree(X),\ \neg \ipath (X)$\\
(3) $\ipath (c(0))\longleftarrow $ \\
(4) $\notpath (X) \longleftarrow \firstO(X), \tree(X),\ \ipath (Y),$\\  
\mbox{}\ \hspace*{.2in} $\firstO(Y), \tree (Y), \samelength (X,Y), \diff (X,Y)$\\
(5) $\notpath (X) \longleftarrow \firstO(X),\tree(X),\  \firstO(Y), \tree(Y),\
\ipath (Y)$, \\ $\shorter (Y,X),$  
$ \notincluded (Y,X)$\\
(6) $\control (X)\longleftarrow \firstO(Y), \ipath (Y),\ \length (Y,X)$\\
(7) $\control (X) \longleftarrow \mathit{greater0}(X),\num(X),  \neg \control
(X)$\\ 
\ \\
Next we add the clauses involving an additional predicate $in(X)$ which is used
to ensure that the final program $Q_{a,e}$ has the {\it a.a.} $\FS$ property if
and only if the tree $T_{\psi(a,e)}$ is nearly bounded. \\
\ \\
(8) $\ipath (0) \longleftarrow$  \\
(9) $\notpath(X) \longleftarrow \notfirstO(X),\tree(X)$\\
(10) $in(X) \longleftarrow \notfirstO(X),\tree(X),\lengthonetwo(X)$\\
(11) $in(X) \longleftarrow \notfirstO(X),\tree(X),\length3(X),
\thirdO(X)$,\\  
\mbox{} \hspace*{.2in}
$\notfirstO(Y),\tree(Y),\length3(Y),\notthirdO(Y),\neg in(Y)$,\\
(12) $\control(0) \longleftarrow$\\
\ \\
Clearly, $Q_{a,e}  = P^{-} \cup \{ (1),\ldots ,(12)\}$ is a finite
predicate logic program. 

As in the proof of Theorem \ref{tree2prog}, we can establish 
establish a ``normal form'' for the stable models of $Q_{a,e}$. Each such
model must contain $M^{-}$, the least model of $P^{-}$. In fact,
the restriction of a stable model of $P_T$ to $H(P^{-})$ is
$M^{-}$. 
Given any $\beta =  
(0,\beta {(1)}, \beta {(2)}, \ldots ) \in \omega^{\omega}$, 
we let 
\begin{eqnarray*}\label{aaMbeta}
M_{\beta} = & M^{-} & \cup \{ \control (n) : n \in \omega\} 
\cup \{\ipath(0)\} \\
&& \cup  
\{ \ipath (c((0,\beta {(1)}, \ldots, \beta {(n)})) : n \geq 1 \}\\
&&  \cup
\{ \notpath (c(\sigma)) :\sigma \in T_{\psi(a,e)} \ \mbox{and} \ 
\sigma \not \prec \beta\}\\
&& \cup \{ in(c((m,n))): m > 0 \ \mbox{and} \ n \geq 0\} \\
&& \cup \{in(c((m,n,0))): m > 0 \ \mbox{and} \ n \geq 0\}.
\end{eqnarray*}
We claim that $M$ is a stable model of $Q_{a,e}$ if and only if 
$M = M_{\beta}$ for some $\beta \in [T_{\psi(a,e)}]$. 

First, let us consider the effect of the clauses (8)-(12).  Clearly, clause (8)
forces that $\ipath(0)$ must be in every stable model of $Q_{a,e}$  and the
clauses  in (9) force  that $\notpath(c(\sigma))$ is in every stable model of
$Q_{a,e}$  for all $\sigma \in T_{\psi(a,e)}$ which do not start with 0. Since
all the clauses (1)-(6) require $\firstO(c(\sigma))$ to be true, the only
minimal $Q_{a,e}$-proof schemes for  $\notpath(c(\sigma))$ for $\sigma \in
T_{\psi(a,e)}$ which do not start with 0 must use the Horn clause of type (9). 
Thus the minimal $Q_{a,e}$-proof schemes with conclusion $\notpath(c(\sigma))$
where $\sg$ does not start with 0 consist of the set of pairs of a  minimal
$P^{-}$-proof schemes of $\tree(c(\sg))$ followed by the tuple $\langle
c(\sg),(9)^*\rangle$ where $(9)^*$ is the ground instance of (9) where $X$ is
replaced by $c(\sg)$.  Thus support of such a proof-scheme is $\emptyset$.
Thus all the minimal $Q_{a,e}$- proof schemes of $\notpath(c(\sigma))$, where
$\sg$ does not start with 0, have empty support.  Similarly, $in(c(\sigma))$
can be derived only using clause (10) if $\sigma$ has length 1 or 2 so that all
minimal $Q_{a,e}$-proof schemes of $in(c(\sigma))$, where $\sg$ has length 1 or
2, have empty support.  Clause (12) is the only way to derive $\control(0)$ so
that the only minimal $Q_{a,e}$-proof scheme of $\control(0)$ uses clause (12)
and has empty support. 

The only way to derive $in(\sigma)$ for $\sigma$ of length 3 is via an instance
of clause (11).  Such clauses will allow us  to derive $in(c((m,n,0)))$ for any
$m >0$ and $n \geq 0$ with a proof scheme whose support is of the form
$\{in(c((m,n,p)))\}$ for some $p > 0$ where $(m,n,p) \in T_{\psi(a,e)}$. Since
we always put $(m,n,1) \in T_{\psi(a,e)}$, there is at least one such proof
scheme but there could be infinitely many of such proof schemes if  $(m,n,p)
\in T_{\psi(a,e)}$ for infinitely many $p >0$.  It then follows from our
definition of $T_{\psi(a,e)}$ that there will be finitely  many $m > 0$ and $n
\geq 0$ such that $in(c((m,n,0)))$ has infinitely many proof schemes if and
only if the tree $\mathit{NotFirst0}(T_{\psi(a,e)})$ is nearly bounded, which
occurs  if and only if $a \in S$.  Now, if $T_e$ is bounded, then we can use
the same argument that we used in Theorem \ref{tree2prog} to show that there
are only finitely many minimal $Q_{a,e}$-proofs schemes for the ground
instances of predicates in the heads of such clauses for $\sg \in T_{a,e}$ that 
start with 0.  It follows that if $T_e$ is bounded, then $a \in S$ if and only
if $Q_{a,e}$ has the $a.a.$ $\FS$ property. 

We can use the same arguments that we used in  Theorem \ref{tree2prog} to show
that the clauses (1)-(7) force that the only stable models of $Q_{a,e}$ are
$M_\beta$ where $\beta =  (0,\beta {(1)}, \beta {(2)}, \ldots ) \in
\omega^{\omega}$ and $(\beta {(1)}, \beta {(2)}, \ldots ) \in [T_e]$.  The only
difference is that the clause (12) allows us to derive $\control(0)$ directly. 
Thus if $T_e$ is bounded, then there will be an effective one-to-one degree
preserving correspondence between $Stab(Q_{a,e})$ and $[T_{\psi)a,e)}]$ and
$Q_{a,e}$ has the $a.a.$ $\FS$ property if and only if $a \in S$.  

The $\Sigma^0_4$-completeness results for the remaining parts of theorem can
all be proved by the following type  argument. Suppose, for example,  that we
want to prove that 
\begin{multline*}
A = \{e:Q_e \ \mbox{has the $a.a.$ $\FS$ property and
$Stab(Q_e)$}  \\
\text{is nonempty and recursively empty}\}
\end{multline*}
is $\Sigma^0_4$-complete. Then we know that there exists a recursively bounded
tree $T$ which is nonempty but which has no recursive paths (Jockusch and Soare
\cite{JS72a}.) Thus let us fix $e$ such that  $T_e$ is recursively bounded and
$[T_e]$ is nonempty and has no recursive elements.  Then for our $\Sigma^0_4$
predicate $S$, we have the property that $a \in S$ if and only if
$T_{\psi(a,e)}$ is nearly bounded and $[T_{\psi(a,e)}]$ is nonempty and has no
recursive elements.  But then $T_{\psi(a,e)}$ is nearly bounded and
$[T_{\psi(a,e)}]$ is nonempty and has no recursive elements if and only
$Q_{a,e}$ is $a.a.$ bounded and $Stab(Q_{a,e})$ is nonempty and has no
recursive elements. Now if $g$ is the recursive function such that $Q_{g(a)} =
Q_{a,e}$, then $a \in S$ if and only if $g(a) \in A$. Thus $A$ is complete for
$\Sigma^0_4$ sets. 
\end{proof}

\section{Conclusions}\label{concl}

In this paper, we have determined the complexity of various index sets
associated with properties of the set of stable models of finite normal logic
programs. In particular, we determined the complexity of the index sets
associated with various properties on the cardinality or recursive cardinality
of the set of stable models of a program relative to all finite normal
predicate logic programs as well as to all finite predicate logic programs that
have the $\FS$ (rec. $\FS$, $a.a$ $\FS$, $a.a.$ rec. $\FS$) property. The
results of this paper refine and extend earlier results on index sets for
finite predicate logic programs that appeared in \cite{MNR}.

In most cases, we showed that the problem of finding the complexity of such
index sets can be reduced to problem of finding the corresponding complexity of
an index set associated with the cardinality or recursive cardinality of the
set of infinite paths through primitive recursive trees, bounded primitive
recursive trees, and recursively bounded primitive recursive trees. However,
due to the fact that there is no analogue of the compactness theorem for the
stable model semantics of logic programs, there are a few cases where there is
is difference between the complexity of an index set associated with the
property of logic programs which have  no stable models and the corresponding
index set associated with the property of primitive recursive trees which have
no infinite paths. 

Nevertheless, we have shown that there is a close connection with the problem
of finding stable models of finite predicate logic programs and the problem of
finding infinite paths through primitive recursive trees. In fact, our original
definitions of the finite support property and recursive finite support
property were motivated by trying to find the analogue in logic programs of
bounded and recursively bounded trees.  Moreover, in this paper, we defined the
new concept of decidable logic programs based on finding an analogue of
decidable trees. Thus while the computation of the stable model semantics of
logic programs may, at the first glance, look different from the classical
Turing-machine based computations, our results show once more the unity of
underlying concepts and abstractions so beneficial to both Computer Science and
Computability Theory.

\section*{Acknowledgements}
During the work on this paper D. Cenzer was partially supported by NSF grant
DMS-65372. V.W. Marek was partially supported by
the following grants and contracts: Image-Net: Discriminatory Imaging and
Network Advancement for Missiles, Aviation, and Space, United States Army SMDC
contract, NASA-JPL Contract, Kentucky Science and Engineering Foundation grant,
and NSF ITR: Decision-Theoretic Planning with Constraints grant.
J.B. Remmel was partially supported by NSF grant DMS 0654060.

\end{document}